\documentclass[11pt,a4paper,oneside]{book}

\usepackage[T1]{fontenc}
\usepackage[utf8]{inputenc}

\usepackage{amsmath,amssymb,amsfonts,amsthm,mathtools}
\usepackage{tensor}
\usepackage{bm}
\usepackage{slashed}
\usepackage{physics}

\usepackage{graphicx}
\usepackage{subcaption}
\usepackage{float}
\usepackage{microtype}
\usepackage{hyperref}
\usepackage{xcolor}
\usepackage[bottom, flushmargin]{footmisc}
\usepackage{tikz-cd}
\newtheorem*{theorem*}{Theorem}
\newtheorem*{definition*}{Definition}

\usepackage[left=2.5cm,right=2.5cm,top=2.5cm,bottom=2.5cm]{geometry}

\usepackage{fancyhdr}
\begin{document}
\pagestyle{fancy}

\fancyhead{}\fancyfoot{}
\fancyfoot[C]{\thepage}

\frontmatter

\begin{titlepage}
\begin{center}
    \includegraphics[width=0.30\textwidth]{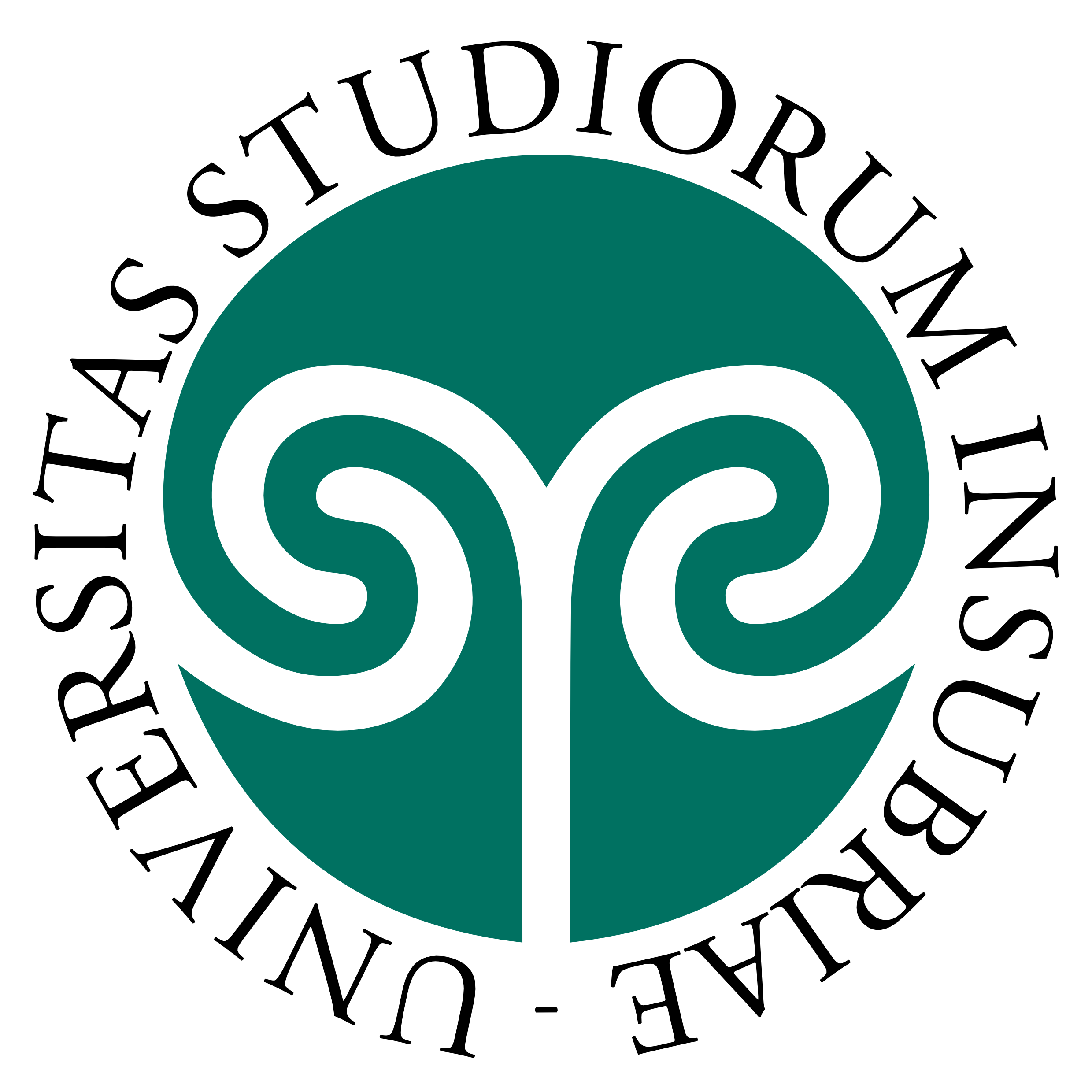}
    \par \vspace{1cm}

    {\LARGE\scshape Università degli Studi dell'Insubria\par}
    \vspace{2cm}

    {\scshape Dipartimento di Scienza e Alta Tecnologia - Como\par}
    \vspace{0.5cm}
    
    {\scshape Academic Year 2024-2025\par}
    \vspace{0.5cm}

    {\scshape \large Master's degree in Physics\par}

    \vspace{2 cm}

    {\Huge \bfseries \textsl{Classical Spinors on Curved Spacetime: applications to Cosmology and Astrophysics} \par}

\end{center} 

\vspace{2 cm}

\begin{flushleft}
    \large
    Supervisor: Prof. Oliver Fabio Piattella\\
    Co-supervisor: Prof. Sergio Luigi Cacciatori\par
\end{flushleft}

\vspace{1 cm}

\begin{flushright}
    \large Master's thesis by:\\
    Andrea La Delfa\\
    Matricola: 751668
\end{flushright}

\end{titlepage}

\newpage
\begin{flushright}
\textit{Al mio papà e alla mia mamma}
\end{flushright}

\vspace*{\fill}
\begin{flushleft}
    ``\textit{Solo ora capisco che il modo in cui si viene al mondo è irrilevante, 
    \newline
    è quello che fai del dono della vita che stabilisce chi sei.}'' 
\end{flushleft}

\tableofcontents\markboth{Contents}{Contents}

\fancyhead[L]{\nouppercase{\leftmark}}

\chapter{Introduction}

\markboth{Introduction}{}

In 1905, Albert Einstein, in the article \textit{Zur Elektrodynamik bewegter Körper} \cite{Einstein_Electr}, introduced to the scientific community of that time what has since been known as the \textit{Special Theory of Relativity}. He formulated that theory to address the shortcomings of Maxwell's Electromagnetism concerning phenomena described in reference frames in uniform relative motion (inertial reference frames).
\newline
\newline
Eleven years later, in the article \textit{Die Grundlage der allgemeinen Relativit$\ddot{a}$tstheorie} \cite{Einstein_GenRel}, Einstein enlarged this theory by formulating what was called by himself the \textit{General Theory of Relativity}. This theory enlarges Special Relativity by including non inertial reference frames and, as Einstein shows, it also provides a theory of gravitation, which, as we will explain later, comes quite naturally from the principles and the logical setup of this new theory. What is astonishing about General Relativity is that Einstein, with the help of his friend and mathematician Marcel Grossmann, derived it from scratch solely through logical reasoning, without any empirical or experimental proof.
\newline
\newline
It took a very short amount of time for this new theory of gravitation to be applied in various contexts: for example, the spherically symmetric solution of Karl Schwarzschild \cite{Schwarzschild} and the cosmological considerations by Einstein himself\footnote{This is the article where the cosmological constant $\Lambda$  was introduced for the first time.} \cite{Einstein_Cosmo}. 
In particular, GR provided Cosmology with a powerful tool that allowed for a comprehensive mathematical description of our Universe. Indeed, the present standard cosmological model, the $\Lambda$CDM model, is based, taking into account only the gravitational part, solely on the Theory of General Relativity.
\newline
\newline
Unfortunately, despite the giant leaps made in recent years and the significant discoveries made by our satellites, we are still far from having a complete overview of what the components of the Universe actually are and why they behave as we observe. Among the multiple questions, the most important ones are what Dark Energy (DE) is and what Dark Matter (DM) is. While the first one can be viewed as a geometrical contribution and we can suppose that Einstein's theory of gravitation must be modified on large scales to incorporate the effects of DE\footnote{For a complete account of all the problems concerning Dark Energy and the cosmological constant $\Lambda$, the article \cite{HistLambda} can be very useful.}, the second one cannot be viewed in this way, because we need DM to explain very different processes on very different scales. For example, the evidence for the need for Dark Matter is in the rotation curves of galaxies, in the process of formation of structures in our Universe, and in some cases of weak gravitational lensing sources like the Bullet Cluster. So, we see that the scales for which DM is needed range from a few kpc to cosmological scales of the order of hundreds of Mpc. Consequently, modifying the theory of gravitation used to solve one problem may not solve the others.
\newline
\newline
There are many candidates for DM nowadays. For example, axions, the Neutralino, and WIMPs in general, referred to as Cold Dark Matter (CDM), are the most likely candidates. Otherwise, Sterile neutrinos and other particles, referred to as Warm Dark Matter, have some problems related to the constraints on their mass, which make them less appealing.
\newline
Most of the models for DM concern fermions that interact only weakly. This makes a fermionic particle coupled to gravity, i.e., a spinor in curved space-time, an interesting subject of study. For example, in a cosmological setting, see \cite{Chimento}, \cite{Szekeres}, \cite{Shapiro}, \cite{Saha}, \cite{Ribas}, \cite{Jantzen} and \cite{Villalba}. 
\newline
\newline
Most of the works cited above focus on the background cosmology, i.e., what influence a classical spinor has on the expansion of the Universe. In particular, \cite{Magueijo} focuses on its influence on the cosmological singularity in the context of what is referred to as Bouncing Cosmology.
\newline
On the other hand, an analysis of the evolution of cosmological perturbations is missing or incomplete.
\newline
\newline
In this thesis, we partly fill this gap by taking advantage of the results presented in the article \cite{Fabbri}, thanks to which the spinor Stress-Energy Tensor (SET) can be mapped to that of a real (imperfect) fluid. 
\newline
We also study the physics of spherically symmetric objects formed or sustained by classical spinors, with the aim of assessing whether a description of a DM halo, in the context of the formation of structures, or of generic compact objects, is viable.
\newline
\newline
Before presenting our results in Chapters \ref{Cosm_Perturb}, \ref{SET_Decomp} and \ref{Spher_Symm}, we offer the reader a general smattering of the topics that are useful to fully understand the framework in which we are working. In Chapters \ref{Intro_GR}, \ref{Spinor_Bundles}, and \ref{Spinor_Model}, we discuss the key features of General Relativity, what Spin Bundles are, the Spinor Model we worked on and the results obtained from it in the context of the cosmological background.

\mainmatter

\renewcommand{\chaptermark}[1]{%
\markboth{\chaptername~\thechapter\ --\ #1}{}%
}

\chapter{A brief recap of General Relativity}\label{Intro_GR}

As stated in the Introduction, before getting to the point, we believe that the reader should have a brief recap of General Relativity; particularly of its logical principles and main features.

\section{The logical principles at the basis of the Theory}

As Einstein himself said in his article from 1916, there is a main reason that lies at the basis of knowledge and empirical science to extend the principle of relativity to all reference frames and not only to inertial ones. This is the fact that, in general, the cause of a phenomenon must be something that is experimentally measurable and not fictitious, like the choice of a reference frame. Thus, the inertial forces introduced by Newton to explain the different behavior of objects in what he calls non-inertial reference frames cannot exist, since their cause is simply the choice of a particular reference frame. As a consequence, the different behavior of objects in these frames must be explained in another way, and there is no reason to distinguish the inertial frames from the non-inertial ones. As we will soon see, this is due to how differently the geometry of space-time appears in these frames. Indeed, this is fully in agreement with the Machian theory, in which Einstein believed, that masses at infinity influence the behavior of objects according to their relative movement with respect to them\footnote{In the language of GR, masses at infinity modify the geometry of space-time, which appears in different ways according to the reference frame chosen.}. 
\newline
This reasoning is also supported by the fact that observers in a frame moving in empty space with an acceleration $\mathbf{g}$ with respect to one inertial frame cannot assert that they are in a non-inertial frame because, from their point of view, all objects in their proximity behave as if they are on a frame on the surface of the Earth. Therefore, they could say that they are conducting their experiments in an inertial frame on the surface of the Earth, since all objects in their frame experience an acceleration that is equal to that on Earth and is independent of their composition. As a consequence, if we consider the second reference frame as legitimate, we must similarly consider the first. By this follows the need to extend the principle of relativity also to non-inertial frames. 
\newline
Furthermore, thanks to this last \textit{Gedankenexperiment}, we see that an enlargement of the Theory of Special Relativity must also include a theory of gravitation, since we can generate a gravitational field only by a change of reference frame.
\newline
\newline
However, we notice that in non-inertial reference frames, spatial geometry is not Euclidean. This fact can be seen better if we imagine two reference frames: one inertial and one that is uniformly rotating about the $z$ axis of the former. If we draw a circumference on the $xy$ plane of the inertial frame, we can reasonably suppose that this remains a circumference also in the non-inertial one. Nevertheless, if we measure it in the inertial frame using a ruler that is at rest with respect to this frame, we find that the ratio between the circumference and the diameter is $\pi$, but if we do the same measurement with the same ruler in the non-inertial frame, since when we measure the circumference the ruler is subjected to a boost in the tangential direction, but not in the radial one, we find that, in this second frame, the ratio between the circumference and the diameter is greater than $\pi$. As a consequence, the spatial geometry in this frame cannot be Euclidean. 
\newline
Furthermore, in this frame, an observer in the origin sees a clock placed in a point on the circumference running slower than his own due to the boost given to it by the rotation. So, in this frame, there is also no way to define time such that its flow does not depend on position.
\newline
These facts have a fundamental consequence: we cannot use the same coordinates used in Special Relativity to describe these systems, since those coordinates were based on Euclidean geometry and on the fact that time flowed at the same rate for observers at rest with respect to each other. Hence, the only way to solve this problem is to assume that the laws of Physics must be the same in any system of coordinates. Consequently, the main principle of General Relativity becomes:

\begin{center}
\fbox{ \begin{minipage}{12cm}
\textit{The general laws of nature are to be expressed by equations which
hold good for all systems of coordinates, that is, are covariant with
respect to any substitutions whatever (generally covariant).}
\end{minipage}}
\end{center}
It is easy to see that the principle, stated in this way, includes the fact that the laws of Physics must be the same in any reference frame, because, among all systems of coordinates, there are also those that correspond to the description of all possible relative motions between reference frames.

\section{The metric tensor and the gravitational field}

Considering what we have pointed out in the previous section, we would like to formulate a new theory of gravitation whose equations abide by the principle just stated.
\newline
The cornerstone for doing this is the metric tensor.
\newline
Indeed, if we state another reasonable principle with a strong physical foundation, the identification of the metric tensor with the gravitational field comes quite naturally. This principle is:

\begin{center}
\fbox{ \begin{minipage}{12cm}
\textit{For infinitely small four-dimensional regions the theory of relativity in the restricted
sense is appropriate, if the coordinates are suitably chosen.}
\end{minipage}}
\end{center}
The naturalness of this principle can be understood if we think about astronauts in the ISS, who experience this fact every day, as they are in free fall on Earth.
\newline
Hence, there exists a reference frame in which an observer at rest with respect to it does not measure any gravitational field and, consequently, there exists a system of coordinates ($X^0,X^1,X^2,X^3$) in which the line element takes the particular form

\begin{equation}
    ds^2 = -(dX^0)^2 + (dX^1)^2 + (dX^2)^2 + (dX^3)^2,
    \label{SR_Line_Element}
\end{equation}
whose value is totally independent of the orientation of the chosen reference frame.
\newline
If we change the reference frame and choose, for example, a reference frame at rest on the surface of the Earth described by the coordinates ($x^0,x^1,x^2,x^3$), we obtain that the line element takes the form:

\begin{equation}
    ds^2 = g_{\mu \nu}dx^\mu dx^\nu.
    \label{GR_Line_Element}
\end{equation}
The 10 independent quantities $g_{\mu \nu}$ are the components of what is known as the metric tensor.
\newline
Since an observer at rest in this frame can measure a gravitational field and, in general, the $g_{\mu \nu}$ are functions of the four new coordinates chosen, we can conclude that their non-constancy is linked to the gravitational field.
\newline
\newline
Another way to see this fact is by means of the equation of motion of a free particle.
\newline
Thanks to our daily experience, we know that the trajectory of a body in a gravitational field is not a straight line, as would occur in an inertial frame, but a curved line. We also know that, in Special Relativity, the equation of motion of a free particle is given by the variation of the action

\begin{equation}
    S = -mc \int \sqrt{-ds^2}
    \label{Free_Particle_Action}
\end{equation}
with respect to $x^\mu$. This, using the line element given by the eq. \ref{SR_Line_Element}, leads to 

\begin{equation}
    \frac{d^2x^\mu}{ds^2} = 0,
    \label{Eq.Mot_Free_Part_SR}
\end{equation}
which is the equation of a straight line in a four-dimensional space, as expected.
\newline
By the main principle of General Relativity, the action of a free particle in motion in a gravitational field, since $S$ is a scalar quantity under change of coordinates, must be the same. So, by varying the action as done before, but using the line element given by eq. \ref{GR_Line_Element}, we obtain\footnote{We notice that this calculation can be done only if the line element is timelike. It can be proven, but with a different procedure, that the same equations hold for a null line element (i.e., light is also subjected to a gravitational field).}

\begin{equation}
    \frac{d^2x^\mu}{ds^2} + \Gamma^{\mu}_{\hphantom{\mu} \nu \rho} \frac{dx^\nu}{ds} \frac{dx^\rho}{ds} = 0.
    \label{Eq.Mot_Free_Part_GR}
\end{equation}
The 40 independent quantities $\Gamma^{\mu}_{\hphantom{\mu} \nu \rho}$ are called Christoffel Symbols and are defined in terms of the metric as
\begin{equation}
    \Gamma^{\mu}_{\hphantom{\mu} \nu \rho} = \frac{1}{2}g^{\mu \sigma}(\partial_\nu g_{\sigma \rho} + \partial_\rho g_{\nu \sigma} - \partial_\sigma g_{\nu \rho}).
    \label{Christ_Symb}
\end{equation}
Through these facts, we immediately see that the deviation from straightness is expressed by the Christoffel Symbols, which, since they are functions of the partial derivatives of the metric tensor, confirm what we have stated about the fact that a non-constant metric tensor is directly linked to the gravitational field.
\newline
Since the metric tensor, for those who know a little about differential geometry, determines the metrical properties of space-time, we see that gravity is no longer a force, as stated by Newton, but is something directly related to the geometry of our Universe. 

\section{The equations of motion for the gravitational field}

Since the gravitational field is represented by the metric tensor, whose components are the $g_{\mu \nu}$, in order to complete this theory of gravitation, we must find the equations that determine how they change in the presence of a matter source. 
\newline 
\newline
Before doing that, we must point out the big difference between real gravitational fields (i.e., due to massive objects) and gravitational fields due to the choice of a non-inertial reference frame.
\newline
Indeed, the latter can be eliminated in every region of space-time by a simple change of coordinates, while the former cannot, since, in this case, the gravitational field is not uniform throughout space-time. This feature is what marks a curved space-time.
\newline
\newline
For those who know some differential geometry, we can say that the Tangent-Bundle constructed on space-time, visible as a Differential Manifold, in the latter case is trivial or trivializable, while in the first case is not. Therefore, while in the second case we do not need to define a Connection and a Curvature Form to parallel transport the tangent vectors (that is, the Curvature Form vanishes), in presence of a gravitational field we need all the machinery of Differential Geometry to do computations.
\newline
These facts pose the problem of choosing a suitable connection for what we have to do\footnote{We notice that, in his article from 1916, Einstein implicitly assumes it, but his choice is not the only one. This is due to the fact that, at the time, the Theory of Connections had not been developed yet and the only known possible connection was the Levi-Civita one.}.
\newline
First of all, we require it to be a metric connection. This is a very reasonable choice, since the way vectors are parallel transported must, in some sense, follow the gravitational field present in that specific region of space-time, i.e., the metric in that point. Hence, our connection must be linked to the metric defined on the manifold. Since this choice does not uniquely define the connection, we have to establish also the Torsion Form defined on the Tangent-Bundle. For simplicity, we choose it to be zero, that is, we choose the Levi-Civita Connection, so that we can use expression \ref{Christ_Symb} for the connection components\footnote{We notice that this is not the unique possible choice, as we will see in later chapters.}. With this choice, we see that the Curvature Tensor, whose components with respect to the coordinate basis are

\begin{equation}
    R^{\alpha}_{\hphantom{\alpha} \mu \beta \nu} = \partial_\beta \Gamma^{\alpha}_{\hphantom{\alpha} \nu \mu} - \partial_\nu \Gamma^{\alpha}_{\hphantom{\alpha} \beta \mu} + \Gamma^{\alpha}_{\hphantom{\alpha} \beta \lambda}\Gamma^{\lambda}_{\hphantom{\gamma} \nu \mu} - \Gamma^{\alpha}_{\hphantom{\alpha} \nu \lambda}\Gamma^{\lambda}_{\hphantom{\lambda}  \beta \mu},
    \label{RiemannTen}
\end{equation}
the Ricci tensor $R_{\mu \nu} := R^{\alpha}_{\hphantom{\alpha} \mu \alpha \nu}$ and the Ricci scalar $R := g^{\mu \nu}R_{\mu \nu}$,
are entirely expressed in terms of the metric components and their first and second derivatives.
\newline
\newline
Now that we have chosen the connection, we can derive the equations of motion of the gravitational field. In particular, we are looking for equations that contain at most the second derivatives with respect to space and time of the $g_{\mu \nu}$. Hence, in order to satisfy this last requirement and general covariance, we need to find a scalar action that contains at most their first derivatives with respect to space and time. Unfortunately, there is a problem in doing this, because every such a scalar is obtained by a suitable contraction of Christoffel Symbols. These, however, as stated before, can be put to zero in every point in space with an appropriate choice of coordinates. So, since, if a scalar is zero in one system of coordinates, it is zero in every system of coordinates, we have that every scalar constructed in this way is identically zero. As a consequence, we have to make another choice for the action. For example, we can choose a scalar which contains also the second derivatives of the metric, but in such a way that they give a boundary term. In this way, when we compute the variation of the action, they give a null contribution due to Hamilton Principle.
\newline
It can be proven that one possible choice\footnote{For more details on the argument we refer the reader to \cite{Landau}.} (in particular, the only one for our choice of the connection) is given by the Ricci scalar $R$. So, the action for the gravitational field is 

\begin{equation}
    S_g = -\frac{1}{\kappa}\int R \sqrt{-g} \hspace{1mm} d^4x,
    \label{ActionGravField}
\end{equation}
where $\kappa$ is a dimensional coupling constant, and the minus sign is fully justified in the book \cite{Landau}.
\newline
Varying this action with respect to the metric components leads to the equations for the gravitational field in the absence of matter:

\begin{equation}
    R_{\mu \nu} - \frac{1}{2}g_{\mu \nu}R = 0.
    \label{GravFieldEq_Vacuum}
\end{equation}
We notice that these equations have other solutions besides $g_{\mu \nu} = \eta_{\mu \nu}$, which is the metric of SR, as can be physically and mathematically justified.
\newline
\newline
In case a matter source is present, we have to add the matter action 

\begin{equation}
    S_m = \int \mathcal{L}_m \sqrt{-g} \hspace{1mm} d^4x
    \label{ActionMatter}
\end{equation}
to the previous one, so that the total action becomes $S_{tot} = S_g + S_m$.
\newline
Varying this action with respect to the metric components leads to the equations

\begin{equation}
    R_{\mu \nu} - \frac{1}{2}g_{\mu \nu}R = \frac{\kappa}{2c} T_{\mu \nu},
    \label{GravFieldEq_NC}
\end{equation}
where the tensor $T_{\mu \nu}$ can be proven to be the SET of the matter considered and $\kappa$, taking the weak field limit of these equations, can be proven to be equal to $\frac{16 \pi G}{c^3}$ (where $c$ is the speed of light in vacuum and $G$ is the Universal Gravitational Constant). In this way, they take the known form 

\begin{equation}
    R_{\mu \nu} - \frac{1}{2}g_{\mu \nu}R = \frac{8 \pi G}{c^4} T_{\mu \nu}.
    \label{GravFieldEq}
\end{equation}
These equations show how the gravitational field, described by the components of the metric, change under the presence of a matter source. Furthermore, in the weak field limit, as we expected, they lead to the known equation for the gravitational Newtonian potential: $\Delta \phi = 4 \pi G \rho$.
\newline
\newline
Before talking about how spinors can be mathematically viewed in such a space-time, however, we need to point out two peculiarities that will help to better understand the work presented in this master's thesis.
\newline
The first is that, if we suppose that $\mathcal{L}_m$ is independent of the Christoffel Symbols and their derivatives (this occurs in almost every classical case), and we compute the variation of the total action with respect to the $\Gamma^\mu_{\hphantom{\mu} \nu \rho}$, we obtain eq. \ref{Christ_Symb}. Therefore, if we had chosen the action given by \ref{ActionGravField} without making any assumptions about the Connection defined on the Tangent-Bundle, we would have obtained that a gravitational field described by that Lagrangian would have led to a Tangent-Bundle with zero Torsion. Vice versa, if we had chosen a different Lagrangian we would have obtained a different connection and different field equations, as we will see in Chapter \ref{Spinor_Model}.
\newline 
\newline
The second peculiarity concerns the SET $T_{\mu \nu}$. It can be proven that it satisfies the equations

\begin{equation}
    \nabla^\mu T_{\mu \nu} = 0,
    \label{Cov_Cons_eq}
\end{equation}
which can be recast\footnote{We refer the reader to \cite{Landau} for the complete derivation.} 

\begin{equation}
    \partial_\mu \big( (-g)(T^{\mu \nu} + t^{\mu \nu}) \big)= 0,
    \label{Tot_Cons_eq}
\end{equation}
where $t^{\mu \nu}$ is the gravitational Stress-Energy Pseudotensor and it is an expression of the energetic content of the gravitational field.
\newline
We see, then, that eq. \ref{Cov_Cons_eq} represents the conservation equations for the total energy of the system, given by the sum of gravitational energy and matter's energy. These equations are contained in eq. \ref{GravFieldEq}, since $\nabla^\mu (R_{\mu \nu} - \frac{1}{2}g_{\mu \nu}R) = 0$. Furthermore, if we assume as $T_{\mu \nu}$ that of a perfect fluid, computing its covariant divergence gives back the equations of motion of that fluid in the presence of a gravitational field. Thus, eq. \ref{Cov_Cons_eq} can be viewed as the equations of motion of the matter content.
\newline
We conclude that we cannot establish a priori what the matter distribution is and then determine the gravitational field generated by it, and vice versa. Matter and the gravitational field are intertwined. As Wheeler and Ford write in their book \cite{Wheeler}: ``Spacetime tells matter how to move; matter tells spacetime how to curve.''
\newline
\newline
Considering what we have said so far, we notice that studying the structure of a specific SET is useful for determining how space-time behaves under certain symmetry conditions. That is what we are going to do in this master's thesis.

\chapter{Spin Bundles} \label{Spinor_Bundles}

As we pointed out in the Introduction, since the model we worked on is a spinorial one and the underlying manifold, as stated in the previous chapter, is curved, it is essential to introduce Spin Bundles in order to ensure that the reader can fully understand all the calculations done in the following chapters from both the mathematical and conceptual point of view.
\newline
\newline
The whole chapter is technical. The reader not interested in the mathematical details behind the model we are going to study can skip directly to Section \ref{PhysApplicSpin}.

\section{The Clifford Algebra}

Since Spin Bundles are directly related to Clifford Algebras and their Spin Groups, it is essential to mention their most important properties and features.
\newline
\newline
Let $V$ be a vector space over a field $\mathbb{K} = \mathbb{R}, \mathbb{C}$ and $Q: V\longrightarrow \mathbb{K}$ a quadratic form defined by the bilinear form $B$ such that $\forall v,w \in V$, $Q(v) = B(v,v)$, and: 
\begin{equation}
    Q(v + w) - Q(v) - Q(w) = 2B(v, w).
    \label{Reg_Parallelogramma}
\end{equation}
Let also $T(V) := \bigoplus_k V^{\bigotimes k}$ be the Tensor Algebra constructed from $V$ and $I(Q) \subset T(V)$, the two-sided ideal multiplicatively generated by $v \otimes v - Q(v)$. Its Clifford Algebra is defined as:
\begin{equation}
    C(Q) := T(V)/I(Q).
    \label{Clif_Algebra}
\end{equation}
Given $v,w \in V \subset C(Q)$, it is easy to see that

\begin{equation}
    vw + wv = 2B(v, w).
    \label{Clif_Relations}
\end{equation}
Furthermore, considering that each base of $T(V)$ is naturally projected in the Clifford Algebra by means of the natural projection map $\pi: T(V) \longrightarrow C(Q)$ s.t. $\pi(e_{i_1} \otimes ... \otimes e_{i_k}) = e_{i_1}...e_{i_k}$, where $\{e_i\}_{_{i=1,...,n}}$ is a base in $V$, we have that their elements, considered as elements of $C(Q)$, are also generators of $C(Q)$. However, it is easy to see that those that are linearly independent are only the $e_I = e_{i_1}e_{i_2}...e_{i_k}$ s.t. $I = \{i_1, ..., i_k \}$ and $1 \leq i_1 < i_2 < ... < i_k \leq n$. From this fact we obtain that dim$C(Q)$ = $2^n$.
\newline
In addition, since, by the Sylvester's Theorem, a quadratic form is uniquely determined by its signature, we denote $C(p, q) := C(I_{p,q})$, if $\mathbb{K} = \mathbb{R}$, and $C(n) := C(I_n)$, if $\mathbb{K} = \mathbb{C}$, where $I_n$ and $I_{p,q}$ are the signature matrices of each quadratic form and $p$ and $q$ are respectively the number of positive and negative eigenvalues, such that $p + q = n$.
\newline
\newline
It can be proved that every  Clifford Algebra benefits from the following property, called the universal property of Clifford Algebras. Let $E$ be an associative $\mathbb{K}$-algebra with unit element $1 \in E$ and let $j: V \longrightarrow E$, s.t. $j(v)^2 = Q(v)$ $\forall v \in V$, be a linear map. There exists a unique homomorphism of $\mathbb{K}$-algebras $J:C(Q) \longrightarrow E$ s.t. $J(v) = j(v)$ $\forall v \in V$. This homomorphism is such that $J(vw) := J(v)J(w)$, $\forall v,w \in V$ and for every kind of product of elements in $V$.
\newline
This property can be used to prove a lot of isomorphisms between Clifford Algebras. The most important are\footnote{The proofs of these isomorphisms can be found in \cite{Friedrich}.}:
\begin{itemize}
    \item $C(n) \cong C(p, q) \bigotimes_\mathbb{R} \mathbb{C}$
    \item $C(q, p+2) \cong C(p, q) \bigotimes_\mathbb{R} C(0,2)$
    \item $C(q+2, p) \cong C(p, q) \bigotimes_\mathbb{R} C(2, 0)$
\end{itemize}
Furthermore, it can be easily proven that $C(0, 1) \cong \mathbb{C}$ and $C(2) \cong M_2(\mathbb{C})$, where $M_2(\mathbb{C})$ is the algebra of $2 \times 2$ matrices with complex entries. From these facts and the previous isomorphisms directly follows that $C(n + 2) \cong C(n) \bigotimes_\mathbb{C} M_2(\mathbb{C})$ and, consequently, we obtain the isomorphisms:

\begin{equation}
    \begin{split}
    &\text{if } n = 2k\\
    & C(n) \cong M_2(\mathbb{C}) \otimes_\mathbb{C} M_2(\mathbb{C}) \otimes_\mathbb{C} ... \otimes_\mathbb{C} M_2(\mathbb{C}) \cong End(\mathbb{C}^2 \otimes_\mathbb{C} \mathbb{C}^2 \otimes_\mathbb{C} ... \otimes_\mathbb{C} \mathbb{C}^2) \cong End(\mathbb{C}^{2^k}),\\
    &\text{if } n = 2k + 1\\
    & C(n) \cong \big( M_2(\mathbb{C}) \otimes_\mathbb{C} ... \otimes_\mathbb{C} M_2(\mathbb{C}) \big) \oplus \big( M_2(\mathbb{C}) \otimes_\mathbb{C} ... \otimes_\mathbb{C} M_2(\mathbb{C}) \big) \cong End(\mathbb{C}^{2^k}) \oplus End(\mathbb{C}^{2^k}).
    \label{Isomorphisms}
    \end{split}
\end{equation}
This last result leads to important consequences. Since $C(p, q) \subset C(n)$, this isomorphism induces a non-trivial representation of $C(p, q)$ in $End(\mathbb{C}^{2^k})$ whether $n$ is even or not. As a consequence, $\mathbb{C}^{2^k}$ is a representation space for the Clifford Algebra $C(p,q)$. The elements of this space are called spinors\footnote{We note that this denomination is due only to historical reasons, as we will see later.}.
\newline
\newline
The Clifford Algebra benefits also from another property. Since, $\forall v \in V$, the map $u:V \longrightarrow C(Q)$ s.t. $u(v) = -i(v)$ (where $i: V \longrightarrow C(Q)$ is the inclusion map) is well-defined, there exists a unique algebra homomorphism $\alpha: C(Q) \longrightarrow C(Q)$ s.t. $\alpha^2 = Id_{C(Q)}$. This property allows to write the algebra as the direct sum of two subspaces $C(Q)^+$ and $C(Q)^-$, where $C(Q)^{\pm} := \{ x \in C(Q): \alpha(x) = \pm x \}$. It is easy to see that the first one is a subalgebra of $C(Q)$ and that its dimension is dim$C(Q)^+ = 2^{n-1}$. 
\newline
\newline
Last but not least, there are two other maps that can be defined on the Clifford Algebra. These are the antiinvolutions $\tau$ and $*$, which are defined using the antiinvolution $t: T(V) \longrightarrow T(V)$, s.t. $t(v_1 \otimes v_2 \otimes ... \otimes v_k) = v_k \otimes ... \otimes v_2 \otimes v_1$ and $t(x \otimes y) = t(y) \otimes t(x)$, and the map $\alpha$ defined before. 
\newline
More in detail, since $t$ preservers the ideal $I_Q$, it induces an antiinvolution on $C(Q)$. This antiinvolution is precisely $\tau$, which is defined as the map $\tau: C(Q) \longrightarrow C(Q)$ s.t. $\tau(v_1 v_2 ... v_k) = v_k ... v_2 v_1$ and $\tau(xy) = \tau(y)\tau(x)$. Instead, $*$ is defined as the composition of $\tau$ and $\alpha$, that is $* = \tau \circ \alpha: C(Q) \longrightarrow C(Q)$, s.t. $(v_1 v_2 ... v_k)^* = (-1)^k v_k ... v_2 v_1$.
\newline
This last map allows to define the well-known Spin group.

\section{The Group $Spin(Q)$}\label{Spin_Group_Sec}
The Group $Spin(Q)$ is defined as

\begin{equation}
    Spin(Q) : = \{ x \in C(Q)^+: xx^* = \epsilon, \hspace{0.3cm} xVx^{-1} \subset V \},  \hspace{1cm} \epsilon =
    \begin{cases}
        \pm 1 \text{ if } \mathbb{K} = \mathbb{R}\\
        1 \text{ if } \mathbb{K} = \mathbb{C}
    \end{cases}
    \label{Spin_Group}
\end{equation}
It is easy to see that this group is a subgroup of $(C(Q)^+)^\times$ ($x^{-1} = \epsilon x^*$ for every $x \in Spin(Q)$) and it can be proven that it is a Lie Group. Furthermore, one can also prove that, if $Q$ is given by $I_{p,q}$ with $p,q \geq 1$, this group is not connected.
\newline
This group has a very important property, which we state in the following 

\begin{theorem*}
    Let $B$ be a non-degenerate bilinear form on a vector space $V$ over a field $\mathbb{K} = \mathbb{R}$ or $\mathbb{C}$ and let $Q(x) := B(x,x)$. There exists a continuous surjective group homomorphism $\rho: Spin(Q) \longrightarrow SO(Q)$ s.t. ker$(\rho)$ = \{1, -1\}.
\end{theorem*}

\noindent
The homomorphism can be proven to be the map $\rho(x)(w) := xwx^{-1}$, $\forall w \in V$ and  $\forall x \in Spin(Q)$.
\newline
This theorem has an important consequence. It shows that $Spin(Q)$ is a double cover of $SO(Q)$; fact that is directly related to the definition of Spin Bundles and connections on them. 
\newline
\newline
Since $Spin(Q)$ is a Lie Group, it owns an associated Lie Algebra. It is easy to see, that this Lie Algebra is defined by

\begin{equation}
    \mathfrak{spin}(Q) : = \{ x \in C(Q)^+: x + x^* = 0, \hspace{0.3cm}  xv + vx^* \in V \hspace{0.3cm} \forall v \in V \}.
    \label{Spin_Algebra}
\end{equation}
If we choose an orthonormal basis in $V$ and we construct the associated basis of $C(Q)^+$, we can see that, among all its generators, only the elements $e_i e_j$, with $1 \leq i < j \leq n$, are in $\mathfrak{spin}(Q)$\footnote{In particular, the elements $e_{i_1}e_{i_2}...e_{i_{2k}}$ with $k$ even or equal to 0 do not satisfy the first condition, while those with $k$ odd and $\neq 1$ do not satisfy the second.}. As a consequence, every element of $\mathfrak{spin}(Q)$ is generated only by them. Moreover, since they are linearly independent, they form a basis for $\mathfrak{spin}(Q)$. This leads to dim$\mathfrak{spin}(Q) = \frac{n(n-1)}{2}$. 
\newline
In addition, it is easy to see, exploiting the usual procedure, that the homomorphism $\rho$ between $Spin(Q)$ and $SO(Q)$ leads to an homomorphism of algebras $d\rho$ between $\mathfrak{spin}(Q)$ and $\mathfrak{so}(Q)$. This homomorphism inherits the surjectivity of $\rho$ and, since dim$\mathfrak{spin}(Q)=$ dim$\mathfrak{so}(Q)$, it becomes also injective. It immediately follows that the two algebras are isomorphic. This fact has the consequence that $\mathfrak{so}(Q)$ admits a representation in $C(Q)^+$. 
\newline 
\newline
The isomorphism just described can be used to map the basis of $\mathfrak{spin}(Q)$ to a basis of $\mathfrak{so}(Q)$ and vice versa. In particular, given the well-known basis of $\mathfrak{so}(Q)$ made up by the antisymmetric matrices $E_{ij}$, it can be proved that $d\rho(e_i e_j) = 2E_{ij}$, where $1 \leq i < j \leq n$. In practice, is more convenient to make the antisymmetry of the generators of $\mathfrak{spin}(Q)$ explicit and redefine the generators as $2e_i e_j = e_i e_j - e_j e_i = [e_i, e_j]$. This leads to $d\rho([e_i, e_j]) = 4 E_{ij}$. This fact will be useful when we compute the connection, because it will allow us to extend the sums not only on $i<j$.
\newline
\newline
Furthermore, it can be proven that the restriction to $Spin(Q)$ and to $\mathfrak{spin}(Q)$ of the representation $\kappa$ of $C(n)$ on $\mathbb{C}^{2^k}$, whether for $n = 2k$ or $n = 2k + 1$, is faithful. So the group and the algebra both admit a faithful representation on $\mathbb{C}^{2^k}$. 
\newline
In the end, since $\mathbb{R}^n \in C(n)$, an element $x \in \mathbb{R}^n$ can be viewed as an endomorphism on $\mathbb{C}^{2^k}$. This leads to the definition of Clifford multiplication between vectors and spinors, that is, a linear map $\mu: \mathbb{R}^n \bigotimes \mathbb{C}^{2^k} \longrightarrow \mathbb{C}^{2^k}$ s.t., $\forall x \in  \mathbb{R}^n$ and $\forall \psi \in \mathbb{C}^{2^k}$, $x \cdot \psi \equiv \mu(x \otimes \psi) := \kappa(x)\psi$. It can be easily proven that the Clifford multiplication is a homomorphism of $Spin(Q)$-representations, in the sense that, given $g \in Spin(Q)$, $x \in \mathbb{R}^n$ and $\psi \in \mathbb{C}^{2^k}$, $\kappa(g)(x \cdot \psi) = (\rho(g)x) \cdot (\kappa(g) \psi)$.
\newline
All these properties will be very useful for the definition of the Dirac operator and to study some properties of physical particles described by spinorial fields.

\section{Spin Structures}
As we have just seen and as we will see more in detail at the end of this chapter, by definition, spinors can be multiplied by vectors in $\mathbb{R}^n$. Unfortunately, there is no non-trivial representation $\lambda$ of the Linear or of the Orthogonal Group that is compatible with the Clifford multiplication, that is, $\forall A \in SO(Q)$, $\lambda(A)(x \cdot \psi) = (Ax) \cdot (\lambda(A) \psi)$.
\newline
As a consequence, if we defined spinorial fields on pseudoriemannian manifolds as sections of the vector bundle associated to the linear frame bundle or to the orthonormal frame bundle, as we do for usual vector fields, they would not be well-defined.
\newline
This fact leads to the need of defining Spin Structures and Spin Bundles.
\newline
\newline
In order not to dwell too much on the argument, we take for granted the reader's knowledge about fibre bundles, $G$-principal bundles and their associated vector bundles.
\newline
Moreover, for the same reason, we prefer not to spend much time on the problem of the existence of these structures and these bundles, but we bring to the attention of the reader their definitions and constructions, that will be useful in the next section to deal with the much more physical-related problem of defining a connection on them.
\newline
\newline
Since, as we have shown before, $Spin(Q)$ is related to $SO(Q)$, it is quite natural to define a spin structure so that is composed by a $Spin(Q)$-principal bundle linked in some way to an $SO(Q)$-principal bundle. This fact has also a physical justification. The two fibre bundles must be linked because we would like the components of a spinorial field and the components of the connection computed with respect to a local basis of $Spin(Q)$ (gauge in the $Spin(Q)$-bundle) to coincide to those computed with respect to a local orthonormal frame (gauge in the $SO(Q)$-bundle). This will prove to be very useful in the next chapters.
\newline
As a consequence, we write the

\begin{definition*}
    Let M be a pseudoriemannian manifold of dimension $n$, let $(E, \pi_E, M; SO(Q))$ be an $SO(Q)$-principal bundle over M. A spin structure on the principal bundle $E$ is a pair $(P, \Lambda)$ where
    \begin{itemize}
        \item $P$ is a $Spin(Q)$-principal bundle over $M$
        \item $\Lambda: P \longrightarrow E$ is a 2-fold covering for which the diagram\\
        
        \vspace{0.3cm}
        
        \begin{center}
            \begin{tikzcd}
                P \times Spin(Q) \ar[dd, "\Lambda \times\rho"] \ar[rr] & & P \ar[dd, "\Lambda"] \ar[rrd, "\pi_P"]\\
               & & & & X\\
               E \times SO(Q) \ar[rr] & & E \ar[rru, "\pi_E"]\\
            \end{tikzcd}
        \end{center}
        commutes. The rows represent the action of each group on the corresponding principal bundle and $\pi_P$ and $\pi_E$ are the projection maps of the two fibre bundles. 
    \end{itemize}
\end{definition*}
\noindent
Unfortunately, the existence of this structure is not always guaranteed, but it is related to topological properties of the $E$ fibre bundle. More precisely, the following theorem, whose proof (only for riemannian manifolds) can be found on \cite{Lawson} and on \cite{Friedrich}, holds:

\begin{theorem*}
    Let $E$ be an oriented principal bundle over a manifold $X$. Then, there exists a spin structure on $E$ if and only if the second Stiefel-Whitney class of $E$ is zero, that is $w_2(E) = 0$.
    \newline
    Furthermore, if $w_2(E) = 0$, then the distinct spin structures on $E$ are in one-to-one correspondence with the elements of $H^1(X;\mathbb{Z}_2)$.
\end{theorem*}
\noindent
We see, then, that the existence of a spin structure is related to $H^1(X;\mathbb{Z}_2)$, the first cohomology group of the topological space $X$ with coefficients in the field $\mathbb{Z}_{2}$.
\newline 
Nevertheless, since we are physicists and we have proofs of the existence of fermions in our Universe, in this work we assume that this kind of structure exists.
\newline
\newline
From this structure we can construct the Spin Bundle, which is no less than a vector bundle associated to the $Spin(Q)$-principal bundle. 
\newline
In particular, given the representation $\kappa$ of the Spin Group on the vector space $\mathbb{C}^{2^k}$, we define the Spin Bundle associated to the spin structure $(P, \Lambda)$ as the vector bundle 

\begin{equation}
    S := P \times_\kappa \mathbb{C}^{2^k}.
    \label{Spin_Bundle}
\end{equation}
We immediately see that, since the Tangent Bundle is isomorphic to the vector bundle associated to the $SO(Q)$-principal bundle (that is $TM \cong E \times_{SO(Q)} \mathbb{R}^n \cong P \times_{\rho} \mathbb{R}^n$), the Clifford multiplication $\mu$ induces a bundle morphism of the associated bundles: $\mu : TM \bigotimes S \longrightarrow S$.
\newline
Therefore, if we defined spinorial fields as global sections of this vector bundle, this definition would be well-posed, since, as we have seen  in the previous section, the Clifford multiplication is a homomorphism of $Spin(Q)$-representations. So, we affirm that a spinorial field $\psi$ is a global section of $S$, that is $\psi \in \Gamma(S)$.
\newline 
\newline
Now that we have shown how spin structures and Spin Bundles are constructed and what spinorial fields really are, we would like to show how we can parallel transport them. This requires the definition of a connection on the bundles just constructed. In this way, we will be able to fully justify all the calculations done in the following chapters.

\section{Connections on Spin Bundles}
As we know from Differential Geometry, there is no natural way to parallel transport sections of a generic fibre bundle. It is for this reason that we need to define a connection on it.
\newline
Given a point $\xi$ on a fibre bundle $E$ with fibre $F$ and given $T_\xi E$ the tangent space to the bundle in that point, an infinitesimal connection is a linear surjective map $\omega_\xi: T_\xi E \longrightarrow T_\xi F$. It allows to have a decomposition of $TE$ in each point of $E$, in the sense that $T_\xi E \cong T_\xi F \bigoplus \mathcal{H}_\xi$, where $\mathcal{H}_\xi$ is called horizontal subspace. This decomposition is what allows to define a parallel transport and what ensures its uniqueness.
\newline
However, if we are working with a $G$-principal bundle $\mathcal{P}$, the presence of a right, free and transitive action of the group on the fibres $\mathcal{G}$ ensures that every vertical vector (that is every $X \in T_p \mathcal{G}$) is in a one-to-one correspondence with an element of the Lie Algebra of $G$. Therefore, the connection on these bundles can be viewed as a map $\omega_p: T_p \mathcal{P} \longrightarrow Lie(G)$. 
\newline
\newline
Now, let $Z$ be a connection on the $SO(Q)$-principal bundle $E$, so that $Z: TE \longrightarrow \mathfrak{so}(Q)$. 
\newline
The following theorem holds:

\begin{theorem*}
     Let $(P, \Lambda)$ be a spin structure defined on the $SO(Q)$-principal bundle $E$. If $d\Lambda$ maps horizontal vectors in horizontal vectors, then any $\tilde{Z}$ is a connection on $P$ if and only if it is a lift of the connection $Z$ defined on $E$, that is, iff the following diagram commutes
     \begin{center}
         \begin{tikzcd}
            TP \ar[dd, "d\Lambda"] \ar[rr, "\tilde{Z}"] & & \mathfrak{spin}(Q) \ar[dd, "d\rho"]\\
            \\
            TE \ar[rr, "Z"] & & \mathfrak{so}(Q)\\
        \end{tikzcd}
    \end{center}
\end{theorem*}

\begin{proof}
    Let $X \in \mathfrak{spin}(Q)$. Its associated vertical vector (fundamental field) in a point $p \in P$ is 

    \begin{equation}
        X^\#_p = \frac{d}{dt} \bigg|_{t = 0} p \hspace{1mm} \text{exp}(tX).
        \label{Fund_Field}
    \end{equation}
    Let us suppose that $\tilde{Z}$ is a connection on $P$. By definition, we have that $\tilde{Z}(X^\#_p) = X$.
    \newline
    If we compute $\Lambda^*(Z)$, we have
    \begin{equation}
        \begin{split}
            \Lambda^*(Z)(X^\#_p) &= Z(d\Lambda(X^\#_p)) = Z \bigg(\frac{d}{dt} \bigg|_{t = 0} \Lambda \big(p \hspace{1mm} \text{exp}(tX) \big)\bigg) = Z \bigg(\frac{d}{dt} \bigg|_{t = 0} \Lambda(p)\hspace{1mm} \rho \big(\text{exp}(tX) \big)\bigg) = \\
            & = Z \bigg(\frac{d}{dt} \bigg|_{t = 0} \Lambda(p)\hspace{1mm} \text{exp} \big(t \hspace{1mm} d \rho(X) \big) \big)\bigg) = d \rho(X) = d \rho \big( \tilde{Z}(X^\#_p) \big) 
        \end{split}
        \label{Lift}
    \end{equation}
    As a consequence, we have that, on vertical fields, $\tilde{Z}$ is a lift of $Z$. 
    \newline
    The inverse follows immediately.
    \newline
    Now, we have to prove that they coincide also on horizontal vectors. 
    \newline
    This follows immediately from the fact that, given an horizontal vector $v_h \in TP$, $d \Lambda(v_h) = w_h$, where $w_h$ is an horizontal vector $\in TE$. From this we can conclude that the two expressions for the connection coincide for every vector in $TP$.
\end{proof}

\noindent
Furthermore, it is easy to prove that, if the theorem holds, the lift is unique. 
\newline
\newline
From now on, for simplicity, we assume that the spin structure defined on the space-time manifold satisfies the hypothesis of the theorem\footnote{Otherwise, $\tilde{Z} = d\rho^{-1}(\Lambda^*(Z)) + \tilde{Z}'$, where $\tilde{Z}'$ is zero on vertical vectors and $-d\rho^{-1}(\Lambda^*(Z))$ on horizontal ones.}. This has a very important consequence on the connection defined on the Spin Bundle. In fact, if we further assume that the connection $Z$ is a $G$-connection\footnote{It can be proven that, in order to be well-defined, such a connection must be a metric connection. This condition guarantee that, once the torsion on the manifold has been fixed, the connection on the $SO(Q)$-principal bundle is uniquely determined without any further choice.}, $\tilde{Z}$ is uniquely determined. 
\newline
In fact, given a local gauge $\varepsilon: U \longrightarrow E$, we have that

\begin{equation}
    Z_\varepsilon := \varepsilon^*(Z) = \sum_{I < J} \omega^{IJ}E_{IJ},  
    \label{SO(Q)_Conn}
\end{equation}
where $\omega^{IJ} := \eta^{JK}\omega^I_{\hphantom{I}K}$ are the components of the connection in the chosen tetrad $e^I$, computed using the structure equations ($\eta_{IJ}$ are the components of the signature matrix of $Q$ and $T^I$ is the Torsion form in the chosen tetrad):

\begin{equation}
    \begin{split}
    &de^I + \omega^I_{\hphantom{I}J} \wedge e^J = T^I \\
    & \eta_{IK} \omega^K_{\hphantom{K}J} + \eta_{KJ} \omega^K_{\hphantom{K}I} = 0.
    \end{split}
    \label{Structure_Eqs}
\end{equation}
Therefore, if we choose a local gauge $\sigma: U \longrightarrow P$ and we notice that $\varepsilon := \Lambda \circ \sigma$ is a local gauge in $E$, calling $e^I$ the associated tetrad, we have that $\varepsilon^*(Z)$ has the expression given by eq. \ref{SO(Q)_Conn}. Then, if we consider that $\tilde{Z}$ is the lift of $Z$, thanks to the properties of the pullback, we obtain

\begin{equation}
    \begin{split}
    Z_\varepsilon = \varepsilon^*(Z) = (\Lambda \circ \sigma)^*(Z) = \sigma^* (\Lambda^*(Z)) = &\sigma^*(d\rho(\tilde{Z})) = d\rho(\sigma^*(\tilde{Z})) =: d \rho (\tilde{Z}_\sigma)\\
    & \Big \Downarrow\\
    \tilde{Z}_\sigma = d \rho^{-1}(Z_\varepsilon) = d \rho^{-1} \Bigg(\sum_{I < J} \omega^{IJ} & E_{IJ} \Bigg) = d \rho^{-1} \Bigg( \frac{1}{2}\sum_{I,J} \omega^{IJ}E_{IJ} \Bigg) =  \\
    = \frac{1}{2}\sum_{I,J} \omega^{IJ} d \rho^{-1}(E_{IJ}) = & \frac{1}{8}\sum_{I,J} \omega^{IJ}[\tilde{e}_{I}, \tilde{e}_{J}], 
    \end{split}
    \label{Spin(Q)_Conn}
\end{equation}
where $\tilde{e}_I$ are the elements of the multiplicative basis of $C(Q)$.
\newline
We see, then, that the components of the connection in the $Spin(Q)$-principal bundle with respect to the local gauge $\sigma$ are the same as those of the connection in the $SO(Q)$-principal bundle with respect to the corresponding local gauge $\varepsilon$. This is perfectly in agreement with what we have stated in the previous section. Therefore, when we choose a tetrad in the $SO(Q)$-principal bundle and we compute the components of the $SO(Q)$-connection with respect to it, we automatically compute the components of the $Spin(Q)$-connection with respect to any local gauge $\sigma$ such that $\Lambda \circ \sigma = \varepsilon$ (in general this is not unique).
\newline 
\newline
Now, we can analyze what happens to spinorial fields in the Spin Bundle.
\newline
From Differential Geometry we know that sections on a vector bundle $E = P \times_\rho V$ associated to a $G$-principal bundle $P$ are in one-to-one correspondence with equivariant maps $\alpha^\psi: P \longrightarrow V$. Since, in Theoretical Physics, we always start from an Action to build a theory, but sections, formally, are equivalence classes and dealing with them is very difficult, exploiting this correspondence, we might think to formulate the action in terms of equivariant maps. In particular, we define (in this case spinorial) local fields as local equivariant maps, that is $\psi_\sigma : = \alpha^\psi \circ \sigma$, where  $\sigma$ is a gauge in the $Spin(Q)$-principal bundle. Since this definition depends on the gauge chosen, we point out that the action must be composed by suitable gauge-invariant combinations of this term, so that the associated theory does not depend on an arbitrary gauge choice.
\newline
Moreover, once specified a connection on $P$, if we establish that the connection on the associated vector bundle is that induced by $P$, we know that this connection defines a covariant derivative which act on the sections of the vector bundle so that $\nabla \alpha^\psi = d\alpha^\psi + \rho_*(\omega) \alpha^\psi$, where $\rho_*$ is the representation of the algebra of the structure group on the vector space isomorphic to the fibre.
\newline
Therefore, in the gauge $\sigma$, we obtain that the local covariant derivative of the spinorial field is

\begin{equation}
      D_\sigma \psi_\sigma := \sigma^* (\nabla \alpha^\psi) =  \bigg(\partial_\mu \psi_\sigma + \frac{1}{8} \sum_{I,J}\omega^{IJ}_{\hphantom{IJ}\mu}[\kappa_*(\tilde{e}_I), \kappa_*(\tilde{e}_J)] \psi_\sigma \bigg) dx^\mu.
    \label{Cov_Dev_Spinors}
\end{equation}
In addition, from the general expression we also notice that, with this choice of connection on the Spin Bundle, since the Clifford multiplication is an homomorphism of $Spin(Q)$-representations, we have 

\begin{equation}
    D_\sigma (X_\sigma \cdot \psi_\sigma) = D_\sigma X_\sigma \cdot \psi_\sigma + X_\sigma \cdot D_\sigma \psi_\sigma.
    \label{Cliff_Mult_Der_Cov}
\end{equation}
This will turn very useful when we compute the conserved current of the spinor model analyzed.
\newline
\newline
We notice that the choice made about the connection on the Spin Bundle is not the only one possible. We could have chosen a connection that it is not induced by that in $P$. The choice made, however, is the most natural for two reasons. The first is that it is the simplest choice that can be made. The second is that, since in the Chapter \ref{Intro_GR} we established that the connection on the Tangent Bundle must be a metric one, it is quite natural that also the connection on the Spin Bundle is linked in some way to the metric connection.
\newline
\newline
Before introducing the notion of Weyl spinor, which will be very useful in the next chapter, we show another important result about the connection just defined.
\newline
We know that, given a generic metric connection with components $\omega^{IJ}$, this can be written as the sum of a Levi-Civita connection and a Contorsion form; that is $\omega^{IJ} = \tilde{\omega}^{IJ} + C^{IJ}$, where $C^{IJ}\wedge e_J = T^I$.  
If we assume that $T^I = 0$, i.e., $\omega^{IJ} = \tilde{\omega}^{IJ}$, in the coordinate basis this condition translates into $\Gamma^\mu_{\hphantom{\mu} \nu \rho} = \Gamma^\mu_{\hphantom{\mu} \rho \nu}$ and the $\Gamma^\mu_{\hphantom{\mu} \nu \rho}$ are given by eq. \ref{Christ_Symb}. This leads to the formula\footnote{It comes out that this formula still holds for a generic metric connection.}

\begin{equation}
    \Gamma^\mu_{\hphantom{\mu} \nu \rho} = e^\mu_{\hphantom{\mu}J}\partial_\nu e^J_{\hphantom{J}\rho} + e^\mu_{\hphantom{\mu}J} \omega^J_{\hphantom{J}I \nu}e^I_{\hphantom{I}\rho}.
    \label{Christ_Symb_Metr_Conn}
\end{equation}

\section{Weyl spinors}
Before showing how the whole mathematical formalism is applied to Physics, we want to explain what Weyl Spinors are, what properties they possess and how the connection acts on them.
\newline
\newline
In Section \ref{Spin_Group_Sec} we have defined the Spin Group and we have seen that $Spin(Q) \subset C(Q)^+ \subset C(n)^+$. In the case $n = 2k$, it is easy to prove that the element $\tilde{e}_1 \tilde{e}_2...\tilde{e}_{2k} \in \mathcal{Z}(C(n)^+)$; that is, it belongs to the center of the subalgebra $C(n)^+$. Hence, $\tilde{e}_1 \tilde{e}_2...\tilde{e}_{2k}$ commutes with all the elements of $Spin(Q)$. Therefore, the map

\begin{equation}
    f := i^{k-q} \kappa(\tilde{e}_1 \tilde{e}_2...\tilde{e}_{2k}): \mathbb{C}^{2^k} \longrightarrow \mathbb{C}^{2^k}
    \label{Weyl_End} 
\end{equation}
is an automorphism of $Spin(Q)$ representations, that is, $f(\kappa(g)\psi) = \kappa(g)f(\psi)$, $\forall g \in Spin(Q)$ and $\forall \psi \in \mathbb{C}^{2^k}$. Furthermore, it is easy to prove that $f^2 = I_{\mathbb{C}^{2^k}}$, since $(\tilde{e}_1 \tilde{e}_2...\tilde{e}_{2k})^2 = (-1)^{k-q}$. 
\newline
As a consequence, the space $\mathbb{C}^{2^k}$ decomposes in the two eigensubspaces of this automorphism;

\begin{equation}
    \Delta^{\pm} := \{\psi \in \mathbb{C}^{2^k}: f(\psi) = \pm \psi \},
    \label{Weyl_Subspaces}
\end{equation} 
respectively called right and left Weyl subspaces. 
\newline
Spinors in these eigenspaces are obviously called right and left Weyl spinors. 
\newline
\newline
It is quite straightforward to see that these subspaces are invariant under the action of the $Spin(Q)$-representation $\kappa$ and that dim$_\mathbb{C} \Delta^+ =$ dim$_\mathbb{C} \Delta^- = 2^{k-1}$. In addition, it can be easily proven that, given $\psi \in \Delta^{\pm}$, $x \cdot \psi \in \Delta^{\mp}$ $\forall x \in \mathbb{R}^{2k}$. Thus, Clifford multiplication induces an homomorphism from $\Delta^{\pm}$ to $\Delta^{\mp}$.
\newline
\newline
There exists a very important theorem concerning Weyl subspaces\footnote{The proof, only for the case of completely positive signature, can be found on \cite{Friedrich}.}:

\begin{theorem*}
    The $Spin(Q)$-representations $\Delta^{\pm}$ are irreducible.
\end{theorem*}

\noindent
This means that every element of the $Spin(Q)$-representation $\mathbb{C}^{2^k}$ can be written as combination of two Weyl spinors, on which the Spin Group acts irreducibly. 
\newline
Therefore, a Weyl spinor cannot be decomposed any further. 
\newline
\newline
It is worth to mention that a similar theorem exists also in odd dimensions, but without the decomposition in Weyl subspaces, since they exist only in even dimensions:

\begin{theorem*}
    Let $n = 2k + 1$. Then, the $Spin(Q)$-representation $\mathbb{C}^{2^k}$ is irreducible.
\end{theorem*}

\noindent
For $n = 2k$, the decomposition in Weyl subspaces translates into a decomposition of subbundles; in the sense that the vector bundle $S$ splits into two subbundles $S^{\pm} := P \times_{\kappa_\pm} \Delta^\pm$.
\newline
This last piece of information leads to another important fact. 
\newline
In the previous section, we defined a connection $\kappa_*(\tilde{Z})$ on $S$. Since every spinor can be written as the direct sum of two Weyl spinors, whose subspaces are invariant under the action of the $Spin(Q)$ Group and under the action of the associated Lie Algebra, we see that the total connection can be split into two connections acting separately on the two subbundles. Hence, the covariant derivative splits into two covariant derivatives, each acting independently on only one of the two subbundles.
As a consequence, if there is no type of interaction between the two Weyl spinors in the system under study, they evolve freely and independently from each other and, in the action of a theory of this kind, we have to separate their free-evolution terms.
\newline
The Lagrangian has to be written, then, as an appropriate linear combination of suitable contractions composed by metric terms and a particular combination of only one Weyl spinorial field.
\newline
\newline
In the end, we point out that Weyl spinorial fields are very important in physics. For example, before the discovery of its mass, it was thought that the neutrino were a left Weyl spinorial field, due to its way of interacting with other matter components.

\section{The physical application of spinors}\label{PhysApplicSpin}
The mathematical construction shown so far is general. 
\newline
However, from the Theory of Relativity we know that we live in a four dimensional manifold equipped with a metric, whose signature matrix is $\eta = \text{diag}(-1, 1, 1, 1)$\footnote{As shown also in the previous chapter, in this work we decided to adopt the mostly-plus signature.}. The immediate implication of this fact is that the Clifford Algebra and the Spin Group we have to consider are respectively $C(3,1)$ and $Spin(3,1)$, whose space of spinors is $\mathbb{C}^4$. So, in our framework, spinors have 4 independent complex components. 
\newline
Moreover, by exploiting the fact that $C(3,1) \subset C(4) \cong M_2(\mathbb{C}) \otimes M_2(\mathbb{C})$, we can write a basis of $\mathbb{R}^4 \in C(3,1)$ in this representation starting from one of $M_2(\mathbb{C})$.
\newline
The most common choice for a basis of $M_2(\mathbb{C})$ concerns Pauli matrices and the Identity:

\begin{equation}
    I = 
    \begin{pmatrix}
        1 & 0\\
        0 & 1
    \end{pmatrix},
    \hspace{1cm}
    \sigma_x = 
    \begin{pmatrix}
        0 & 1\\
        1 & 0
    \end{pmatrix},
    \hspace{1cm}
    \sigma_y = 
    \begin{pmatrix}
        0 & -i\\
        i & 0
    \end{pmatrix},
    \hspace{1cm}
    \sigma_z = 
    \begin{pmatrix}
        1 & 0\\
        0 & -1
    \end{pmatrix}.   
    \label{Pauli_Matrices}
\end{equation}
From these we can define the sought basis as $\kappa_*(\tilde{e}_1) := i \gamma_0 := -i \sigma_x \otimes I$, $\kappa_*(\tilde{e}_2) := i \gamma_1 := -\sigma_y \otimes \sigma_x$, $\kappa_*(\tilde{e}_3) := i \gamma_2 := -\sigma_y \otimes \sigma_y$ and $\kappa_*(\tilde{e}_4) := i \gamma_3 := -\sigma_y \otimes \sigma_z$, where \{$\tilde{e}_1$, $\tilde{e}_2$, $\tilde{e}_3$, $\tilde{e}_4$\} is the corresponding basis of $\mathbb{R}^4 \in C(3,1)$ and  we have defined the gamma matrices

\begin{equation}
    \gamma_0 = 
    \begin{pmatrix}
        0 & -I\\
        -I & 0
    \end{pmatrix},
    \hspace{1cm}
    \gamma_1 = 
    \begin{pmatrix}
        0 & \sigma_x\\
        -\sigma_x & 0
    \end{pmatrix},
    \hspace{1cm}
    \gamma_2 = 
    \begin{pmatrix}
        0 & \sigma_y\\
        -\sigma_y & 0
    \end{pmatrix},
    \hspace{1cm}
    \gamma_3 = 
    \begin{pmatrix}
        0 & \sigma_z\\
        -\sigma_z & 0
    \end{pmatrix}.    
    \label{Gamma_Matrices}
\end{equation}
It can be proven that they satisfy the relations

\begin{equation}
    \gamma_I \gamma_J + \gamma_J \gamma_I = -2 \eta_{IJ}I.  
    \label{CliffMink}
\end{equation}
Hence, $\{i \gamma_1, i \gamma_2, i \gamma_3, i \gamma_4\}$ is the representation of the basis of $\mathbb{R}^4 \in C(3,1)$ we were looking for. 
\newline
The immediate consequence of this fact is that the local expression of the connection on the Spin Bundle becomes

\begin{equation}
    \kappa_*(\tilde{Z}_\sigma) = \frac{1}{8} \sum_{I,J}\omega^{IJ}_{\hphantom{IJ}\mu}[\kappa_*(\tilde{e}_I), \kappa_*(\tilde{e}_J)]dx^\mu = -\frac{1}{8} \sum_{I,J}\omega^{IJ}_{\hphantom{IJ}\mu}[\gamma_I, \gamma_J]dx^\mu =: -\frac{i}{2} \omega^{IJ}_{\hphantom{IJ}\mu}J_{IJ}dx^\mu,
    \label{Spin_Conn_Gamma}
\end{equation}
where we have defined the spin generators $J_{IJ} = -\frac{i}{4}[\gamma_I, \gamma_J]$ and we have adopted the Einstein convention on repeated indices.
\newline
\newline
Clearly, the expression chosen for the gamma matrices is not the unique possible.
\newline
The form given by eq. \ref{Gamma_Matrices} is called the Weyl form, since it allows to have clear distinction between the actions on the two Weyl subspaces. In fact, if we write in this basis the involution $f$ defined in the previous section, we see that it assumes the form

\begin{equation}
    \gamma_5 := i \gamma_0 \gamma_1 \gamma_2 \gamma_3 = 
    \begin{pmatrix}
        I & 0\\
        0 & -I
    \end{pmatrix},
\end{equation}
which shows explicitly what form left and right Weyl spinors assume.
\newline
P.A.M. Dirac, for example, chose another expression for the gamma matrices in its famous article \textit{The Quantum Theory of the Electron} \cite{Dirac}.
Indeed, he was the first to introduce gamma matrices and to show that, in 4-dimensions, spin-$\frac{1}{2}$ particles where well described by objects $\in \mathbb{C}^4$, named spinors after his work.
\newline
\newline
Another credit of the article by Dirac is the introduction of a particular operator, then called Dirac operator, that, mathematically, is the square root of the Laplacian. In curved space-time, this operator is formally defined, with respect to a local orthonormal frame $\tilde{e}_I$, as 

\begin{equation}
    (Di) \psi := \eta^{IJ}(\tilde{e}_I \cdot D_{\tilde{e}_J}) \psi = \eta^{IJ} \kappa_*(\tilde{e}_{I})D_{\tilde{e}_J} \psi.
    \label{Dirac_Op}
\end{equation}
In our framework, since $D_{\tilde{e}_J} = e^{\hphantom{J} \mu}_{J}D_\mu$, this operator assumes the form

\begin{equation}
    (Di) \psi := i\eta^{IJ} \gamma_I e^{\hphantom{J} \mu}_{J} D_\mu \psi = i \gamma^J e^{\hphantom{J} \mu}_{J}D_\mu \psi =: i \Gamma^\mu D_\mu \psi,
    \label{Dirac_Op_2}
\end{equation}
where we have defined the curved gamma matrices $\Gamma_\mu = e^J_{\hphantom{J}\mu} \gamma_J$, which satisfy

\begin{equation}
    \Gamma_\mu \Gamma_\nu + \Gamma_\nu \Gamma_\mu = -2 g_{\mu \nu}I.  
    \label{CliffOldMetr}
\end{equation}
It is easy to see that, in the case of a flat manifold ($e^J_{\hphantom{J}\mu} = \delta^J_{\hphantom{J}\mu}$), the Dirac operator reduces to $(Di)\psi = i \gamma^\mu \partial_\mu \psi = i \slashed{\partial}\psi$, that is the expression appearing in the usual Dirac equation.
\newline
\newline
Before proceeding to the next chapter, we mention an important property of the curved gammas. 
\newline
We have defined these matrices as contractions of the tetrad components with the gamma matrices. Since these lasts belong to the space of endomorphisms on $\mathbb{C}^4$, formally, the one form $\Gamma := e^J \gamma_J = \Gamma_\mu dx^\mu \in T_x^*M \bigotimes {\mathbb{C}^4}^* \bigotimes \mathbb{C}^4$, where ${\mathbb{C}^4}^*$ is the dual space of $\mathbb{C}^4$, composed of dual spinors\footnote{This definition is due to the fact that, as can be easily proven, $\bar{\psi}\psi$ is a scalar under $Spin(3,1)$ transformations.} $\bar{\psi} := \psi^\dagger \gamma^0$.
\newline
Therefore, we obtain

\begin{equation}
    \begin{split}
        D_\nu \Gamma_\mu &= \nabla_\nu (e^L_{\hphantom{L}\mu}) \gamma_L -\frac{i}{2}\omega^{IJ}_{\hphantom{IJ}\nu}J_{IJ}e^L_{\hphantom{L}\mu} \gamma_L + \frac{i}{2}e^L_{\hphantom{L}\mu} \gamma_L\omega^{IJ}_{\hphantom{IJ}\nu}J_{IJ} = \\
        & = \partial_\nu e^L_{\hphantom{L}\mu} \gamma_L - \Gamma^\rho_{\hphantom{\rho} \nu \mu} e^L_{\hphantom{L}\rho}\gamma_L -\frac{i}{2}\omega^{IJ}_{\hphantom{IJ}\nu}J_{IJ}e^L_{\hphantom{L}\mu} \gamma_L + \frac{i}{2}e^L_{\hphantom{L}\mu} \gamma_L\omega^{IJ}_{\hphantom{IJ}\nu}J_{IJ} = \\
        & = -\omega^{LI}_{\hphantom{LI} \nu} e_{I \mu} \gamma_L -\frac{i}{2}\omega^{IJ}_{\hphantom{IJ}\nu}J_{IJ}e^L_{\hphantom{L}\mu} \gamma_L + \frac{i}{2}e^L_{\hphantom{L}\mu} \gamma_L\omega^{IJ}_{\hphantom{IJ}\nu}J_{IJ} = \\
        & = -\omega^{LI}_{\hphantom{LI} \nu} e^J_{\hphantom{J} \mu} \eta_{IJ} \gamma_L - \frac{1}{2} \omega^{IJ}_{\hphantom{IJ}\nu}e^L_{\hphantom{L}\mu}(\eta_{IL}\gamma_J - \eta_{LJ}\gamma_I) =\\
        & = -\omega^{LI}_{\hphantom{LI} \nu} e^J_{\hphantom{J} \mu} \eta_{IJ} \gamma_L - \frac{1}{2} \omega^{IL}_{\hphantom{IL}\nu}e^J_{\hphantom{J}\mu}\eta_{IJ}\gamma_L + \frac{1}{2}\omega^{LI}_{\hphantom{LI}\nu}e^J_{\hphantom{J}\mu}\eta_{JI}\gamma_L = \\
        & =  -\omega^{LI}_{\hphantom{LI} \nu} e^J_{\hphantom{J} \mu} \eta_{IJ} \gamma_L + \frac{1}{2} \omega^{LI}_{\hphantom{LI}\nu}e^J_{\hphantom{J}\mu}\eta_{IJ}\gamma_L + \frac{1}{2}\omega^{LI}_{\hphantom{LI}\nu}e^J_{\hphantom{J}\mu}\eta_{IJ}\gamma_L = 0,
    \end{split}
    \label{Cov_Der_Curved_Gamma}
\end{equation}
where we have used eq. \ref{Christ_Symb_Metr_Conn} and the fact that $\gamma_L J_{IJ} = J_{IJ}  \gamma_L + i \eta_{IL}\gamma_J - i \eta_{LJ} \gamma_I$.
\newline
This fact has the direct consequence that $D_\mu \Gamma^\mu = D_\mu( g^{\mu \nu} \Gamma_\nu) = g^{\mu \nu} D_\mu \Gamma_\nu + \nabla_\mu(g^{\mu \nu}) \Gamma_\nu = 0$, since the connection is a metric one. This will be important in the next chapter.

\chapter{The background Spinor Model}\label{Spinor_Model}

Among all models cited in the Introduction, we focused our attention on the one presented by João Magueijo, T.G. Zlosnik and T.W.B. Kibble in the article \cite{Magueijo}.
\newline
The aim of this chapter is to introduce their model, focusing on what will be useful in the next chapters for dealing with cosmological perturbations and with spherically symmetric objects.

\section{The Lagrangian of the Model}

In this section, we follow the derivation presented in the article cited above. There, the authors emphasize that the difficulties of gravity in partaking in quantum field theory lead to consider General Relativity as the symmetry broken phase of gauge theories of groups such as the Poincaré Group or the de Sitter/anti-de Sitter Groups. This fact has the immediate effect that the connection defined on the Spin Bundle naturally takes the form shown in eq. \ref{Spin_Conn_Gamma}.
We see, then, that the choice previously made for mere simplicity regarding this connection is forced by basic principles. 
Moreover, this symmetry breaking leads also to consider a non-zero torsion in the formulation of a theory that concerns spinorial fields.
\newline
Therefore, the theory we are going to formulate must rely on a Langrangian function that, in the gravitational part, contains also terms related to torsion and, in the matter part, contains spinorial fields living on a Spin Bundle endowed with the connection defined in eq. \ref{Spin_Conn_Gamma}.
\newline
\newline
Since we are dealing with spinorial fields, it is more convenient to write the Lagrangian exploiting the one-form formalism, i.e., by making use of tetrads, connections and curvature forms defined on the $SO(3,1)$-bundle. These quantities, by construction, are independent of the chart chosen and dependent only on the reference frame chosen. Hence, the general covariance is guaranteed if we construct terms that are Lorentz-invariant, since the transformations from one reference frame to another belong to the $SO(3,1)$ Group, that is exactly the Special Lorentz Group. Moreover, for simplicity, we restrict our choice to Lagrangians that are polynomials in the fields and their derivatives.
\newline
\newline
We begin by writing the gravitational part of the total action. Since we are dealing with a non-zero torsion, we can take into account other terms besides the usual Palatini term, which we have seen to be the unique possible choice in the absence of torsion. It turns out that there is a unique term that could be considered\footnote{We could have taken into account also a Cosmological Constant term, but we have decided to neglect it for two reasons. First, we would like to see if the spinorial field we are going to discuss can describe also DE. Second, in the case the first purpose is not satisfied, neglecting $\Lambda$ simplifies considerably our calculations. We do not exclude that this contribution can be considered in future works.}, called the Holst term \cite{Holst}. 
\newline
Therefore, given $\Omega_{IJ} := d\omega_{IJ} + \omega_{IL} \wedge \omega^L_{\hphantom{L}J}$ the components of the known curvature form on the $SO(3,1)$-principal bundle, the gravitational part of the action is

\begin{equation}
    \begin{split}
    S_G :=& -\frac{1}{32 \pi G}\int \Bigg( \epsilon_{IJKL} + \frac{2}{\gamma}\eta_{IK}\eta_{JL} \Bigg) e^I \wedge e^J \wedge \Omega^{KL} =\\
    =&-\frac{1}{16 \pi G}\int d^4x \hspace{1mm} \sqrt{-g} \hspace{1mm} \Bigg( R + \frac{1}{2\gamma}\varepsilon^{\alpha \beta \mu \nu} R_{\alpha \beta \mu \nu}\Bigg),
    \end{split}
    \label{Grav_Action_Model}
\end{equation}
where $\varepsilon^{\alpha \beta \mu \nu} = \frac{1}{\sqrt{-g}}\tilde{\epsilon}^{\alpha \beta \mu \nu} = - \frac{1}{\sqrt{-g}}\epsilon^{\alpha \beta \mu \nu}$ is the curved Levi-Civita symbol, $\gamma$ is called the Immirzi parameter and $R_{\alpha \beta \mu \nu} := g_{\alpha \sigma}R^\sigma_{\hphantom{\sigma} \beta \mu \nu} = g_{\alpha \sigma} e^\sigma_{\hphantom{\sigma}I} e^J_{\hphantom{J}\beta}e^K_{\hphantom{K}\mu}e^L_{\hphantom{L}\nu}\Omega^I_{\hphantom{I}JKL}$ are the components of the Riemann tensor defined on the Tangent Bundle.
\newline
It can be proven that all the other possible terms that are polynomials in $\Omega^{KL}$ and $e^I$ are boundary terms, which do not contribute to the equations of motion due to Hamilton Principle.
\newline
\newline
In order to write the matter action, instead, we are going to exploit the Weyl spinorial fields defined in the previous chapter. In fact, since our purpose is to write a Lorentz-invariant action, the more convenient way is to write it using the fundamental parts that transform under this group. Therefore, since the Lorentz Group is linked by the homomorphism $\rho$ to the Spin Group, it is natural to write such a matter action using the irreducible representations of this latter group. These are Weyl spinors. 
\newline
\newline
As anticipated in the previous chapter, if we consider only the kinetic part, the possible combinations of Weyl terms must be separated from each other. Hence, the purpose is to build an action composed of terms made of a single Weyl spinorial field each, that transform as scalars and vectors\footnote{For simplicity, we do not consider tensorial terms.} under Lorentz transformations. It is easy to see that scalar terms composed only by a single Weyl spinorial field cannot be constructed. Therefore, we consider only vectorial ones. From the transformation properties of the gamma matrices in the chosen representation, it can be derived that the combinations of Weyl spinorial fields

\begin{equation}
    \begin{split}
    &K^{I(l)} := \phi^\dagger\bar{\sigma}^I D^{(l)}\phi,\\
    &K^{I(r)} := \chi^\dagger \sigma^I D^{(r)}\chi,
    \end{split}
    \label{Weyl_Terms}
\end{equation}
(where $\chi$ and $\phi$ are respectively the right and left Weyl spinorial fields, $D^{(l),(r)}$ are the covariant derivatives defined on the two Weyl Spin Subbundles, while $\sigma^I := (I, \sigma^i)$ and $\bar{\sigma}^I := (I, -\sigma^i)$), transform as vectors under $SO(3,1)$ transformations.
\newline 
Hence, we can write the matter action using these terms, contracting them properly with the tetrad components. 
\newline
Furthermore, we require that the Lagrangian we are going to build reduces, in flat space, to the Dirac one, as expected from QFT in a Minkowski space-time. Therefore, we take into account only contributions that are linear in $K^{(l),(r)}$.
\newline
Since $K^{(l),(r)}$ are complex and we are looking for a real action, the most general and simple action we can write satisfying all our requirements is $S_{M_{free}} = S_{(l)} + S_{(r)}$, where

\begin{equation}
    \begin{split}
    &S_{(l)} := \int e_{IJKL} e^I \wedge e^J \wedge e^K \wedge \bigg(a_{(l)}\big(K^{L(l)} + K^{L(l)^*} \big) + ib_{(l)} \big(K^{L(l)} - K^{L(l)^*} \big) \bigg), \\
    &S_{(r)} := \int e_{IJKL} e^I \wedge e^J \wedge e^K \wedge \bigg(a_{(r)}\big(K^{L(r)} + K^{L(r)^*} \big) + ib_{(r)} \big(K^{L(r)} - K^{L(r)^*} \big) \bigg),
    \end{split}
    \label{Weyls_Actions}
\end{equation}
with $a_{(l)}$, $a_{(r)}$, $b_{(l)}$ and $b_{(r)}$ real constants. 
\newline
\newline
From the properties of $\bar{\sigma}^I$ and $\sigma^I$ under $Spin(3,1)$ transformations we can obtain their own properties under $\mathfrak{spin}(3,1)$ transformations. These give 

\begin{equation}
    \begin{split}
    &K^{I(l)} +  K^{I(l)^*} = D(\phi^\dagger\bar{\sigma}^I\phi),\\
    &K^{I(r)} + K^{I(r)^*} = D(\chi^\dagger \sigma^I\chi),
    \end{split}
    \label{Weyl_Terms_Complex_Conj}
\end{equation}
where the covariant derivative is that of an element in the $SO(3,1)$-bundle.
\newline
Furthermore, since $\epsilon_{IJKL} \Lambda^I_{\hphantom{I}M} \Lambda^J_{\hphantom{J}N} \Lambda^K_{\hphantom{K}O} \Lambda^L_{\hphantom{L}P} = \epsilon_{MNOP}$ for every $SO(3,1)$ transformation, it is easy to see that, for every $\mathfrak{so}(3,1)$ transformation, 

\begin{equation}
    \begin{split}
        \epsilon_{MJKL}\omega^M_{\hphantom{M}I} + \epsilon_{IMKL}\omega^M_{\hphantom{M}J} + &\epsilon_{IJML}\omega^M_{\hphantom{M}K} + \epsilon_{IJKM}\omega^M_{\hphantom{M}L} = 0\\
        & \Big \Downarrow\\
        D(e_{IJKL}) := de_{IJKL} - \epsilon_{MJKL}\omega^M_{\hphantom{M}I} - &\epsilon_{IMKL}\omega^M_{\hphantom{M}J} - \epsilon_{IJML}\omega^M_{\hphantom{M}K} - \epsilon_{IJKM}\omega^M_{\hphantom{M}L} = 0.
    \end{split}
    \label{Levi-Civita_Cov_Der}
\end{equation}
Therefore, considering that $De^I := de^I + \omega^I_{\hphantom{I}J} \wedge e^J = T^I$, we obtain that

\begin{equation}
    \begin{split}
    &\int e_{IJKL} e^I \wedge e^J \wedge e^K \wedge \bigg(a_{(l)}\big(K^{L(l)} + K^{L(l)^*} \big) + a_{(r)}\big(K^{L(r)} + K^{L(r)^*} \big) \bigg) = \\
    =& \int e_{IJKL} e^I \wedge e^J \wedge e^K \wedge D \big(a_{(l)}\phi^\dagger \bar{\sigma}^L \phi + a_{(r)}\chi^\dagger \sigma^L \chi \big) = \\
    & 3\int e_{IJKL} T^I \wedge e^J \wedge e^K   \big(a_{(l)}\phi^\dagger \bar{\sigma}^L \phi + a_{(r)}\chi^\dagger \sigma^L \chi \big) + \text{boundary term},
    \end{split}
    \label{Torsion_Terms}
\end{equation}
where the boundary term is the total differential of the integrand, which, thanks to Hamilton Principle, gives a null contribution to the equations of motion. So, we can neglect it.
\newline
\newline
Now, the last piece we can add to the total Lagrangian is a potential term given by combinations of the two Weyl spinorial fields. Given $U(\phi, \chi)$ a scalar function, we define

\begin{equation}
    S_U = \frac{1}{4!} \int U(\phi, \chi) e_{IJKL}e^I \wedge e^J \wedge e^K \wedge e^L.
    \label{Potential_action}
\end{equation}
Therefore, the total matter action becomes $S_M = S_{(l)} + S_{(r)} + S_U$.
\newline
However, the action presented has a strong dependence on the representation chosen for the gamma matrices. Hence, we would like to rewrite it in a general form, using Dirac spinorial fields and the general expression for gamma matrices. This can be done by redefining the $\chi$ field in a suitable manner and by exploiting the properties of Dirac spinors.
\newline
In particular, if we define $\chi' := \sqrt{\frac{b_{(r)}}{b_{(l)}}}\chi$ and $\psi = (\phi, \chi')$, it is easy to see that, in the chosen representation, $\bar{\psi}\gamma^ID\psi = \phi^\dagger \bar{\sigma}^ID^{(l)}\phi + \chi'^\dagger\sigma^ID^{(r)}\chi'$. Hence, the matter action takes the form 

\begin{equation}
    \begin{split}
    S_M =& -\frac{1}{12} \int e_{IJKL}e^I \wedge e^J \wedge e^K \wedge  \bigg(i \big( \bar{\psi} \gamma^L D \psi - (D\bar{\psi)}\gamma^L \psi\big) \bigg) +\\
    &-\frac{1}{4}\int e_{IJKL}T^I \wedge e^J \wedge e^K \wedge  \big(\alpha V^L + \beta A^L \big) + \frac{1}{4!} \int U(\psi) e_{IJKL}e^I \wedge e^J \wedge e^K \wedge e^L,
    \end{split}
    \label{Matter_Action}
\end{equation}
where we have assumed that $U$ is a scalar function of $\psi$ and we have defined 

\begin{equation}
    \begin{split}
    &V^L:= \bar{\psi}\gamma^L\psi = \phi^\dagger \bar{\sigma}^L \phi + \chi'^\dagger \sigma^L \chi', \hspace{1cm} A^L:= \bar{\psi}\gamma^5\gamma^L\psi = \phi^\dagger \bar{\sigma}^L \phi - \chi'^\dagger \sigma^L \chi',\\
    &b_{(l)} := -\frac{1}{12}, \hspace{1cm} b_{(l)}\alpha := \frac{1}{2} \bigg(a_{(l)} + a_{(r)}\frac{b_{(l)}}{b_{(r)}} \bigg) \hspace{0.5cm} \text{and} \hspace{0.5cm} b_{(l)}\beta := \frac{1}{2} \bigg(a_{(l)} - a_{(r)}\frac{b_{(l)}}{b_{(r)}} \bigg). 
    \end{split}
    \label{Definitions}
\end{equation}
We see that, with the value chosen for $b_{(l)}$, in flat space-time (i.e., in Special Relativity), we obtain the usual Dirac Action for a particle subjected to the potential $U(\psi)$. In fact,

\begin{equation}
    \begin{split}
    S_M =& -\frac{i}{2} \int \big( \bar{\psi} \Gamma^\mu D_\mu \psi - (D_\mu\bar{\psi)}\Gamma^\mu \psi\big) \sqrt{-g} \hspace{1mm} d^4x
    -\frac{1}{2}\int T^\mu_{\hphantom{\mu}\mu \nu}  \big(\alpha V^\nu + \beta A^\nu \big) \sqrt{-g} \hspace{1mm} d^4x +\\
    & + \int U(\psi) \sqrt{-g} \hspace{1mm} d^4x
    \end{split}
    \label{Matter_Action_Coord}
\end{equation}
(where $V^\nu = \bar{\psi}\Gamma^\nu\psi$, $A^\nu = \bar{\psi}\gamma^5\Gamma^\nu\psi$ and $T^\mu_{\hphantom{\mu} \rho \nu} : = e^\mu_{\hphantom{\mu}I}T^I_{\hphantom{I} \rho \nu}$ are the components of the Torsion tensor defined on the Tangent Bundle), which, considering the case of zero torsion and $g_{\mu \nu} = \eta_{\mu \nu}$, and remembering that in this case $\bar{\psi}\gamma^\mu\psi$ is conserved, reduces to the Dirac action, as requested.
\newline
\newline
In conclusion, we state that the total action of this theory is the sum of the two actions shown in equations \ref{Grav_Action_Model} and \ref{Matter_Action}. However, if one wants a spinorial field model which takes into account also a weak interaction, it is possible to add to this total action a term $S_{int}$ which includes the appropriate gauge-fields suitably combined with the spinorial field. 
\newline
Moreover, we could have added other terms coupling spinor scalars or generic tensors with the metric or the torsion, like $\bar{\psi}\psi T^IT_I$, but, for simplicity, we have not considered them in this model.
\newline
\newline
Before deriving the equations of motion from this action, it is necessary to do a clarification.
\newline
The reader might naturally ask why we did not write our Lagrangian directly in terms of Dirac spinorial fields, since the Dirac Lagrangian naturally separates the two evolution terms for Weyl spinorial fields, as requested before. 
\newline
If we had followed such a procedure, surely we would have lost some information, such as the term proportional to the torsion, which should have been added by hand, without any justification. Instead, formulating the Lagrangian through the irreducible representations of $Spin(3,1)$ ensures that this term arises naturally and that no other information is lost.

\section{The equations of motion}

As usual, we can derive the equations of motion of the dynamical variables by varying the total action with respect to them and assuming its stationarity. For convenience, we work with eq. \ref{Matter_Action}, where $S_M$ is expressed with respect to the orthonormal reference frame.
\newline
\newline
The first equation of motion that we derive is that of the connection $\omega^{KL}$:

\begin{equation}
    \begin{split}
    0 =  \frac{\delta S}{\delta \omega^{KL}} =& \frac{1}{16 \pi G} \bigg( \epsilon_{IJKL} + \frac{2}{\gamma} \eta_{I[K}\eta_{L]J} \bigg)T^I \wedge e^J - \frac{1}{4!}\epsilon_{IMNP} e^I \wedge e^M \wedge e^N \tilde{\epsilon}^{DP}_{\hphantom{DP}KL}A_D + \\
    & + \frac{1}{4} \epsilon_{[K|MNQ} e^M \wedge e^N \wedge e_{|L]}\big( \alpha V^Q + \beta A^Q \big).
    \end{split}
    \label{Eq.Connection}
\end{equation}
As we have shown in the previous chapter, a generic metric connection $\omega^{IJ}$ can be written as the sum of the Levi-Civita connection $\tilde{\omega}^{IJ}$ and a Contorsion Form $C^{IJ}$ such that $C^I_{\hphantom{I}J} \wedge e^J = T^I$. From this one-form we can define what is known as Contorsion Scalar, that is, $C^I_{\hphantom{I}JK} e^K := C^I_{\hphantom{I}J}$.
\newline
We notice that, by solving eq. \ref{Eq.Connection}, we can express it in terms of the spinorial field

\begin{equation}
    \begin{split}
    C_{TLK} = & \frac{4\pi G\gamma^2}{\gamma^2 + 1} \bigg[ \tilde{\epsilon}^D_{\hphantom{D}TLK} \frac{1}{2} \bigg( A_D + \frac{1}{\gamma}\big(\alpha V_D + \beta A_D \big) \bigg) +\\ 
    &- \frac{1}{\gamma} A_{[L}\eta_{T]K} + \alpha V_{[L}\eta_{T]K} + \beta A_{[L}\eta_{T]K} \bigg].
    \end{split}
    \label{Contorsion}
\end{equation}
This last equation is the keystone of the model. In fact, it allows to express the part of the connection related to the torsion in terms of the spinorial bilinears. In this way, the whole model can be recast as a zero-torsion-model with the addition of an effective potential, as we will see.
\newline
\newline
The second equation of motion we can derive is that of the tetrad $e^I$:

\begin{equation}
    \begin{split}
    0 =  \frac{\delta S}{\delta e^{I}} =& -\frac{1}{16 \pi G} \epsilon_{JIKL} e^J \wedge \tilde{R}^{KL} + \frac{i}{4} \epsilon_{IJKL} e^J \wedge e^K \wedge \big( \bar{\psi} \gamma^L \tilde{D}\psi - (\tilde{D}\bar{\psi})\gamma^L \psi \big) +\\
    & + \frac{U}{6} \epsilon_{MJIL} e^J \wedge e^M \wedge e^L - \frac{1}{16 \pi G} \bigg( \epsilon_{JIKL} + \frac{2}{\gamma} \eta_{JK}\eta_{IL} \bigg) C^K_{\hphantom{K}MP}C^{ML}_{\hphantom{ML}Q} e^J \wedge e^P \wedge e^Q + \\
    & - \frac{1}{2}\epsilon_{IKNL} \bigg( \frac{1}{4} \tilde{\epsilon}^{EL}_{\hphantom{EL}AB}A_E + \frac{1}{2}\delta^L_{\hphantom{L}A} \big( \alpha V_B + \beta A_B \big)\bigg) C^{AB}_{\hphantom{AB}P} e^N \wedge e^K \wedge e^P,
    \end{split}
    \label{Eq.Tetrad}
\end{equation}
where the "\textasciitilde"  above some terms means that they have been computed with the Levi-Civita connection (that is, the Contorsion has been factored out).
\newline
\newline
In the end, the last equation of motion we can derive is that of the spinorial field $\psi$:

\begin{equation}
    \begin{split}
    0 = \frac{\delta S}{\delta \bar{\psi}} =& -4i \epsilon_{IJKL} e^J \wedge e^K \wedge e^I \wedge \gamma^L D \psi - 6 \alpha \epsilon_{IJKL}e^J \wedge e^K \wedge T^I \gamma^L \psi +\\
    &- 6(1 + \beta) \epsilon_{IJKL} e^J \wedge e^K \wedge T^I \gamma^5  \gamma^L \psi + \frac{\delta U}{\delta \bar{\psi}}e_{IJKL}e^I \wedge e^J \wedge e^K \wedge e^L.
    \end{split}
    \label{Eq.Dirac_Field}
\end{equation}
\newline
\newline
The last two equations can be rewritten in standard tensor notation (more suitable for doing explicit calculations). While the second equation can be easily reformulated in tensor notation, obtaining

\begin{equation}
    i\Gamma^\mu \tilde{D}_\mu \psi = \frac{\delta W}{\delta \bar{\psi}},
    \label{Dirac_Equation_Model}
\end{equation}
the first requires some manipulation. In fact, in order to obtain an appropriate equation, we have to symmetrize the tensor equation that follows directly from eq. \ref{Eq.Tetrad}. This leads to

\begin{equation}
    \begin{split}
        \frac{1}{8 \pi G}\tilde{G}_{\mu \nu} = -\frac{i}{4} \big( \bar{\psi} \Gamma_{(\mu} \tilde{D}_{\nu)} \psi - \tilde{D}_{(\nu} \bar{\psi}\Gamma_{\mu)}\psi \big) +& \frac{i}{2} \big( \bar{\psi} \Gamma^{\lambda} \tilde{D}_{\lambda} \psi - \tilde{D}_{\lambda} \bar{\psi}\Gamma^{\lambda}\psi \big)g_{\mu \nu} - W g_{\mu \nu},\\
        \end{split}
    \label{Einstein_Equations_Model}
\end{equation}
where $W$ is defined as

\begin{equation}
    W := U + \frac{3 \pi G \gamma^2}{2(1 + \gamma^2)} \bigg[ \bigg( 1 - \beta^2 + \frac{2}{\gamma}\beta \bigg) A_I A^I - \alpha^2 V_I V^I - 2 \alpha \bigg( \beta - \frac{1}{\gamma} \bigg)A_I V^I \bigg].
    \label{Effective_Potential}
\end{equation}
This equation can be further simplified by using the Dirac equation \ref{Dirac_Equation_Model}, which, assuming that $U$ is a function only of $\bar{\psi}\psi$, leads to

\begin{equation}
    \frac{1}{8 \pi G}\tilde{G}_{\mu \nu} = -\frac{i}{4} \big( \bar{\psi} \Gamma_{(\mu} \tilde{D}_{\nu)} \psi - \tilde{D}_{(\nu} \bar{\psi}\Gamma_{\mu)}\psi \big) + \bigg( W + \frac{\partial U}{\partial (\bar{\psi}\psi)}\bar{\psi}\psi -2U \bigg) g_{\mu \nu}.
    \label{Ein_Equations_Model_Coord}
\end{equation}
Before proceeding with the analysis of the model, we would like to highlight two important results already obtained by it.
\newline
The first is that eq. \ref{Dirac_Equation_Model} is precisely the Dirac equation in curved space-time. It is easy to see that, in a flat manifold, if torsion is neglected and $U = m \bar{\psi}\psi$, it reduces to the usual Dirac equation.
\newline
The second is that, as anticipated before, the model has been recast as an effective model without torsion and with the addition of the effective potential $W$. This has been made possible by eq. \ref{Contorsion}, which states that torsion effects on the connection are related to the axial and vector current densities of the spinorial field. The influence of torsion on the equations of motion for the metric, then, can be reformulated in terms of their mutual and self-interaction. Therefore, if the spinorial field did not interact with itself, there would not be contributions due to torsion (this is the case of $\gamma = 0$).
\newline
\newline
Clearly, the case $\gamma = 0$ is not the unique case possible. There are many other cases obtained by varying the parameters $\alpha$ and $\beta$. For a full account of all the possibilities we refer the reader to \cite{Magueijo}.

\section{Classical spinorial fields}

As we remarked at the beginning of this chapter, the spinorial field that appears in the equations is a quantum field. Hence, eq. \ref{Dirac_Equation_Model} must be regarded as an operator equation. 
\newline
Therefore, also eq. \ref{Ein_Equations_Model_Coord} must be regarded in the same way. Nevertheless, we can assume that the classical gravitational field that we observe satisfies this equation with the terms on the right-hand side substituted with their expectation values. 
\newline
However, as remarked in \cite{Magueijo}, in \cite{Dolan} it is stated that in the first-order formulation of gravity there is an ambiguity in taking expectation values, since it would be probably more natural to take them in eq. \ref{Eq.Connection}. This, then, would yield an equation where the classical gravitational field is sourced by terms like $\langle \bar{\psi} \rangle \langle \psi \rangle$ and not $\langle \bar{\psi} \psi \rangle$, as assumed earlier and as expected from the second-order formalism with four-fermion interaction.
\newline
Since these two expectation values in general are not equal, in order to circumvent any ambiguity, the authors decided to consider what are known as \textit{classical spinorial fields}. These are quantum fields described by a spinorial operator $\hat{\psi}$ which are in a quantum state where $\langle f(\hat{\psi}) \rangle \approx f(\langle\hat{\psi}\rangle)$ for every function $f$. In this way, it is easy to see that the above ambiguity is bypassed. 
\newline
As stated in \cite{Picon}, this kind of fields can be treated as Dirac spinors.
\newline
It is straightforward to notice that, since quarks and leptons are not classical on cosmological scales, the fields we are considering must be independent from the fields of the Standard Model. If this were not the case, then we could not make the assumption of a classical spinorial field.
\newline
\newline
This assumption considerably simplifies the equations. In fact, Dirac spinors satisfy what are known as Pauli-Fierz relations:

\begin{equation}
\psi \bar{\chi} = \frac{1}{4}\bigg[\bar{\chi}\phi I - \big(\bar{\chi}\gamma_L\psi \big)\gamma^L + \frac{1}{2}\big(\bar{\chi}\sigma_{LS}\psi \big)\sigma^{LS} + \big(\bar{\chi}\gamma^5 \gamma_L\psi \big)\gamma^5 \gamma^L + \big(\bar{\chi}\gamma^5\psi \big)\gamma^5 \bigg],
    \label{Pauli-Fierz}
\end{equation}
where $\sigma_{LS} := \frac{i}{2}[\gamma_L, \gamma_S]$. These can be used to obtain the values of $\langle A_I A^I\rangle = \langle A_I \rangle \langle A^I\rangle$ and the other contractions in $W$. After some calculations, it turns out that

\begin{equation}
    \langle A_I A^I\rangle = - \langle V_I V^I\rangle = (E^2 + B^2),
    \label{Value_Contractions}
\end{equation}
where $E := \langle \bar{\psi}\psi \rangle$ and $B:= -i\langle \bar{\psi}\gamma^5\psi \rangle$, while $\langle A_I V^I\rangle = 0$.
\newline
As a consequence, $W$ becomes

\begin{equation}
    W = U + \frac{3 \pi G \gamma^2}{2(\gamma^2 + 1)}\bigg( 1 - \beta^2 + \alpha^2 + \frac{2}{\gamma} \beta\bigg) \big( E^2 + B^2 \big) := U + \xi \big( E^2 + B^2 \big),
    \label{W_New}
\end{equation}
where $\xi$, now, includes all possible torsion effects and it can be positive, negative or zero.
\newline
\newline
Before applying the model to the cosmological background case, however, it is natural to ask ourselves what physical states satisfy the condition that characterizes a \textit{classical spinorial field}. 
\newline
At this point a problem arises. In fact, the vacuum expectation value of some quantities, like $\bar{\psi}\psi$, diverges. Hence, a more correct definition for a classical spinorial field would be in terms of the renormalized expectation values:

\begin{equation}
    \langle...\rangle_{ren} := \bra{s}...\ket{s} - \bra{0}...\ket{0}.
    \label{Renorm_Expect}
\end{equation}
After this correction, we can state that a spinorial field, in order to be called \textit{classical}, in the case of the density $\bar{\psi}\psi$ for example, must satisfy a relation of this kind:

\begin{equation}
    \Bigg| \frac{\langle\bar{\psi}\psi \rangle_{ren} - \langle\bar{\psi} \rangle_{ren} \langle\psi \rangle_{ren}}{\langle\bar{\psi} \rangle_{ren} \langle\psi \rangle_{ren}}\Bigg| \ll 1.
    \label{Classical_Spinors_Cond}
\end{equation}
Since, in the bosonic case, the states that satisfy this condition have large occupation numbers, while, in the fermionic case, the largest occupation number of a state is one, due to Pauli exclusion principle, it is believed that there is no fermionic physical state that satisfy this relation.
\newline
This is not true. 
\newline
Let us suppose to have a free spinorial field in flat space-time and let us define the physical state

\begin{equation}
    \ket{s} := A\ket{0} + B\ket{1},
    \label{State_Class}
\end{equation}
where $A, B$ are two complex numbers s.t. $|A|^2 + |B|^2 = 1$, $\ket{0}$ is the vacuum and $\ket{1}$ is the one-fermion state. Considering that, in the canonical quantization, a spinorial field of this kind can be expanded in terms of creation and annihilation operators as

\begin{equation}
    \psi := \sum_k u_k a_k + v_k b^\dagger_k
    \label{Quantum_Field}
\end{equation}
(where $u_k$ and $v_k$ are normalized Dirac spinors, i.e., $\bar{u}_ku_k = -\bar{v}_kv_k = 1$ and $\bar{u}_kv_k = \bar{v}_ku_k = 0$), we can immediately state that condition \ref{Classical_Spinors_Cond}, in this case, translates into

\begin{equation}
    \Bigg| \frac{|B|^2 - |A|^2 |B|^2}{|A|^2 |B|^2}\Bigg| \ll 1 \Longrightarrow |B|^2 \ll 1.
    \label{Classic_Spinors_Cond_State}
\end{equation}
This condition is perfectly realizable. In particular, this result implies that a free fermionic field in flat space-time can be treated classically if its state differs very little from the renormalized vacuum. We can suppose that a condition of this kind can be found also in our case\footnote{As we will see, the solution for the cosmological background in the case $U(\bar{\psi}\psi) = m \bar{\psi}\psi$ and zero torsion is proportional to a free spinorial field in flat space-time. Therefore, we can assert that the quantum state in which this spinorial field must be in order to be called \textit{classical}, in this case, is exactly given by condition \ref{Classic_Spinors_Cond_State}.}. 
\newline
Furthermore, it could be supposed that a \textit{classical spinorial field} does not describe only a single fermion in this kind of quantum state. In fact, there exist some cases where the behavior of interacting fermions in some quantum states is well described by a single, relativistic, classical scalar field. Hence, we could think that some systems of strongly interacting spinorial fields can be described by a single, relativistic, \textit{classical spinorial field}. 
\newline
\newline
In conclusion, we say that, in all the cases we are going to study, the gravitational field is sourced by the renormalized expectation values of the spinor bilinears, which in turn are formed by renormalized expectation values of the classical spinorial field. Hence, considering the analysis made so far, our calculations aim at understanding how states of matter well described by a classical spinorial field influence the geometry of our Universe.

\section{Cosmological equations}\label{Cosmo_Eq_Sec}
Clearly, the simplest case to which we can apply the model presented is the cosmological one. This section is crucial for the whole work presented in this master's thesis, since it sets the basis for what we are going to do in the next chapters. 
\newline
\newline
As it is stated in \cite{Magueijo}, due to the presence of the space-like current $A^I$ which identifies a preferred spatial direction, one might think that equations \ref{Dirac_Equation_Model} and \ref{Ein_Equations_Model_Coord} do not admit a maximally symmetric solution such as the FLRW one. However, it has been proven in \cite{Isham} that such a solution exists as long as the spatial curvature $K$ is zero\footnote{This assumption could seem quite restrictive. However, we remember that observations \cite{Planck} give $\Omega_{K0} = 0.0007 \pm 0.0019$. Therefore, all the results obtained from these calculations can be applied to our Universe with a very good approximation.}.
\newline
\newline
We remember that the FLRW metric with $K = 0$ is given by the line element

\begin{equation}
    ds^2 = -dt^2 + a^2(t) \delta_{ij}dx^idx^j.
    \label{FLRW_metric}
\end{equation}
As we remarked in Chapter \ref{Spinor_Bundles}, in order to find the components of the connection defined on the Spin Bundle, it is necessary to choose a reference frame and then solve eqs. \ref{Structure_Eqs}. Since the connection we want to compute is the Levi-Civita one, we have to impose $T^I = 0$. 
\newline
The simplest tetrad we can choose, in this case, is 

\begin{equation}
    \begin{split}
    & e^0 = dt, \\
    & e^i = a(t)dx^i.
    \end{split}
    \label{Tetr}
\end{equation}
This gives $\omega^{0i} = He^i$ (the other components being zero). Hence,  eq. \ref{Dirac_Equation_Model} becomes

\begin{equation}
    \gamma^0 \bigg(\dot{\psi} + \frac{3}{2}H \psi \bigg) = -i \frac{\delta W}{\delta \bar{\psi}},
    \label{Dirac_Eq_Model_Cosmo}
\end{equation}
while equations \ref{Ein_Equations_Model_Coord} become

\begin{equation}
    \begin{split}
    &H^2 = \frac{8 \pi G}{3}W =: \frac{8 \pi G}{3}\rho_{eff},\\
    &-2 \frac{\ddot{a}}{a} - H^2 = 8 \pi G [U'E - U + \xi(E^2 + B^2)] =:  8 \pi G p_{eff}.
    \end{split}
    \label{Eins_Eqs_Model_Cosmo}
\end{equation}
The Dirac equation and its complex conjugate can be used to obtain some equations for the quadratic invariants of the spinorial field. More in detail, they lead to

\begin{equation}
    \begin{split}
    &\dot{E} + 3HE = 4 \xi B A^0,\\ 
    &\dot{B} + 3HB = -4 \xi E A^0 -2 U'A^0,\\
    & \dot{A^0} + 3HA^0 = 2U'B,\\
    &\dot{V}^0 + 3HV^0 = 0.
    \end{split}
    \label{Quadratic_Inv_Eqs}
\end{equation}
It is easy to see that, if we multiply the first equation by $E$, the second by $B$ and the third by $A^0$ and we sum them all together, we obtain the integral of motion

\begin{equation}
    a^6[E^2 + B^2 + (A^0)^2] = const =:M^2.
    \label{Integral_Motion}
\end{equation}
This result is independent of the form of $U$ chosen and can be used to obtain the explicit expression of the scale parameter $a$ as a function of the cosmic time $t$.

\section{The DM case and its bouncing solution}\label{DM_Bounce}
In general, eqs. \ref{Quadratic_Inv_Eqs} cannot be solved analytically. Therefore, we have to make some assumptions on the spinorial field under analysis.
\newline
\newline
The simplest assumption would be that of a Weyl spinorial field. However, it is quite straightforward to see that these fields have no effect on the Einstein equations in this case, since they produce vanishing $E$ and $B$. Therefore, if there are no other matter components, they lead to Minkowski space-time.
\newline
Hence, we must consider different types of spinorial fields.
\newline
\newline
The second simplest assumption is $A^0 = 0$, since, if $U' \neq 0$, this choice leads to $B = 0$, as can be immediately read from the third of the equations \ref{Quadratic_Inv_Eqs}. In our representation for gamma matrices this translates into $|\phi|^2 = |\chi|^2$, i.e., solutions without parity violation.
\newline
The immediate consequence of this choice can be derived from eq. \ref{Integral_Motion}, which gives $E = \frac{M}{a^3}$.
\newline
The function that is left to be fixed is $U$. 
\newline
Since the aim of the present work is describing DM, as it will be clear soon, the right choice is $U(E) = mE$. This leads to

\begin{equation}
    \begin{split}
    & W = mE + \xi E^2,\\
    & p_{eff} = \xi E^2.
    \end{split}
    \label{DM_backCond_B0}
\end{equation}
We immediately see that, in the case $\xi = 0$, the Einstein equations give $a(t) \propto t^{\frac{2}{3}}$, which is the same behavior that the scale parameter follows in presence of DM. This justifies the choice made.
\newline
In the case of $\xi < 0$, instead, the first of the eqs. \ref{Eins_Eqs_Model_Cosmo} gives

\begin{equation}
    a(t) = \bigg[ M \bigg( \frac{|\xi|}{m} + \frac{3mt^2}{4} \bigg)\bigg]^{\frac{1}{3}}.
    \label{a_bounce}
\end{equation}
As can be seen in fig. \ref{Boucing_Parameter}, the behavior of the scale parameter in this case is peculiar. In fact, when $t$ goes to zero, $a$ does not go to zero, as in standard cosmology, but approaches a finite value. Hence, torsion causes a bounce of the scale parameter, as anticipated in the Introduction\footnote{For more detailed comments on this case and an analysis of the case where $\xi > 0$ or of the cases where $B$ and $A^0$ are different from zero we refer the reader to the original article.}.

\begin{figure}[H]
\centering
  {\includegraphics[width=0.65
\linewidth]{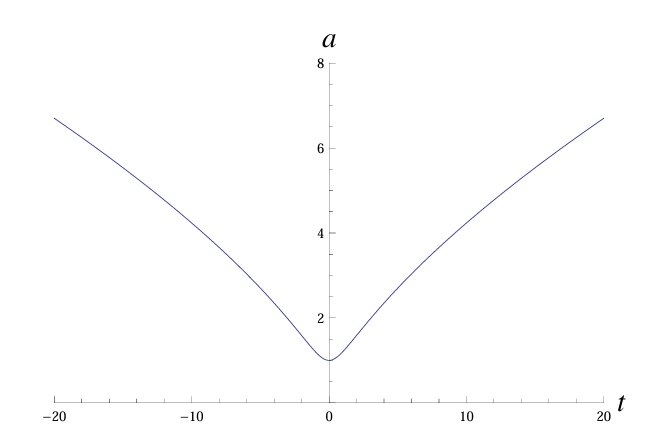}}
  \caption{Behavior of the scale parameter in case $A^0 = B = 0$, $U(E) = mE$ and $\xi < 0$.
  \newline
  (Image taken from \cite{Magueijo})}
  \label{Boucing_Parameter}
\end{figure}

\section{The DE case}
In the previous section, we analyzed the case where $U(E) = mE$, and we have seen that it well-describes DM.
\newline
We would like to see if this model can also describe DE by means of a suitable choice of $U$. The condition to satisfy in order for the spinorial field to act as DE is $p_{eff} = -\rho_{eff}$. 
\newline
This translates into

\begin{equation}
    \begin{split}
    &U'E = - 2 \xi \big( E^2 + B^2\big)\\ 
    &\hphantom{ahdkfkfk} \Downarrow \\
    & U = - \xi E^2-2\xi B^2 \ln|E| + C,
    \end{split}
    \label{Pot_DE}
\end{equation}
where $C$ is an integration constant, related to the cosmological constant $\Lambda$, as we will see shortly.
\newline
From this expression for $U$ we can immediately derive the expressions for $W$ and $\mathcal{E} \equiv p_{eff}$:

\begin{equation}
    \begin{split}
    & W = C + \xi B^2 \big( 1 - 2 \ln|E|\big),\\
    & \mathcal{E} = -C - \xi B^2 \big( 1 - 2 \ln|E|\big) = - W.
    \end{split}
    \label{DE_backCond}
\end{equation}
The reader should notice that, in these calculations, we considered $B$ to be constant because $U$ is a function only of $E$.
\newline
\newline
As in the previous case, we can consider $A^0 = 0$. This, in the presence of torsion, implies $B = 0$. Therefore, the first of the Einstein equations becomes

\begin{equation} 
    H^2 = \frac{8 \pi G}{3}W = \frac{8 \pi G}{3}C \Longrightarrow a(t) \propto e^{\sqrt{\frac{8 \pi G}{3}C}t}.
    \label{a_DE_B0}
\end{equation}
We immediately see that the result obtained in this case is the same that would be obtained by considering the usual cosmological constant $\Lambda$ without the addition of any matter component. From this we conclude that $C = \rho_{\Lambda} := \frac{\Lambda}{8 \pi G}$.
\newline
\newline
The case where $A^0 \neq 0$ and $\xi \neq 0$, instead, leads to a contradiction. In fact, after some calculation, we find that $E$, $B$ and $A^0$ must be constants. However, they lead to an expression for the scale parameter which has the same behavior as that in eq. \ref{a_DE_B0}. This contradicts eq. \ref{Integral_Motion}.
\newline
On the other hand, it is immediate to notice that, in the case of $A^0 \neq 0$ and $\xi = 0$, the results obtained for $W$, $\mathcal{E}$ and $a$ are the same as those obtained in the case with torsion and with $A^0 = 0$, without the occurrence of any contradiction. 
\newline
\newline
In conclusion, we can state that a classical spinorial field well-describes DE only in two cases: the zero-torsion case and the $A^0 = 0$ case. In both of them the constant $C$ that enters in the expression of $U$ is related to the cosmological constant $\Lambda$. 
\newline
This means that, if also its cosmological perturbations behave in the same way as those of DE do, this model can be explored in future works as a valid alternative to the well-known geometrical description of DE given by $\Lambda$.

\section{The Perfect Fluid SET and the RF at rest}\label{PF_SET_RF_Rest}
Before dealing with cosmological perturbations, we make some considerations about the SET of the model presented. Furthermore, we show, in the case of DM without torsion, how the components of the spinorial field with respect to the Reference Frame (RF) specified by the tetrad in eq. \ref{Tetr} are almost entirely determined by the fact that the Coordinate RF is that with respect to which the "spinorial fluid" is at rest, since this imposes some constraints on the spinorial field in question.
\newline
\newline
From eq. \ref{Ein_Equations_Model_Coord} we can state that the SET of the spinorial fluid we are considering is\footnote{From now on, since all torsion terms are included in the parameter $\xi$, we take for granted that every covariant derivative is computed using the Levi-Civita connection. Hence, we avoid to write "\textasciitilde" wherever it was written before.} 

\begin{equation} 
    T_{\mu \nu} = -\frac{i}{4}\Big[\bar{\psi} \Gamma_{(\mu} D_{\nu)} \psi - D_{(\nu} \bar{\psi} \Gamma_{\mu)} \psi \Big] + g_{\mu \nu} \mathcal{E},
    \label{Spinor_SET}
\end{equation}
where we have defined, as in the previous section, $\mathcal{E} := W + U'E - 2U$.
\newline
We have seen in Section \ref{Cosmo_Eq_Sec} that, in the cosmological background case, this tensor is diagonal. Therefore, it could come quite natural to identify it with that of a perfect fluid\footnote{We will see in Chapter \ref{SET_Decomp} that this assumption is not valid in the general case and that the decomposition of the SET of a spinorial fluid needs some extra care.} in its rest-frame.
\newline
We remember that the SET of a perfect fluid is 

\begin{equation} 
    T_{\mu \nu} = (\rho + P)u_\mu u_\nu + g_{\mu \nu}P.
    \label{PF_SET}
\end{equation} 
Since, in the cosmological background case, we have that $W \equiv \rho$ and $\mathcal{E} \equiv p$, after a suitable choice of $u_\mu$, the SET of eq. \ref{Spinor_SET} can be naturally rewritten in the form of eq. \ref{PF_SET} as

\begin{equation} 
    T_{\mu \nu} = (W + \mathcal{E})u_\mu u_\nu + g_{\mu \nu} \mathcal{E}.
    \label{Spinor_SET_PF}
\end{equation}
Another evidence for this identification can be found in the trace of the tensor.
In fact, the trace of the perfect fluid SET is $T^{\mu}_{\hphantom{\mu} \mu} = 3P - \rho$, while, for the spinorial fluid, using the Dirac equation, we obtain 

\begin{equation}
    \begin{split}
    T^{\mu}_{\hphantom{\mu} \mu} &= -\frac{i}{2}\Big[\bar{\psi} \Gamma^{\mu} D_{\mu} \psi - D_{\mu} \bar{\psi} \Gamma^{\mu} \psi \Big] +  4\mathcal{E} = \\
    &= -\frac{1}{2} \Big[ \bar{\psi} \frac{\delta W}{\delta \bar{\psi}} + \Big( \frac{\delta W}{\delta \bar{\psi}}\Big)^\dagger \Gamma^0 \psi \Big] + 4 \mathcal{E} = - (\mathcal{E} + W) + 4\mathcal{E} = 3 \mathcal{E} - W.
    \end{split}
    \label{Spinor_SE_trace}
\end{equation}
However, in order to complete the identification of the spinorial SET with that of a perfect fluid we need to find a suitable four-velocity $u^\mu$. It must be a properly normalized time-like vector related to the fluid. Among the several possible choices, the simplest and the most convenient is the curved vector current density $V^\mu = \bar{\psi} \Gamma^\mu \psi$, since, as can be proven using the results obtained at the end of the previous chapter and the Dirac equation, it is also a conserved current: $\nabla_\mu V^\mu = 0$.
\newline
Therefore, given that $V^\mu V_\mu = -(E^2 + B^2)$, we can identify
\begin{equation} 
    u^\mu := \frac{V^{\mu}}{\sqrt{E^2 + B^2}}.
    \label{u_mu}
\end{equation}
This decision, at present, might seem quite arbitrary, but, in the next chapters, it will prove to be the most clever choice to make in order to easily decompose the SET of a spinorial fluid.
\newline
\newline
As we have anticipated at the beginning of this section, this choice allows to almost entirely determine the spinorial field in the DM case with zero torsion.
\newline
In fact, if the assumption of a perfect fluid in its rest-frame is correct, as it will prove to be, we must have that, in the frame specified by our choice of coordinates, $u^i = 0$. Hence, $V^i = 0$.
\newline
Furthermore, in the DM case with zero torsion, eq. \ref{Dirac_Eq_Model_Cosmo} can be solved directly, obtaining 

\begin{equation}
    \psi(t) = a^{-\frac{3}{2}}e^{-imt\gamma^0}\psi_0,
    \label{Solution_DM_Cosmo_T0}
\end{equation}
where $\psi_0$ is a constant spinor set by initial conditions.
\newline

\noindent
This leads immediately to

\begin{equation}
    V^i = \bar{\psi}\Gamma^i\psi = a^{-2}\big( \bar{\psi}_0\gamma^i\psi_0 \cos(2mt) + i \bar{\psi}_0 \gamma^0\gamma^i\psi_0 \sin(2mt) \big). 
    \label{V_i}
\end{equation}
Therefore, we see that, since the spinor $\psi_0$ has 8 real independent components and the condition $V^i = 0$ leads to 6 independent real conditions specified by $\bar{\psi}_0\gamma^i\psi_0 = 0$ and $\bar{\psi}_0 \gamma^0\gamma^i\psi_0 = 0$, we determine all but two components of the spinorial field in the RF specified by eq. \ref{Tetr}.
\newline
\newline
The remaining components are arbitrary and can be fixed by choosing the direction of the axial current density in the Coordinate RF. In fact, since $A^\mu$ and $V^\mu$ are orthogonal four-vectors, we can state that $A^0 = 0$. Furthermore, it can be proven that requiring $A^{2,3} = 0$, for example, is equivalent to imposing two real conditions on $\psi_0$. This determines uniquely $\psi(t)$ in the RF specified by the chosen tetrad\footnote{We remember that specifying some components with respect to an $SO(3,1)$ RF is equivalent to specifying some components with respect to a $Spin(3,1)$ RF, as remarked in Chapter \ref{Spinor_Bundles}.}.
\newline
\newline
In addition, remembering the discussion made in Section \ref{DM_Bounce}, since $A^0 = 0$, we can state that, in the frame chosen, $\psi(t)$ does not violate parity. This could be expected, since the Lagrangian of the model, in the absence of torsion, does not violate parity.
\newline
\newline
Clearly, if it had not been possible to choose the initial conditions of the spinorial field in order that $V^i = 0$, we could not have assumed that the spinorial background SET could be rewritten as that of a Perfect Fluid. Therefore, it is more correct to say that the spinorial background SET can be written in the form given by eq. \ref{PF_SET} if and only if the initial conditions can be chosen in order to satisfy $V^i = 0$\footnote{We have shown it for DM without torsion, but it can be also shown for DE, with or without torsion.}.
\newline
Hence, in Chapter \ref{SET_Decomp}, we will have to prove that the spinorial SET, in the case of the cosmological background, reduces to that of a Perfect Fluid if and only if, in the Coordinate RF, the vector current density has only one non-vanishing component: $V^0$.

\chapter{Cosmological Perturbations}\label{Cosm_Perturb}

In the previous chapter, we introduced a model for a spinorial field in a curved space-time, and we presented some results obtained by it in the context of the cosmological background. In particular, we focused on two cases of interest: DM and DE\footnote{Naturally, there exist several other possibilities that have not been taken into account and could be studied in future works.}.
\newline
In this chapter, our aim is to analyze these cases from the perspective of scalar cosmological perturbations. 
\newline
The reasons behind this choice are mainly two. First, the study of perturbations is useful in order to verify the stability of the results obtained at the background level. Then, if the perturbations related to these results behave in the same way as those of the known DM and DE, we can state that the spinorial field we are studying is a good candidate for DM or DE. 

\section{Perturbations Theory}\label{Perturb_Th}

Before exposing our findings, we believe that the reader should have an overview about how perturbations of the background metric and of the other background quantities are treated in Cosmology. 
\newline
In doing this, we follow \cite{Piattella}. Therefore, although the results in Section \ref{Scal_Pert_SET} are presented as functions of the cosmic time, in all the other sections of this chapter we decided to adopt the conformal time $\eta$, defined by

\begin{equation}
    \eta := \int_0^t \frac{dt'}{a(t')}.
    \label{Conf_Time}
\end{equation}
\newline
\newline
We remember that the FLRW metric with $K = 0$, written using the conformal time, is given by the line element

\begin{equation}
    ds^2 = a^2(\eta)(-d\eta^2 + \delta_{ij}dx^idx^j).
    \label{FLRW_metric_Conf_Time}
\end{equation}
In order to study the perturbations of this metric tensor, it is natural to define a new metric that differs very little from it. That is, given the background metric $\bar{g}$, we define
\begin{equation}
    g(p) := \bar{g}(p) + \delta g(p) \hspace{1cm} \forall p \in \mathcal{M},
    \label{Metric_Perturb}
\end{equation}
where $|\delta g(p)| \ll |\bar{g}(p)|$ and $\mathcal{M}$ is the manifold on which $g$ is defined.
\newline
The problem is that, mathematically, we cannot define two metrics on the same manifold, but we need two different manifolds, the physical one (where $g$ is defined) and the background one (where $\bar{g}$ is defined). This gives rise to a problem. If the two metrics are defined on different manifolds, the quantity $\delta g(p) := g(p) - \bar{g}(\hat{p})$ (where $\hat{p}$ is a point on the background manifold) is not well-defined, since the difference between two tensors defined on distinct manifolds is not well-defined. Therefore, we need a map (more precisely, a diffeomorphism) between the points of the two manifolds in order to make this difference well-defined. This map, besides, would allow to use the same charts for the two manifolds and, hence, the same systems of coordinates. The problem is that this diffeomorphism is arbitrary and establishes a gauge, giving rise to the so-called problem of the gauge, which we will see in the next section.
\newline
Once we have chosen this map, the quantity\footnote{More precisely, since we are dealing with tensors, we should consider the pullback of the background metric with respect to the map chosen. That is, given the map $\mathcal{D}$, the correct definition would be $\delta g(p) := g(p) - (\mathcal{D}^{-1})^*\bar{g}(p)$.} $\delta g(p) := g(p) - \bar{g}(p)$, $\forall p \in \mathcal{M}$, becomes well-defined. 
\newline
Naturally, we assume that the metric defined on the physical manifold is such that it is always possible to choose a gauge in which $|\delta g(p)| \ll |\bar{g}(p)|$.
\newline
Therefore, considering that, chosen a system of coordinates on the background manifold, we can use the same system in the physical one, locally the physical metric assumes the form

\begin{equation}
    g_{\mu \nu} = a^2(\eta) 
    \begin{pmatrix}
        -[1 + 2 \psi(\eta, \mathbf{x})] & & & w_i(\eta, \mathbf{x}) & \\
        \\
        w_i(\eta, \mathbf{x}) & & & \delta_{ij} [1 + 2 \phi(\eta, \mathbf{x})] + \chi_{ij} & \\
    \end{pmatrix},
    \hspace{1cm} \delta^{ij}\chi_{ij} = 0,
    \label{Pert_Metric}
\end{equation}
where $|\phi|, |\psi|, |w_i|, |\chi_{ij}| \ll 1$.
\newline
Consequently, since we are interested in first order perturbations, in all the computations presented in this and the next chapters, we neglected all non-linear terms formed by these quantities.
\newline
Furthermore, it can be easily proven that $\psi$ and $\phi$ transform as scalars under $SO(3)$, $w_i$ as a vector and $\chi_{ij}$ as a tensor. Hence, they can be viewed respectively as scalars, a vector and a tensor in the three-dimensional Euclidean space.
\newline
\newline
Using the local form of the physical metric presented in eq. \ref{Pert_Metric}, which can be written as $g_{\mu \nu} = a^2(\eta)(\eta_{\mu \nu} + h_{\mu \nu})$, we can compute the perturbed mixed Einstein Tensor.
\newline
In particular, its components are:

\begin{equation}
    2a^2 \delta G^0_{\hphantom{0}0} = -6 \mathcal{H}^2 h_{00} + 4 \mathcal{H}\partial^{k}h_{k0} + 2 \mathcal{H} (\delta^{kj}h_{kj})' + \nabla^2(\delta^{kj}h_{kj}) - \partial^k \partial^l h_{kl},
    \label{Eins_Tens_00}
\end{equation}

\begin{equation}
    2a^2 \delta G^0_{\hphantom{0}i} = 2 \mathcal{H} \partial_i h_{00} + \nabla^2 h_{0i}) - \partial_i \partial^k h_{k0} + \partial_i (\delta^{kj}h_{kj})' - \partial^k h_{ki}',
    \label{Eins_Tens_0i}
\end{equation}

\begin{equation}
    \begin{split}
        2a^2 \delta G^i_{\hphantom{i}j} = & \bigg[ - 4 \frac{a''}{a}h_{00} - 2 \mathcal{H}h_{00}' - \nabla^2 h_{00} + 2 \mathcal{H}^2 h_{00} - 2 \mathcal{H} (\delta^{kj}h_{kj})' + \nabla^2 (\delta^{kj}h_{kj}) +\\
        &- \partial^k \partial^l h_{kl} + 2 \partial^k h_{k0}' + 4 \mathcal{H} \partial^k h_{k0} - (\delta^{kj}h_{kj})'' \bigg] \delta^i_{\hphantom{i}j} + \partial_i \partial_j h_{00} - \nabla^2 h_{ij} +\\
        & + \partial_j \partial^k h_{ki} + \partial_i \partial^k h_{kj} - \partial_i \partial_j (\delta^{kj}h_{kj}) + h_{ij}'' + 2 \mathcal{H}h_{ij}' - (\partial_j h_{0i}' + \partial_i h_{0j}') + \\
        & - 2 \mathcal{H}(\partial_j h_{0i} + \partial_i h_{0j}),
    \end{split}
    \label{Eins_Tens_ij}
\end{equation}
where $\partial^k := \delta^{ik}\partial_i$, $\nabla^2 := \delta^{ij}\partial_i \partial_j$ is the Euclidean Laplacian, $'$ denotes the derivative with respect to the conformal time and $\mathcal{H} := \frac{a'}{a}$ is the conformal Hubble parameter.
\newline
Obviously, the same perturbation method must be applied to the SET, in order to have the complete perturbed Einstein equations. However, this requires some extra care, as we will see in the next sections.

\section{The problem of the gauge}
This section aims to shed light on the problem of the gauge cited in the previous section. 
\newline
\newline
We stated that, in order to correctly define the perturbations of the FLRW metric, we need a map which links the points on the physical manifold to the points on the background one. This map is arbitrary, that is, it depends on our choice. Therefore, quantities that depend on the map chosen cannot be physical. This leads to the need of defining gauge-invariant quantities, which, by definition, have a physical meaning. In order to construct these quantities, however, we need to know how the quantities defined in eq. \ref{Pert_Metric} transform under a change of the map chosen.
\newline
\newline
Being $\mathcal{D}$ the map between the background manifold and the physical one and $\varphi$ a chart on the background manifold, in terms of local sections, the definition of the components of the perturbed metric is mathematically\footnote{In the sense that this is the formal definition of the local expressions of the components of the perturbed metric given in eq. \ref{Pert_Metric}.} described by

\begin{equation}
    \begin{split}
        (\mathcal{D} \circ \varphi^{-1})^*\delta g_{\mu \nu}(\mathcal{D} \circ \varphi^{-1}(x)) := (\mathcal{D} \circ \varphi^{-1})^*g_{\mu \nu}(&\mathcal{D} \circ \varphi^{-1}(x)) - (\mathcal{D} \circ \varphi^{-1})^*(\mathcal{D}^{-1})^*\bar{g}_{\mu \nu}(\varphi^{-1}(x))\\
        &\Big\Downarrow\\
        (\mathcal{D} \circ \varphi^{-1})^*\delta g_{\mu \nu}(\mathcal{D} \circ \varphi^{-1}(x)) = (\mathcal{D} \circ \varphi^{-1})^*&g_{\mu \nu}(\mathcal{D} \circ \varphi^{-1}(x)) - (\varphi^{-1})^*\bar{g}_{\mu \nu}(\varphi^{-1}(x)).
    \end{split}
    \label{Change_of_gauge}
\end{equation}
We see, then, that changing the map $\mathcal{D}$ is equivalent to a change of chart in the physical manifold, keeping the same chart on the background one. Therefore, $(\mathcal{D} \circ \varphi^{-1})^*g_{\mu \nu}(\mathcal{D} \circ \varphi^{-1}(x))$ change following the usual change of chart's transformation rules for tensors, while $(\varphi^{-1})^*\bar{g}_{\mu \nu}(\varphi^{-1}(x))$ remain unchanged.
\newline
Therefore, taking into account that, in order to preserve the smallness of the perturbation, the change of map must be infinitesimal\footnote{This means that it must correspond to an infinitesimal change of coordinates on the physical manifold.} and denoting with $\xi^\mu$ the vector field generator of the infinitesimal transformation, we obtain that the transformation rule for the perturbed metric\footnote{In order to simplify our notation, we have omitted the pullbacks.} is (the hatted quantities are computed in the new gauge)

\begin{equation}
    \delta g_{\mu \nu} = \hat{\delta g}_{\mu \nu} + (\partial_\alpha \bar{g}_{\mu \nu})\xi^\alpha + (\partial_\mu \xi^\rho)\bar{g}_{\rho \nu} + (\partial_\nu \xi^\rho)\bar{g}_{\rho \mu}.
    \label{Gauge_Transf}
\end{equation}
This relation can be immediately translated into the transformation rules for $\psi$, $\phi$, $w_i$ and $\chi_{ij}$:

\begin{equation}
    \begin{split}
        &\hat{\psi} = \psi - \mathcal{H}\xi^0 - \xi^{0'}, \hspace{3cm}  \hat{w}_i = w_i - \zeta_i' + \partial_i \xi^0,\\
        &\\
        &\hat{\phi} = \phi - \mathcal{H}\xi^0 -\frac{1}{3} \partial_l \xi^l, \hspace{2.6cm}  \hat{\chi}_{ij} = \chi_{ij} - \partial_j \zeta_i + \partial_i \zeta_j + \frac{2}{3}\delta_{ij} \partial_l \xi^l,
    \end{split}
    \label{Gauge_Transf_Specified}
\end{equation}
where $\zeta_i := \delta_{ij} \xi^j$, that is different from $\xi_i := a^2 \delta_{ij} \xi^j = a^2 \zeta_i$.
\newline
\newline
Obviously, the transformation rule \ref{Gauge_Transf} holds for any symmetric 2-covariant tensor defined on the physical manifold. Therefore, it can be applied also to the SET, as we will see later in this chapter.
\newline
\newline
Before constructing gauge-invariant quantities, however, we need to introduce another tool that is very useful in the treatment of cosmological perturbations: the Scalar-Vector-Tensor (SVT) decomposition.

\section{The Scalar-Vector-Tensor decomposition}\label{SVT_Decomp}

As just partially anticipated, this decomposition is very useful when we are dealing with linear perturbations, as the cosmological ones. 
\newline
It is based on the Helmholtz theorem, which states that, under certain conditions of regularity, Euclidean vectors and tensors can be decomposed in multiple parts that behave as scalars, vectors and tensors. Therefore, we can apply it to $w_i$ and $\chi_{ij}$\footnote{It can be proven that they satisfy the conditions of regularity cited above.}.
\newline
\newline
More precisely the theorem states that every vector $\mathbf{w}$ in Euclidean space can be split in two parts: one irrotational ($\mathbf{w}^{||}$) and one divergenceless ($\mathbf{w}^{\perp}$), where

\begin{equation}
    \nabla \times \mathbf{w}^{||} = 0 \hspace{1cm} \text{and} \hspace{1cm} \nabla \cdot \mathbf{w}^{\perp} = 0.
    \label{Irr_DivLess}
\end{equation}
Therefore, since the irrotational part $\mathbf{w}^{||}$ can be written as the gradient of a scalar quantity $w$, every component of a vector in the Euclidean space, given a vector $\mathbf{S}$ such that $\nabla \cdot \mathbf{S} = 0$, can be split in this way: 

\begin{equation}
    w_i = \partial_i w + S_i.
    \label{Split_Vector}
\end{equation}
Regarding an Euclidean 3-tensor, instead, the theorem states that it can be split in three different parts: a longitudinal one $\chi_{ij}^{||}$, an orthogonal one $\chi_{ij}^{\perp}$ and a transverse one $\chi_{ij}^{T}$, defined as

\begin{equation}
    \epsilon^{ijk} \partial_l \partial_j \chi_{lk}^{||} = 0, \hspace{1cm} \partial^i \partial^j \chi_{ij}^{\perp} = 0 \hspace{1cm} \text{and} \hspace{1cm} \partial^j \chi_{ij}^{T} = 0.
    \label{Irr_DivLess_Transv}
\end{equation}
These conditions, given a scalar $\mu$ and a divergenceless vector $\mathbf{A}$, translate into 

\begin{equation}
    \begin{split}
        \chi_{ij}^{||} = \bigg( \partial_i \partial_j - \frac{1}{3}\delta_{ij} \nabla^2 \bigg)2 \mu, \hspace{1cm} \text{a} & \text{nd} \hspace{1cm} \chi_{ij}^{\perp} = \partial_j A_i + \partial_i A_j\\
        &\Big \Downarrow\\
        \chi_{ij} = \bigg( \partial_i \partial_j - \frac{1}{3}\delta_{ij} \nabla^2 \bigg)2 \mu &+ \partial_j A_i + \partial_i A_j +\chi_{ij}^{T},
    \end{split}
    \label{Split_Tensor}
\end{equation}
where $\chi_{ij}^{T}$ cannot be decomposed into any scalar or vector. So, it is purely a tensor perturbation.
\newline
From these results we see that, in the treatment of linear perturbations, as cosmological ones, scalar, vector and tensor parts do not mix up. Therefore, they can be independently analyzed.
\newline
\newline
Furthermore, the SVT decomposition can be applied also to the vector $\xi^\mu$ in order to extract from eq. \ref{Gauge_Transf_Specified} the transformation equations for $w$, $\mu$, $S_i$, $A_i$ and $\chi_{ij}^T$.
\newline
More in detail, denoting $\xi^0 = \alpha$ and $\zeta_i = \partial_i \beta + \epsilon_i$ ($\partial^i \epsilon_i = 0$), we get

\begin{equation}
    \begin{split}
        \hat{\psi} = \psi - \mathcal{H}\alpha - \alpha', \hspace{2cm}&  \hat{w} = w - \beta' + \alpha,\\
        \hat{\phi} = \phi - \mathcal{H}\alpha -\frac{1}{3} \nabla^2 \beta, \hspace{1.5cm}&  \hat{\mu} = \mu - \beta,\\
        &\\
        \hat{S}_i = S_i - \epsilon_i', \hspace{3cm}& \hat{A}_i = A_i - \epsilon_i,
    \end{split}
    \label{Separ_Gauge_Transf}
\end{equation}
whereas $\hat{\chi}^T_{ij} = \chi^T_{ij}$, i.e., it is gauge-invariant.
\newline
\newline
Given these relations, we can finally define some gauge-invariant quantities by combining these functions in a suitable manner. The most important ones are the Bardeen potentials:

\begin{equation}
    \Psi := \psi + \frac{1}{a}[(w - \mu')a]' \hspace{1.5cm} \text{and} \hspace{1.5cm} \Phi := \phi + \mathcal{H}(w - \mu') - \frac{1}{3} \nabla^2 \mu.
    \label{Bardeen_Potentials}
\end{equation}
In particular, the first represents the Newtonian potential to which matter is subjected in every point in space-time. Therefore, it gives information about the local gravitational Newtonian field present in a specific region of space-time.
\newline
\newline
The same decomposition naturally applies to the components of the perturbed SET. This allows to form other gauge-invariant quantities useful in the analysis of cosmological perturbations.
\newline
\newline
In conclusion, we can state that the SVT decomposition brings along a double profit. It allows to analyze each kind of perturbation independently of the others and it helps in constructing gauge-invariant quantities like $\Psi$ and $\Phi$.
\newline
\newline
However, the gauge problem carries also a profit, since it allows us to simplify calculations. In fact, since the total number of independent perturbations is 10 and a gauge transformation is defined by 4 independent quantities, we can always choose these quantities in such a way that 4 of the 10 perturbations go to zero. 
\newline
There are several choices that  can be made. In particular, in this master's thesis, we decided to carry out all calculations in what is known as the Newtonian gauge, which is realized by setting $w$, $\mu$ and $A_i$ to zero.
\newline
It is immediate to notice that, in this gauge, all the remaining components of the metric are gauge-invariant, since $\psi$ and $\phi$ coincide with their respective Bardeen Potentials and $S_i$ coincides with $W_i : = S_i - A_i'$, which is gauge-invariant. Furthermore, it can be proven that this gauge is completely fixed. Therefore, its convenience is quite evident.

\section{Scalar perturbations of the spinorial SET}\label{Scal_Pert_SET}
The time has come for us to analyze the perturbations of the spinorial SET. As anticipated, we could analyze independently all kinds of perturbations by starting, for example, from scalar ones, then move to vector ones and, in the end, move to tensor ones.
\newline
This is the standard approach.
\newline
Unfortunately, it has been proven (see \cite{Farnsworth}) that when fermions couple to gravity, they do not admit this kind of decomposition, since some scalar perturbations act as sources of gravitational waves or some vector terms combine to form a scalar perturbation.
\newline
Nevertheless, in this master's thesis, we decided to focus our attention on scalar perturbations alone and analyze what influences they have on our Universe. We will see that, despite this enormous simplification, it is quite difficult to treat them. Hence, we will need another tool that can help us in this analysis. This is the (1+1+2)-decomposition reported in Chapter \ref{SET_Decomp}.
\newline
\newline
We remember that, if we consider only scalar perturbations, the physical metric in the Newtonian gauge is specified by the following line element:

\begin{equation}
    ds^2 = -\big(1 + 2 \Psi \big)dt^2 + a(t)^2 \big(1 + 2 \Phi \big)\delta_{ij} dx^i dx^j.  
    \label{ScalarMet}
\end{equation}
In order to compute the SET for the spinorial fluid in this case, we need to choose a tetrad\footnote{It is worth to remember that the SET is completely independent of the tetrad chosen. Therefore, it is better to choose the simplest tetrad possible, in order to simplify our calculations.} and, then, compute the connection components in this orthonormal RF. The most natural choice is 

\begin{equation}
    \begin{split}
    & e^0 = \big( 1 + \Psi \big)dt =: \tilde{e}^0 + \delta e^0,\\
    & e^i = a(t) \big( 1 + \Phi \big)dx^i =: \tilde{e}^i + \delta e^i,
    \end{split}
    \label{PertTetr}
\end{equation}
where $\tilde{e}^0$ and $\tilde{e}^i$ are the elements of the tetrad given in eq. \ref{Tetr}.
\newline
\newline
From this tetrad we can also define the perturbed curved gamma matrices $\delta \Gamma_\mu$. In fact, the curved gamma matrices in this RF are given by

\begin{equation}
   \Gamma_\mu = e^L_{\hphantom{L} \mu} \gamma_L = \tilde{e}^L_{\hphantom{L} \mu} \gamma_L + \delta e^L_{\hphantom{L} \mu} \gamma_L := \tilde{\Gamma}_\mu + \delta \Gamma_\mu,
   \label{Gamma_Pert}
\end{equation}
where $\tilde{\Gamma}_\mu$ are the curved gamma matrices computed with the background tetrad.
\newline
Since also the connection can be split in the sum of a background part and a perturbation\footnote{The explicit form of the connection is derived in Section \ref{Conn_Perturbed} of Appendix \ref{Computation_Scalar_Pert}.}, we can write the total SET as the sum of a background part, whose explicit form is given in eq. \ref{Eins_Eqs_Model_Cosmo}, and a perturbation, whose form is

\begin{equation}
    \begin{split}
   \delta T_{\mu \nu} =  &-\frac{i}{4}\Big[ \delta \bar{\psi} \tilde{\Gamma}_{(\mu} \tilde{D}_{\nu)}\psi + \bar{\psi} \tilde{\Gamma}_{(\mu} \tilde{D}_{\nu)}\delta \psi -\frac{i}{2}\bar{\psi} \tilde{\Gamma}_{(\mu} \delta \omega^{IJ}_{\hphantom{I}\hphantom{J}\nu)} J_{IJ} \psi + \bar{\psi} \delta \Gamma_{(\mu} \tilde{D}_{\nu)}\psi + \\
   & - \tilde{D}_{(\nu} \delta \bar{\psi} \tilde{\Gamma}_{\mu)}\psi - \tilde{D}_{(\nu} \bar{\psi} \tilde{\Gamma}_{\mu)}\delta \psi -\frac{i}{2}\bar{\psi} \delta \omega^{IJ}_{\hphantom{I}\hphantom{J}(\nu} J_{IJ} \tilde{\Gamma}_{\mu)}\psi - \tilde{D}_{(\nu} \bar{\psi} \delta \Gamma_{\mu)}\psi \Big] + \\
   &+ \delta g_{\mu \nu} \mathcal{E} + g_{\mu \nu} \delta \mathcal{E},
    \end{split}
    \label{SET_Pert}
\end{equation}
where, the "\textasciitilde" above some quantities denotes that these are computed using the background tetrad or the background connection. 
\newline
Using the explicit forms of the curved gamma matrices, of the connection  and of the metric, we can obtain the explicit forms of the components of this tensor. For convenience, since in eqs. \ref{Eins_Tens_00}, \ref{Eins_Tens_0i} and \ref{Eins_Tens_ij} we reported the mixed components of the Einstein tensor, we report the mixed components of the SET. These are:

\begin{equation}
   \delta T^0_{\hphantom{0}0} = -\frac{i}{2}\Big[ \delta \bar{\psi} \tilde{\Gamma}^{0} \dot{\psi} + \bar{\psi} \tilde{\Gamma}^{0} (\dot{\delta \psi}) - ( \dot{\delta \bar{\psi}}) \tilde{\Gamma}^{0}\psi - \dot{\bar{\psi}} \tilde{\Gamma}^{0}\delta \psi \Big] + \Psi \Big[ W + \mathcal{E}  \Big] + \delta \mathcal{E},  
    \label{SETens00Mix}
\end{equation}

\begin{equation}
    \begin{split}
    \delta T^0_{\hphantom{0} i} = &\frac{i}{4}\Big[ \delta \bar{\psi} \tilde{\Gamma}_{i} \dot{\psi} + \bar{\psi} \tilde{\Gamma}_{(0} \partial_{i)}\delta \psi - \dot{\bar{\psi}} \tilde{\Gamma}_i \delta \psi - \partial_{(i} \delta \bar{\psi} \tilde{\Gamma}_{0)} \psi + \\
    & + \delta^{jl} \partial_l \Big( \Psi - \Phi \Big) \bar{\psi} \gamma_0 \gamma_i \gamma_j \psi + \partial_i \Big( \Psi - \Phi \Big) \bar{\psi} \tilde{\Gamma}_0 \psi \Big],
    \end{split}
    \label{SETens0iMix}
\end{equation}

\begin{equation}
    \delta T^i_{\hphantom{i} j} = -\frac{i}{4}g^{im}\Big[ \bar{\psi} \tilde{\Gamma}_{(m} \partial_{j)}\delta \psi  - \partial_{(j} \delta \bar{\psi} \tilde{\Gamma}_{m)}\psi \Big] + \delta^i_{\hphantom{i}j} \delta \mathcal{E}.
    \label{SETensijMix}
\end{equation}
The derivation of these expressions can be found in Section \ref{SET_Derivation} of the Appendix \ref{Computation_Scalar_Pert}, where all the calculations are reported in detail.
\newline
It is quite straightforward to see that these expressions are very complicated and cannot be used for any kind of calculation. Therefore, we need to find a way to make our calculations easier and obtain some relevant result. This will be done in the next section.
\newline
\newline
Nonetheless, these expressions can be used to obtain an equation which relates the density perturbation with the pressure one.
\newline
In fact, even though we do not know what kind of fluid the spinorial one is and how its SET can be decomposed, for every kind of SET related to a fluid, in absence of dissipation\footnote{In Chapter \ref{SET_Decomp}, we will justify why we do not expect to have dissipative effects in the spinorial fluid.}, $T^\mu_{\hphantom{\mu}\mu} = 3P - \rho$. This, given that $T^\mu_{\hphantom{\mu}\mu} = \tilde{T}^\mu_{\hphantom{\mu}\mu} + \delta T^\mu_{\hphantom{\mu}\mu} = 3P - \rho = 3\tilde{P} + 3\delta P - \tilde{\rho} - \delta \rho$, leads to

\begin{equation}
    \delta T^\mu_{\hphantom{\mu}\mu} = 3\delta P - \delta \rho.
    \label{Pert_SET_Trace}
\end{equation}
As a consequence, by exploiting the perturbed Dirac equation

\begin{equation}
    i\delta \Gamma^\mu \tilde{D}_\mu \psi +\frac{1}{2}\tilde{\Gamma}^\mu \delta \omega^{IJ}_{\hphantom{IJ}\mu} J_{IJ}\psi + i\tilde{\Gamma}^\mu \tilde{D}_\mu \delta \psi = \delta \Bigg( \frac{\delta W}{\delta \bar{\psi}}\Bigg),
    \label{PertDirEq}
\end{equation}
we can arrive at the following relation:

\begin{equation}
     \delta \rho = 3 \delta P + U' \delta E - U'' E \delta E - 2  \delta \mathcal{E}.
    \label{PressDens_Final}
\end{equation}
This relation is completely general, in the sense that it holds for any type of potential function $U(E)$. Its derivation can be found in Section \ref{Rel_Dens_Press} of Appendix \ref{Computation_Scalar_Pert}.
\newline
It should be remarked that, although this result has been derived considering only scalar perturbations, it remains valid also when all kinds of perturbations are simultaneously taken into account. In fact, it can be derived from the general form of the perturbed SET \ref{SET_Pert} and the general forms of the Dirac equation \ref{Dirac_Equation_Model} and the perturbed Dirac equation \ref{PertDirEq}.
\newline
Obviously, we can apply it to our cases of interest: DM and DE.
\newline
Given that 

\begin{equation}
    \begin{split}
    & \delta W = U' \delta E + 2 \xi\big( E \delta E + B \delta B\big), \\
    & \delta \mathcal{E} =  2 \xi\big( E \delta E + B \delta B\big) + U'' E \delta E, 
    \end{split}
    \label{PertOfWandE_Text}
\end{equation}
taking into account the expressions of $U$ used in the previous chapter (remember that for DE, in case $\xi \neq 0$, $B = 0$), we obtain

\begin{equation}
    \begin{split}
        & \delta \rho = 3 \delta P + m \delta E - 4 \xi \big( E \delta E + B \delta B\big) \hspace{3cm} \text{(DM)},\\
        & \delta \rho = 3 \delta P \hspace{7.6cm} \text{(DE)}.
    \end{split}
    \label{PressDens_DM_DE}
\end{equation}
These relations will be fundamental in the next chapter in order to prove that the decomposition of the SET used is correct.

\section{The (1+3)-decomposition}\label{RF_Assumption}

As already stated, the expressions obtained for the components of the SET are very complicated and cannot be used for doing any calculation other than that already carried out.
\newline
For this reason, we need to find a different way to face the problem of perturbations of the SET.
\newline
\newline
As we will prove in the next chapter, under certain conditions (certainly satisfied by the spinorial field), every SET, in the absence of dissipation, can be decomposed in this way:

\begin{equation}
    T_{\mu \nu} = \rho u_\mu u_\nu + Q_\mu u_\nu + Q_\nu u_\mu + \Pi_{\mu \nu} + P h_{\mu \nu},
    \label{Real_Fluid_SET}
\end{equation}
where $Q_\mu$ is the heat transfer, which satisfies $Q_\mu u^\mu = 0$, $\Pi_{\mu \nu}$ is the anisotropic stress, which is traceless and satisfies $\Pi_{\mu \nu}u^\nu = 0$, and $h_{\mu \nu} = g_{\mu \nu} + u_\mu u_\nu$.
\newline
This kind of decomposition is called (1+3)-decomposition, since it sorts the SET in quantities that are parallel to $u^\mu$ and quantities which live in the hyperplane orthogonal to it\footnote{For more details about this decomposition see \cite{Maartens}.}.
\newline
We will see later that the SET of a spinorial fluid can be further decomposed thanks to the presence of the axial current density, which is always orthogonal to the four-velocity.
\newline
Nevertheless, for now, we work using this decomposition and we see how far we can go.
\newline
\newline
Considering this kind of decomposition for the spinorial SET and taking into account the previous  assumption that the background SET can be recast as that of a perfect fluid (hence, $Q_\mu$ and $\Pi_{\mu \nu}$ are all first order quantities), we arrive at the form of the perturbed SET:

\begin{equation}
     \delta T_{\mu \nu} = (\delta \rho + \delta P)\tilde{u}_\mu \tilde{u}_\nu + (\rho + P)\delta u_{(\nu} \tilde{u}_{\mu)} + \delta g_{\mu \nu} P + \tilde{g}_{\mu \nu} \delta P + \Pi_{\mu \nu} + Q_{(\mu} \tilde{u}_{\nu)}.
    \label{SET_Pert_RF}
\end{equation}
Clearly, also the components of $\delta T_{\mu \nu}$ depend on the gauge chosen. Therefore, we have to deal with the problem of the gauge. 
\newline
As it has been done at the beginning of this chapter, we assume that the quantities that appear in eq. \ref{SET_Pert_RF} depend on conformal time.
This, considering that the relations of change of gauge for the perturbed SET are the same as those obtained for the metric and defining $Q_i =: aq_i$, $\Pi_{ij} =: a^2\pi_{ij}$ and $\delta u_i =: a v_i$, leads directly to

\begin{equation}
    \begin{split}
        &\hat{\delta \rho} = \delta \rho - \tilde{\rho}'\xi^0, \hspace{1.5cm}  \hat{v}_i = v_i + \partial_i \xi^0, \hspace{1.5cm}  \hat{q}_i = q_i,\\
        &\\
        & \hspace{1.8cm} \hat{\delta P} = \delta P - \tilde{P}'\xi^0, \hspace{1.9cm}  \hat{\pi}_{ij} = \pi_{ij}.
    \end{split}
    \label{Gauge_Transf_Specified_SET}
\end{equation}
Then, we can apply the SVT decomposition also to these quantities and obtain relations similar to those in eq. \ref{Separ_Gauge_Transf}. From these relations, we can construct some gauge-invariant quantities related to the SET. Moreover, we can see how the density constrast $\delta : = \frac{\delta \rho}{\rho}$ modifies under a change of gauge and, considering the behavior of the Hubble parameter, arrive at the important conclusion that, on sub-horizon scales, the density contrast is gauge-invariant.
\newline
\newline
Unfortunately, right now, we have no way to determine how $\rho$ and $P$, in the case of a generic spinorial fluid, can be expressed in terms of bilinears of the spinorial field. Therefore, we do not know how to express $\delta \rho$ and $\delta P$ neither. Nevertheless, we can make a reasonable assumption.
\newline
Considering what we have seen in the background case, we can reasonably suppose that, in general, $\rho = W$ and $P = \mathcal{E}$\footnote{In the next chapter, we will see that these expressions for $\rho$ and $P$ must be corrected by adding two terms related to the axial and the vector current densities.}. From this assumption we derive that

\begin{equation}
    \delta \rho = \delta W \hspace{1.5cm} \text{and} \hspace{1.5cm} \delta P = \delta \mathcal{E}.
    \label{Pert_Dens_Press}
\end{equation}
Therefore, eq. \ref{SET_Pert_RF} becomes

\begin{equation}
     \delta T_{\mu \nu} = (\delta W + \delta \mathcal{E})\tilde{u}_\mu \tilde{u}_\nu + (W + \mathcal{E})\delta u_{(\nu} \tilde{u}_{\mu)} + \delta g_{\mu \nu} \mathcal{E} + \tilde{g}_{\mu \nu} \delta \mathcal{E} + \Pi_{\mu \nu} + Q_{(\mu} \tilde{u}_{\nu)} \hspace{2mm} .
    \label{SET_Pert_PF}
\end{equation}
We can see that, if we compute the trace using the properties of $\delta g_{\mu \nu}$ and $\delta u_{\mu}$, we arrive at the relation \ref{PressDens_Final}, as expected.
\newline
\newline
Since the background fluid is at rest, the background four-velocity has only one non-null component. In particular, in the system of coordinates of conformal time, it is easy to see that

\begin{equation}
    \tilde{u}^0 = \frac{1}{a} \hspace{1.5cm} \text{and} \hspace{1.5cm} \tilde{u}_0 = -a.
    \label{Four_Velocity_Conf}
\end{equation}
From this condition we can derive some identities between the components of the perturbed SET and the quantities that define the real fluid, like $Q_\mu$ or $\Pi_{\mu \nu}$.
\newline
However, before deriving them, it is necessary to rewrite the components of the SET obtained in the previous section as functions of the conformal time and not of the cosmic one. 
\newline
Considering that the curved gamma in the new metric are $\hat{\Gamma}^0 := \hat{g}^{00} \hat{e}^{L}_{\hphantom{L}0} \gamma_L = \frac{1}{a}\tilde{\Gamma}^0$ and $\hat{\Gamma}^i = \tilde{\Gamma}^i$ and considering that $\dot{\psi} = \frac{1}{a}\partial_\eta \psi$ (which we denote with $\frac{1}{a} \psi'$), we have

\begin{equation}
    \delta \hat{T}^0_{\hphantom{0}0} = -\frac{i}{2}\Big[ \delta \bar{\psi} \hat{\Gamma}^{0} \psi' + \bar{\psi} \hat{\Gamma}^{0} (\delta \psi)' - ( \delta \bar{\psi})' \hat{\Gamma}^{0}\psi - \bar{\psi}' \hat{\Gamma}^{0}\delta \psi \Big] + \Psi \Big[ W + \mathcal{E}  \Big] + \delta \mathcal{E}, 
    \label{T00_Conf}
\end{equation}

\begin{equation} 
    \begin{split}
    \delta \hat{T}^0_{\hphantom{0}i} = & \frac{i}{4a^2}\Big[ \delta \bar{\psi} \hat{\Gamma}_{i} \psi' + \bar{\psi} \hat{\Gamma}_{(0} \partial_{i)}\delta \psi - \bar{\psi}' \hat{\Gamma}_i \delta \psi - \partial_{(i} \delta \bar{\psi} \hat{\Gamma}_{0)} \psi + \\
    & + \delta^{jl} \partial_l \Big( \Psi - \Phi \Big) \bar{\psi} \hat{\Gamma}_0 \gamma_i \gamma_j \psi + \partial_i \Big( \Psi - \Phi \Big) \bar{\psi} \hat{\Gamma}_0 \psi \Big],
    \end{split}
    \label{dT_0i_Conf}
\end{equation}

\begin{equation}
    \delta \hat{T}^i_{\hphantom{i}j} = -\frac{i}{4}\hat{g}^{im}\Big[ \bar{\psi} \hat{\Gamma}_{(m} \partial_{j)}\delta \psi  - \partial_{(j} \delta \bar{\psi} \hat{\Gamma}_{m)}\psi \Big] + \delta^i_{\hphantom{i}j} \delta \mathcal{E},
    \label{dT_ij_Conf}
\end{equation}
where, now, all the quantities are functions of $(\eta, \mathbf{x})$.
\newline
Furthermore, we remember that, in the RF chosen, $Q_0 = 0$, $\Pi_{\mu 0} = 0$ and $\delta u^0 = -\frac{\delta g_{00}}{2a^3}$. 
\newline
Given all these results, we can easily obtain the identities we are looking for:

\begin{equation}
    \begin{split}
    \delta \hat{T}^0_{\hphantom{0}0} = & - \big(\delta W + \delta \mathcal{E} \big) + \big(W + \mathcal{E} \big)\big( \delta u_0 \tilde{g}^{00}\tilde{u}_0 + \tilde{u}^0 \delta u_0 \big) + \tilde{g}^{0 \mu} \delta g_{\mu 0} \mathcal{E} + \delta \mathcal{E} + \\ 
    & +\big( W + \mathcal{E} \big) \tilde{u}_\mu \delta g^{\mu 0} \tilde{u}_0 + \delta g^{0 \mu} \tilde{g}_{\mu 0} \mathcal{E} = \\
    = & - \delta W + 2\big(W + \mathcal{E} \big) \frac{\delta u_0}{a} + \big( W + \mathcal{E} \big) a^2 \delta g^{0 0} = - \delta W, 
    \end{split}
    \label{DT00_Flu}
\end{equation}

\begin{equation}
    \delta \hat{T}^0_{\hphantom{0}i} = \big(W + \mathcal{E} \big) av^s_i \tilde{u}^0 + a q^s_i \tilde{u}^0 = \big(W + \mathcal{E} \big) v^s_i + q^s_i, 
    \label{DT0i_Flu}
\end{equation}

\begin{equation}
    \delta \hat{T}^i_{\hphantom{i}j} = \delta^i_{\hphantom{i}j} \delta \mathcal{E} + \tilde{g}^{i \mu} \delta g_{\mu j} \mathcal{E} + \big( \Pi^i_{\hphantom{i}j} \big)^s + \tilde{g}_{j \mu} \delta g^{\mu i} \mathcal{E} = \delta^i_{\hphantom{i}j} \delta \mathcal{E} + \big( \Pi^i_{\hphantom{i}j} \big)^s,
    \label{DTij_Flu}
\end{equation}
where the "$^s$" above some quantities specifies that we are considering only their scalar part, according to the SVT decomposition and the fact that we are dealing only with scalar perturbations. Clearly, in order to obtain their complete expressions, we need to compute also their vector and tensor parts, given respectively by vector and tensor perturbations.
\newline
These identities, not only allow us to find the explicit expressions of $v^s_i$, $q^s_i$ and $ \big( \Pi^i_{\hphantom{i}j} \big)^s$ in terms of the spinorial field, but give us also the conditions that this one must satisfy in order for the assumption made before to be valid. In fact, eq. \ref{DT00_Flu} states that $\delta W$ should be able to be expressed in terms of bilinears containing the spinorial field's perturbation or its derivative with respect to the conformal time. Moreover, eq. \ref{DTij_Flu} leads to

\begin{equation}
     \big( \pi^i_{\hphantom{i}j} \big)^s = \big( \Pi^i_{\hphantom{i}j} \big)^s = -\frac{i}{4}\hat{g}^{im}\Big[ \bar{\psi} \hat{\Gamma}_{(m} \partial_{j)}\delta \psi  - \partial_{(j} \delta \bar{\psi} \hat{\Gamma}_{m)}\psi \Big]
    \label{Pi_ij_Conf}
\end{equation}
(where $\big( \pi^i_{\hphantom{i}j} \big)^s := \delta^{il}\big( \pi_{lj} \big)^s$), which, since the trace of $\big( \pi^i_{\hphantom{i}j} \big)^s$ in this RF must be zero, states that the spinorial field must satisfy also the following condition:

\begin{equation}
    \bar{\psi} \hat{\Gamma}^{i} \partial_{i}\delta \psi  - \partial_{i} \delta \bar{\psi} \hat{\Gamma}^{i}\psi = 0.
    \label{trace0}
\end{equation}

\section{The Einstein equations for scalar perturbations}

The assumption made provides us with a simplified form of the SET and, therefore, of the Einstein equations, making calculations easier. 
\newline
In fact, using the expressions for the perturbed Einstein tensor given in eqs. \ref{Eins_Tens_00}, \ref{Eins_Tens_0i} and \ref{Eins_Tens_ij} considering only scalar perturbations, we arrive at the equations:

\begin{equation} 
    3 \mathcal{H} \Phi' - 3 \mathcal{H}^2 \Psi - \nabla^2 \Phi = 4 \pi G a^2 \delta W, 
    \label{EE_00}
\end{equation}

\begin{equation} 
    \Phi'' + 2 \mathcal{H}\Phi' - \mathcal{H} \Psi' - 2 \frac{a''}{a}\Psi + \mathcal{H}^2 \Psi - \frac{1}{3} \nabla^2 \big( \Phi + \Psi \big) = -4 \pi G a^2 \delta \mathcal{E},
    \label{EE_P}
\end{equation}

\begin{equation} 
    - \Big(\partial^i \partial_j  - \frac{1}{3} \delta^i_{\hphantom{i}j}\nabla^2 \Big) \big( \Phi + \Psi \big) = 8 \pi G a^2 \big( \pi^i_{\hphantom{i}j} \big)^s,
    \label{EE_ij}
\end{equation}

\begin{equation} 
    \partial_i \big( \Phi' - \mathcal{H} \Psi \big) = 4 \pi G a^2 \Big( \big(W + \mathcal{E} \big)v^s_i  + q^s_i \Big).
    \label{EE_0i}
\end{equation}
These are general; in the sense that they do not depend on the explicit form of the function $U$.
\newline
Clearly, they can be specialized for the two cases of our interest, DM and DE, using the equations \ref{PertOfWandE_Text} and the appropriate expression for $U$. Then, other matter components can be added in order to obtain the complete behavior of the Bardeen Potentials and of the density contrasts and other general results. 
\newline
It is worth to remark that, in these cases, independently of the matter component we want to describe, $q^s_i$ and $\big( \pi^i_{\hphantom{i}j} \big)^s$ cannot be put to zero at will, as it is usually done with DM and DE, since these two conditions might be incompatible. In fact, for now, we do not have the explicit expression of $q^s_i$. Hence, we are not able yet to see what part of $\delta \hat{T}^0_{\hphantom{0}i}$ represents $v^s_i$ and what part $q^s_i$, in order to study this compatibility.

\chapter{A decomposition for the SET}\label{SET_Decomp}

In the previous chapter, we computed the explicit expressions of the components of the perturbed SET, taking into account only scalar perturbations. We saw that, even considering a single type of perturbation, finding a solution to the Einstein equations is very hard. Therefore, we had to make some assumptions based on the fact that every SET, under certain conditions, can be decomposed in the way shown in eq. \ref{Real_Fluid_SET}.
\newline
In this chapter, we will see under what conditions this decomposition can be applied and why the properties possessed by the spinorial field allow for a further decomposition of its SET, as anticipated at the end of Chapter \ref{Cosm_Perturb}. Furthermore, employing the results obtained in \cite{Fabbri}, we will see how the thermodynamic components of the SET that appear in this decomposition can be written in terms of the bilinears of the spinorial field.
\newline
In the end, we will be able to derive the adiabatic speed of sound of the pressure perturbation in the context of scalar perturbations.
\newline
In order to simplify calculations, we consider the case $\xi = 0$ and $U(E) = mE$.

\section{The polar form of spinors}\label{Polar_form}
Before discussing the decomposition of the SET, it is necessary to introduce a representation in which every spinor can be expressed, commonly referred to as the polar form. This representation will be very useful to prove that the quantities derived in Section \ref{Ident_Fabbri} lead to the SET obtained in Section \ref{Cosmo_Eq_Sec} in the case of the cosmological background and to achieve an explicit expression of the adiabatic speed of sound, in the context of scalar perturbations, in Section \ref{Adiab_SS}. 
\newline
\newline
We have seen in Chapter \ref{Spinor_Model} that $V^I$ is time-like. Therefore, we can always perform a Lorentz transformation on every spinor in order that its spatial components vanish. Furthermore, remembering that, in such a case, $A^0 = 0$, we can always perform rotations around two of the three spatial axes of the RF in order to align $A^I$ with the remaining spatial axis. These transformations uniquely determine every spinor up to a global phase, which can be eliminated by employing the rotation around the remaining axis, and a phase difference between the chiral parts. All these conditions lead to the fact that every spinor can be expressed 
as 

\begin{equation}
    \psi = \phi e^{-\frac{i}{2}\beta \gamma^5} L^{-1}
    \begin{pmatrix}
        1 \\
        0 \\
        1 \\
        0
    \end{pmatrix},
    \label{Polar_Form_Eq}
\end{equation}
where $L$ is the Lorentz transformation that it is necessary to perform in order to pass from a generic RF to that where $V^I$ and $A^I$ have only one non-vanishing component each and the global phase is vanishing, and $\beta$, called the chiral angle, is the phase difference between chiral parts.
\newline
Moreover, using the properties of Lorentz transformations, it is easy to see that

\begin{equation}
    \begin{split}
        E = 2 \phi^2 \cos{\beta} \hspace{3cm}& B = - 2 \phi^2 \sin{\beta} \\
        A^\mu A_\mu = -V^\mu V_\mu = 4 &\phi^4
    \end{split}
    \label{Identities_Polar}
\end{equation}
Hence, $\phi$ is related to the norm of $V^\mu$ and $A^\mu$. 
\newline
\newline
It is interesting to compute the values of $\phi$ and $\beta$ in the case of the cosmological background.
\newline
At the end of Chapter \ref{Spinor_Model}, we supposed that the background SET can be written as that of a perfect fluid if the Coordinate RF is that at rest\footnote{We will prove that this assumption is correct in the forthcoming sections.}. In that case, $V^i = 0$. This leads to $B = 0$ and, hence, to $\beta = 0, \pi$. Furthermore, from eq. \ref{Integral_Motion}, we obtained $E = \frac{M}{a^3}$, which, considering that the sign difference can be included in $M$, leads to $\phi = \frac{\sqrt{|M|}}{2 a^\frac{3}{2}}$.
\newline
\newline
The polar form and the results derived can be used also to express $\delta E$ and $\delta B$ in terms of the perturbations of $\phi$ and $\beta$.
\newline
In fact, by perturbing eq. \ref{Polar_Form_Eq} and considering only first order quantities, we can arrive at\footnote{This relation is valid for $\beta = 0, \pi$.}

\begin{equation}
    \delta \psi (\mathbf{x}, t) =  \bigg(\frac{\delta \phi}{\phi} - i\frac{\delta \beta }{2}\gamma^5 + \delta (L^{-1})L\bigg) \psi(t).
    \label{Perturbed_Spinor}
\end{equation}
Then, perturbing the identities $L L^{-1} = I$ and $L \gamma^5L^{-1} = \gamma^5$ leads to the relations 

\begin{equation}
    \begin{split}
        &L\hspace{0.5mm}\delta (L^{-1}) + \delta L\hspace{0.5mm} L^{-1} = 0,\\
        & L \hspace{0.5mm}\gamma^5\hspace{0.5mm}\delta (L^{-1}) + \delta L\hspace{0.5mm}\gamma^5\hspace{0.5mm} L^{-1} = 0,
    \end{split}
    \label{Lorentz-Identities}
\end{equation}
which can be used to obtain the identities

\begin{equation}
    \begin{split}
        &\delta E = 2 E \frac{\delta \phi}{\phi},\\
        & \delta B = - \delta \beta \hspace{0.5mm}E.
    \end{split}
    \label{deltaE_deltaB}
\end{equation}
These relations will be very useful in the last part of this chapter, when we derive the adiabatic speed of sound in the context of scalar perturbations.

\section{The (1+1+2)-decomposition of a SET}\label{112_Decomp}
In the previous chapter, we stated that every SET can be written in the form of eq. \ref{Real_Fluid_SET}. In this chapter, we explain why this is possible and why, in the case of the spinorial fluid, the SET can be further decomposed.
\newline
\newline
From linear algebra, we know that, given a vector space $V$ equipped with a scalar product and given a vector $w \in V$, we can always write $V$ as the direct sum of the subspace generated by $w$ ($W$) and its orthogonal complement ($W^\perp$). Therefore, if we construct the symmetric tensor product of $V$ with itself, due to the properties of the tensor product, we get

\begin{equation}
    V \odot V = (W \odot W) \oplus (W \odot W^\perp) \oplus (W^\perp \odot W^\perp).
    \label{Symm_Tens_Product}
\end{equation}
Hence, every tensor in $V \odot V$ can be written as the sum of elements in these tensor subspaces.
\newline
This reasoning can be directly transposed to Tangent Spaces. 
\newline
\newline
Let us suppose that we have a time-like vector field $u \in \Gamma(T\mathcal{M})$, where $\mathcal{M}$ is the space-time\footnote{The reasons behind the choice of a time-like vector field are mainly two. First, if the vector field were null, the decomposition would not be well-defined, because the subspace spanned at each point of $\mathcal{M}$ by such a vector field would be orthogonal to itself. Then, if the vector field were space-like, we could not find a physical rest-frame and, hence, in the case of the SET, we could not express the tensor in the system in which the fluid is at rest.}. Given a point $p \in \mathcal{M}$, $u(p) \in T_p \mathcal{M}$, which is a vector space. Therefore, we can decompose it as the direct sum of the vector subspace spanned by $u(p)$ and its orthogonal complement. Then, analogously to the previous example, we can construct the symmetric tensor product $T_p \mathcal{M} \odot T_p \mathcal{M}$ and decompose it in the same way as we did with $V \odot V$.
\newline
This is what can be done locally. In order for this decomposition to be valid also globally, extending to a decomposition of $T\mathcal{M} \odot T \mathcal{M}$, the section $u$ must be well-defined all over the entire space-time. Unfortunately, in reality, this is never possible, since there exist singularities where $u$ cannot be defined.
\newline
Nevertheless, we will assume that this decomposition can always be applied in a suitable neighborhood of any point in which $u$ is well-defined. In this way, $TU := \coprod_{p \in U} T_p \mathcal{M}$ can be written as the direct sum of the two vector subbundles with fibres respectively given by the span of $u(p)$ in each point $p$ and its orthogonal complement. This directly leads to the possibility of decomposing also its symmetric tensor product. Obviously, due to the properties of the dual map, this decomposition applies to the dual symmetric tensor product.
\newline
\newline
Thus, considering what we have discussed so far, if we can associate a time-like vector field to a generic fluid, we can decompose its SET in the way just shown and this decomposition assumes a physical meaning. Fortunately, such a vector field exists: the four-velocity. As a consequence, in every neighborhood of each point where the four-velocity of the fluid $u^\mu$ is defined, its SET can be decomposed in this way: 

\begin{equation}
    T_{\mu \nu} = \rho u_\mu u_\nu + Q_\mu u_\nu + Q_\nu u_\mu + \tilde{\Pi}_{\mu \nu},
    \label{Real_Fluid_SET_Intermedio}
\end{equation}
where $Q^\mu u_\mu = 0$ and $\tilde{\Pi}^{\mu \nu} u_\nu = 0$, since they respectively belong to $span(u^\mu)^\perp$ and its symmetric tensor product. 
\newline
This expression can be rewritten in a more convenient way by making the trace of $\tilde{\Pi}_{\mu \nu}$ explicit. 
\newline
Before doing this, however, we need to express the quantities defined as  contractions containing the SET. For this purpose, we define the projector $h^\mu_{\hphantom{\mu}\nu} := \delta^\mu_{\hphantom{\mu}\nu} + u^\mu u_\nu$, which is the projector on the subspace orthogonal to $u^\mu$, as can be easily verified. 
\newline
It is easy to see that

\begin{equation}
    \rho = T_{\mu \nu}u^\mu u^\nu,
    \label{Rho_Definition}
\end{equation}

\begin{equation}
    Q_\mu = -(T_{\mu \nu} u^\nu + \rho u_\mu) = -T_{\lambda \nu}u^\nu( \delta^\lambda_{\hphantom{\lambda}\mu}+ u^\lambda u_\mu) = -T_{\lambda \nu}u^\nu h^\lambda_{\hphantom{\lambda}\mu},
    \label{Q_Definition}
\end{equation}

\begin{equation}
    \tilde{\Pi}_{\mu \nu} = T_{\mu \nu} - \rho u_\mu u_\nu - Q_{(\mu}u_{\nu)} = T_{\lambda \tau} h^\lambda_{\hphantom{\lambda}\mu} h^\tau_{\hphantom{\tau}\nu}. 
    \label{Pi_tilde_Definition}
\end{equation}
From the last equation we see that, if we compute the trace of $\tilde{\Pi}_{\mu \nu}$ and we exploit the properties of the projector $h^\mu_{\hphantom{\mu}\nu}$, we obtain $\tilde{\Pi}_{\mu \nu} g^{\mu \nu} = T_{\lambda \tau} h^{\lambda \tau}$. Remembering that the traceless part of a tensor $\tilde{\Pi}_{\mu \nu}$ is obtained by subtracting to that tensor $\frac{1}{G}Tr(\tilde{\Pi}_{\mu \nu})G_{\mu \nu}$, where $G$ is the trace of $G_{\mu \nu}$ and $G_{\mu \nu}$ is the metric tensor in the subspace on which $\tilde{\Pi}_{\mu \nu}$ is defined, and considering that in our case $G_{\mu \nu} = h_{\mu \nu}$, we arrive at the definitions

\begin{equation}
    P := \frac{1}{3}T_{\lambda \tau} h^{\lambda \tau}, 
    \label{P_Definition}
\end{equation}

\begin{equation}
    \Pi_{\mu \nu} := T_{\lambda \tau}( h^\lambda_{\hphantom{\lambda}\mu} h^\tau_{\hphantom{\tau}\nu} - \frac{1}{3}h_{\mu \nu}h^{\lambda \tau}). 
    \label{Pi_Definition}
\end{equation}
These definitions lead directly to the decomposition

\begin{equation}
    T_{\mu \nu} = \rho u_\mu u_\nu + Q_\mu u_\nu + Q_\nu u_\mu + P h_{\mu \nu} + \Pi_{\mu \nu},
    \label{SET_13}
\end{equation}
which is exactly that given in eq. \ref{Real_Fluid_SET}.
\newline
It is worth to remark that, in the case the fluid we are describing were dissipative, we should add a term called bulk viscosity. However, we do not expect the spinorial fluid to be dissipative, since it is described by a fundamental field, which, by definition, has no internal structure.
\newline
Furthermore, we see that $T^\mu_{\hphantom{\mu}\mu} = 3P - \rho$, confirming what we have stated in Chapter \ref{Cosm_Perturb}, since this decomposition is valid for every SET.
\newline
\newline
As anticipated earlier, if we explicitly write this tensor in the RF in which $u^i = 0$ and we apply the results of relativistic thermodynamics, we can deduce the physical meaning of the various quantities which appear in eq. \ref{SET_13}. In fact, in this RF, we have

\begin{equation}
    T^\mu_{\hphantom{\mu}\nu} = 
    \begin{pmatrix}
        -\rho & u^0Q_1 & u^0Q_2 & u^0Q_3 \\
        \\
        u_0Q^1 & P + \Pi^1_{\hphantom{1}1}& \Pi^1_{\hphantom{1}2} & \Pi^1_{\hphantom{1}3} \\\\
        u_0Q^2 & \Pi^2_{\hphantom{2}1} & P + \Pi^2_{\hphantom{2}2} & \Pi^2_{\hphantom{2}3}\\\\
        u_0Q^3 & \Pi^3_{\hphantom{3}1} & \Pi^3_{\hphantom{3}2} & P+ \Pi^3_{\hphantom{3}3}\\
    \end{pmatrix}.
    \label{SET_13_RF}
\end{equation}
From this expression it can be easily seen that $\rho$ represents the energy density, $P$ the isotropic pressure, $Q^\mu$ the energy transfer and $\Pi_{\mu \nu}$ the stress deviator tensor, also called anisotropic stress tensor, as foreseen in the previous chapter.
\newline
\newline
Now, let us suppose that we can associate to the fluid under study another vector field $s$; this time, a space-like vector field. Let us suppose that, given a point $p \in \mathcal{M}$, $s(p) \perp u(p)$. Since the connection defined on space-time is a metric one, we have that the orthogonality between vector fields is preserved along the manifold. Hence, we can state that $s \perp u$ wherever they are both defined. This means that $s \in span(u)^\perp$. Therefore, analogously to the previous case, we can decompose every subspace $span(u(p))^\perp$ into two subspaces: the one spanned by $s(p)$ and its orthogonal complement. Obviously, this decomposition translates in a further decomposition of every possible symmetric tensor product that includes $span(u(p))^\perp$.
\newline
As in the previous case, if both the vector fields were well-defined on the entire manifold, this decomposition would translate into a decomposition of the fibre bundles. However, since this cannot be the case, we assume that this decomposition can always be applied in a suitable neighborhood of every point $p$ where $u$ and $s$ are both defined.
\newline
This means that the quantities which appear in eq. \ref{SET_13} belonging to $span(u)^\perp$ can be further decomposed. In particular, we have 

\begin{equation}
    Q_\mu = Qs_\mu + \tilde{Q}_\mu,
    \label{Q_Decomposition}
\end{equation}

\begin{equation}
    h_{\mu \nu} = s_\mu s_\nu + N_{\mu \nu},
    \label{h_Decomposition}
\end{equation}

\begin{equation}
    \Pi_{\mu \nu} = \Pi s_\mu s_\nu + \Pi_{(\mu}s_{\nu)} + \tilde{\pi}_{\mu \nu},
    \label{Pi_Decomposition}
\end{equation}
where $\tilde{Q}_\mu u^\mu = \tilde{Q}_\mu s^\mu = 0$, $\Pi_{\mu}u^\mu = \Pi_{\mu}s^\mu = 0$, $\tilde{\pi}_{\mu \nu}u^\nu = \tilde{\pi}_{\mu \nu}s^\nu = 0$ and $N_{\mu \nu} := g_{\mu \nu} + u_\mu u_\nu - s_\mu s_\nu$ is also the projector on the orthogonal complement of both subspaces spanned by the two vector fields.
\newline
\newline 
It is easy to see that 

\begin{equation}
    Q = -T_{\lambda \tau}u^\tau h^\lambda_{\hphantom{\lambda}\mu}s^\mu =  -T_{\lambda \tau}u^\tau s^\lambda,
    \label{Q_Scalar_Definition}
\end{equation}

\begin{equation}
    \tilde{Q}_\mu = -T_{\lambda \tau}u^\tau h^\lambda_{\hphantom{\lambda}\mu} +  T_{\lambda \tau}u^\tau s^\lambda s_\mu = -T_{\lambda \tau}u^\tau N^\lambda_{\hphantom{\lambda}\mu},
    \label{Q_tilde_Definition}
\end{equation}

\begin{equation}
    \Pi = -\frac{1}{3}T_{\lambda \tau}(N^{\lambda \tau} - 2 s^\lambda s^\tau),
    \label{Pi_Scalar_Definition}
\end{equation}

\begin{equation}
    \Pi_\mu = T_{\lambda \tau}N^\lambda_{\hphantom{\lambda}\mu}s^\tau,
    \label{Pi_Vector_Definition}
\end{equation}

\begin{equation}
    \tilde{\pi}_{\mu \nu} = T_{\lambda \tau}\bigg(N^\lambda_{\hphantom{\lambda}\mu}N^\tau_{\hphantom{\tau}\nu} - \frac{1}{3}N^{\lambda \tau}N_{\mu \nu} \bigg).
    \label{pi_tilde_Definition}
\end{equation}
We can factor out the trace from $\tilde{\pi}_{\mu \nu}$, as done before with $\Pi_{\mu \nu}$.
\newline
In particular, since $\tilde{\pi}_{\mu \nu}$ lives in the 2-dimensional hyperplane orthogonal to $u^\mu$ and $s^\mu$, its traceless part is $\pi_{\mu \nu} := \tilde{\pi}_{\mu \nu} - \frac{1}{2}Tr(\tilde{\pi}_{\mu \nu})N_{\mu \nu}$. Hence, considering that $\tilde{\pi} := \tilde{\pi}_{\mu \nu} g^{\mu \nu} = \frac{1}{3}T_{\lambda \tau}N^{\lambda \tau}$, $\Pi_{\mu \nu}$ can be written as:

\begin{equation}
    \Pi_{\mu \nu} = \Pi s_\mu s_\nu + \Pi_{(\mu}s_{\nu)} + \frac{1}{2}\tilde{\pi}N_{\mu \nu} + \pi_{\mu \nu}, 
    \label{Pi_Decomposition_2}
\end{equation}
where 

\begin{equation}
    \pi_{\mu \nu} = T_{\lambda \tau}\bigg(N^\lambda_{\hphantom{\lambda}\mu}N^\tau_{\hphantom{\tau}\nu} - \frac{1}{2}N^{\lambda \tau}N_{\mu \nu} \bigg). 
    \label{Pi_Tensor_Decomposition_2}
\end{equation}
Consequently, in this decomposition, called the (1+1+2)-decomposition, the SET is expressed as

\begin{equation}
    T_{\mu \nu} = \rho u_\mu u_\nu + P(N_{\mu \nu} + s_\mu s_\nu) + Qu_{(\mu} s_{\nu)} + \tilde{Q}_{(\mu} u_{\nu)} + \Pi s_\mu s_\nu + \Pi_{(\mu} s_{\nu)} + \frac{1}{2}\tilde{\pi}N_{\mu \nu} + \pi_{\mu \nu}.
    \label{SET_112_Intermedio}
\end{equation}
However, at first sight, it appears that this SET has 11 independent components. This happens because we did not take into account the fact that $\tilde{\pi}$ and $\Pi$ are related, since the total trace of $\Pi_{\mu \nu}$ must vanish. In fact, considering also this condition, we get $\Pi = -\tilde{\pi}$, which leads to the SET

\begin{equation}
    T_{\mu \nu} = \rho u_\mu u_\nu + P(N_{\mu \nu} + s_\mu s_\nu) + Qu_{(\mu} s_{\nu)} + \tilde{Q}_{(\mu} u_{\nu)} - \frac{\Pi}{2}( N_{\mu \nu} - 2 s_\mu s_\nu) + \Pi_{(\mu} s_{\nu)} + \pi_{\mu \nu}.
    \label{SET_112}
\end{equation}
This is the most general way to decompose the SET of a fluid to which we can associate two vectors with the same properties of $u^\mu$ and $s^\mu$. 
\newline
Considering the explicit form of the SET given in eq. \ref{SET_13_RF}, we see that $\Pi_\mu$ are the shear stresses acting on the surfaces orthogonal to the direction of $s^\mu$, $\Pi$ is the anisotropic pressure acting orthogonally to these surfaces and $\tilde{\pi}_{\mu \nu}$ is the remaining part of the anisotropic stress tensor. Similarly, $Q$ is the energy transfer along the direction of $s^\mu$ and $\tilde{Q}^\mu$ are those orthogonal to it.
\newline
\newline
It is natural to ask ourselves if we can apply this decomposition to the spinorial fluid we are studying. Obviously, the answer is affirmative.
\newline
In particular, there must be a vector field which plays the role of the four-velocity of this fluid and a vector field which plays the role of $s^\mu$.
\newline
As anticipated in Chapter \ref{Spinor_Model}, the most natural choice for $u^\mu$ is the normalized vector current density $V^\mu$. This choice is due not only to mathematical convenience or the fact that $V^\mu$ is a conserved current, but also to physical reasons. In fact, the time component of this time-like vector field represents the probability density of the spinorial field, while its spatial components represent the probability density current of the spinorial field. Hence, this vector is somehow related to the "motion" of the spinorial fluid, as a four-velocity should be.
\newline
After this choice, it is straightforward to identify $s^\mu$ with the normalized axial current density $A^\mu$, since this vector is always orthogonal to the vector current density.  
\newline
In conclusion, we can state that, wherever $V^\mu$ and $A^\mu$ are well-defined, the SET of a spinorial fluid can always be expressed  in the form of eq. \ref{SET_112}.

\section{Thermodynamical components of a spinorial fluid}\label{Ident_Fabbri}

We have just proven that the spinorial SET can be decomposed employing the (1+1+2)-decomposition and we have seen how its thermodynamical properties are defined in terms of contractions of the SET with $N^\mu_{\hphantom{\mu} \nu}$, $u^\mu$ and $s^\mu$. 
\newline
It is natural to ask ourselves how these quantities can be expressed in terms of the spinor bilinears.
\newline
In this section, we answer to this question.
\newline
All the results and definitions, and most of the notations, are directly taken from the article \cite{Fabbri}. We refer the reader interested in their derivation to it, where almost all the calculations are reported.
\newline 
\newline
Before exposing the results derived by the authors of the article, however, it is necessary to define some quantities (in all definitions, $\epsilon_{\mu \nu} := \epsilon_{\mu \nu \rho \sigma}u^\rho s^\sigma$): 

\begin{equation}
    \Omega := \frac{1}{2} \nabla_\mu u_\nu \epsilon^{\mu \nu}, \hspace{3cm}  \xi := \frac{1}{2}\nabla_\mu s_\nu \epsilon^{\mu \nu},
    \label{Scalars_Spin_Field_Def}
\end{equation}

\begin{equation}
    \begin{split}
        \Sigma^\mu := \frac{1}{2}N^{\mu \nu}s^\rho (\nabla_\nu u_\rho + \nabla_\rho u_\nu), \hspace{1.5cm}&  \Omega_\mu := \frac{1}{2}N_{\mu \nu}\epsilon^{\nu \rho \sigma \tau} u_\rho\nabla_\sigma u_\tau,\\ 
        \\
        \mathcal{A}^\mu := N^{\mu \nu} u^\rho \nabla_\rho u_\nu, \hspace{2.5cm} & \hspace{0.5cm} a^\mu := N^{\mu \nu} s^\rho \nabla_\rho s_\nu,
    \end{split}
    \label{Vectors_Spin_Field_Def}
\end{equation}

\begin{equation}
    \Sigma_{\mu \nu} := \frac{1}{2}(N^\rho_{\hphantom{\rho}\mu}N^\sigma_{\hphantom{\sigma}\nu} + N^\rho_{\hphantom{\rho}\nu}N^\sigma_{\hphantom{\sigma}\mu} - N_{\mu \nu}N^{\rho \sigma})\nabla_\rho u_\sigma,
    \label{Tensor_Spin_Field_Def}
\end{equation}

\begin{equation}
    \begin{split}
        \dot{\beta} := u^\mu \nabla_\mu \beta, \hspace{1cm} \hat{\beta} := s^\mu &\nabla_\mu \beta, \hspace{1cm} \delta_\mu \beta := N^\nu_{\hphantom{\nu}\mu}\nabla_\nu \beta,\\
        \\
        \delta_\mu \ln{\phi^2} := &N^\nu_{\hphantom{\nu}\mu}\nabla_\nu \ln{\phi^2},
    \end{split}
    \label{Dir_Der_Spin_Field_Def}
\end{equation}
where $\beta$ and $\phi$ are the parameters which appear in the polar form \ref{Polar_Form_Eq}.
\newline
It turns out that it is possible to express the thermodynamical properties of the spinorial fluid in terms of these quantities.
In particular, we have

\begin{equation}
    \begin{split}
        \rho = 2 \phi^2 \bigg(m \cos{\beta} + \Omega + \frac{\hat{\beta}}{2}\bigg), \hspace{2cm} & \hspace{1.3cm} P = \frac{1}{3}\phi^2 (2 \Omega + \hat{\beta}),\\
        Q = \phi^2 (\xi + \dot{\beta}), \hspace{3.3cm} & \hspace{1.3cm} \Pi = -\frac{2}{3}\phi^2 ( \Omega - \hat{\beta}),
    \end{split}
    \label{Scalars_Thermo}
\end{equation}

\begin{equation}
    \tilde{Q}^\mu = -\frac{1}{2} \phi^2 \epsilon^{\mu \nu}(2 \mathcal{A}_\nu - a_\nu - \delta_\nu \ln{\phi^2}), \hspace{1cm} \Pi^\mu = -\frac{1}{2} \phi^2(\Sigma_\nu \epsilon^{\nu \mu} + \Omega^\mu + \delta^\mu \beta),
    \label{Vectors_Thermo}
\end{equation}

\begin{equation}
    \pi^{\mu \nu} = -\frac{1}{2}\phi^2(\Sigma^\mu_{\hphantom{\mu}\lambda}\epsilon^{\lambda \nu} + \Sigma^\nu_{\hphantom{\nu}\lambda}\epsilon^{\lambda \mu}).
    \label{Tensor_Thermo}
\end{equation}
From these identities we can deduce lots of information.
\newline
First of all, we see, as anticipated in the previous chapter, that $\rho$ is not simply equal to $mE$, but it receives a contribution also from the vorticity $\Omega$ and from the derivative of $\beta$ in the direction of the axial current density. The same reasoning is valid for $P$. Then, we can notice that the component of the energy transfer parallel to the axial current density depends on the derivative of $\beta$ along the four-velocity of the fluid and on the twist $\xi$. Hence, we can expect that it does not vanish only in the case the underlying space-time is dynamic. In the end, we can see that the anisotropic pressure $\Pi$ depends on $\Omega$ and $\hat{\beta}$ in a way that the total pressure along the direction of $s^\mu$, $p_s := p + \Pi$, depends only on the derivative of $\beta$ along the direction of $s^\mu$, as it is reasonable\footnote{For more detailed comments about the thermodynamical quantities just defined, we refer the reader to the original article.}.
\newline
\newline
Now, the first thing we can do is applying these results to the cosmological background case in order to see if the assumption made in Chapter \ref{Spinor_Model} about the Perfect Fluid SET is correct. If this were the case, we could also confirm that the Coordinate RF, associated to the choice of coordinates made, is the RF with respect to which the fluid is at rest. 
\newline
Then, we can apply these results to the case of cosmological perturbations and see if we recover the first of the relations \ref{PressDens_DM_DE}, in order to confirm that eqs. \ref{Scalars_Thermo}, \ref{Vectors_Thermo} and \ref{Tensor_Thermo} are correct.

\section{The spinorial SET in the cosmological background}\label{Cosmo_SET_Decomp}

In Section \ref{PF_SET_RF_Rest} of Chapter \ref{Spinor_Model}, we supposed that the SET of the spinorial fluid, when computed in the cosmological background, under certain conditions, can be recast as that of a Perfect Fluid; more precisely, in absence of torsion, as that of a Relativistic Dust.
\newline
Given the results obtained in the previous section, we can verify if this assumption is correct. 
\newline
In order to do this, we have to verify that all the quantities defined in eqs. \ref{Scalars_Spin_Field_Def}, \ref{Vectors_Spin_Field_Def}, \ref{Tensor_Spin_Field_Def} and \ref{Dir_Der_Spin_Field_Def} vanish in this context. This has been done in detail in Appendix \ref{112_Decomp_Cosmo_Back}.
\newline
More in detail, there, we have proven that, in the context of the cosmological background, the SET of a spinorial field can be rewritten as that of a Relativistic Dust if and only if the Coordinate RF, specified by the choice of coordinates made, is the rest-frame for the fluid.
\newline
This totally confirms the assumption made in the previous chapters, since, as we have already stated before, if initial conditions did not guarantee,  for example, $\bar{\psi}_0 \gamma^i \psi_0 = 0$, we could not assume that the SET, in the background case, could be rewritten as that of a Relativistic Dust.
\newline
\newline
Therefore, if we establish that the background SET is that of a Relativistic Dust, we can consider the background fluid at rest in the calculation of the explicit expressions of the thermodynamical quantities in the context of generic cosmological perturbations.
\newline
However, the calculation of these quantities requires a lot of work.
\newline
Hence, we decided to focus our attention only on the relation between the perturbed density and the perturbed pressure. This has been done for two additional reasons. First, we want to verify that we arrive at the same relation between the two perturbations obtained in Chapter \ref{Cosm_Perturb}. Then, the analysis of these quantities might be useful to derive an explicit expression of the adiabatic speed of sound in terms of the parameters which characterize the spinorial field.

\section{The adiabatic speed of sound for scalar perturbations}\label{Adiab_SS}

We have just seen that, in the case of the cosmological background, the only non-vanishing component of the SET is $\tilde{\rho} = \pm 2 \tilde{\phi}^2 m = m \tilde{E}$. 
Therefore, considering the expressions given in eq. \ref{Scalars_Thermo}, we have that $\delta \rho = m \delta E + 2 \tilde{\phi}^2 \Omega + \tilde{\phi}^2 \hat{\delta \beta}$, where $\Omega$ and $\hat{\delta \beta}$ are first order quantities, and that $3\delta P = 2 \tilde{\phi}^2 \Omega + \tilde{\phi}^2 \hat{\delta \beta}$.
\newline 
As a consequence, we immediately obtain

\begin{equation}
    \delta \rho = m \delta E + 3 \delta P,
    \label{Dens_Press_Rel_Fabbri}
\end{equation}
which is easy to see that it coincides with the first of the equations \ref{PressDens_DM_DE} with $\xi = 0$.
\newline
This result verifies that the expressions found for the pressure and the density are correct, since eq. \ref{PressDens_DM_DE} has been derived without exploiting any kind of decomposition, while using only the general expression of the perturbed spinorial SET. 
\newline
\newline
Thanks to the results obtained in the first section of this chapter and to those reported from the article \cite{Fabbri}, besides, we can derive the expression of the adiabatic speed of sound of the pressure perturbation in the context of scalar cosmological perturbations. 
\newline
In fact, in this case, 

\begin{equation}
    \begin{split}
        \Omega =& \frac{1}{2} \nabla_\rho u_\sigma \epsilon^{\rho \sigma} = \frac{1}{2}\partial_\rho u_\sigma \epsilon^{\rho \sigma} = \\
        = & \frac{1}{2} \partial_{\rho} \delta u_\sigma \epsilon^{\rho \sigma} + \frac{1}{2} \partial_{\rho} \tilde{u}_\sigma \epsilon^{\rho \sigma \mu \nu} \tilde{u}_\mu \delta s_\nu + \frac{1}{2} \partial_{\rho} \tilde{u}_\sigma \epsilon^{\rho \sigma \mu \nu} \delta u_\mu  \tilde{s}_\nu = \\
        = & \frac{1}{2} \partial_{\rho} \delta u_\sigma \epsilon^{\rho \sigma} = \frac{1}{2} \partial_{\rho} \partial_\sigma U \hspace{0.5mm} \epsilon^{\rho \sigma} = 0,
    \end{split}
    \label{Omega_Perturbed}
\end{equation}
since we have to consider only the scalar part of $\delta u_\sigma$: i.e., $\partial_\sigma U$.
\newline
Therefore, $3\delta P = \tilde{\phi}^2 \hat{\delta \beta}$. 
\newline
As a consequence, since $m\delta E = 2m \frac{\delta \phi}{\tilde{\phi}} \tilde{E} = \pm 4 \tilde{\phi}^2 m \frac{\delta \phi}{\tilde{\phi}}$, we obtain

\begin{equation}
    \begin{split}
        \delta \rho &= m \delta E + 3 \delta P = 3 \delta P \pm 4 \tilde{\phi}^2 m\frac{\delta \phi}{\tilde{\phi}} = \\
        =& \delta P \bigg( 3 \pm 12\frac{\delta \phi}{\hat{\delta \beta}} \frac{m}{\tilde{\phi}}\bigg) = \delta P \bigg( 3 \pm 24\frac{\delta \phi}{\tilde{s}^\mu \partial_\mu(\delta \beta)} \frac{a^\frac{3}{2}m}{\sqrt{|M|}}\bigg) ,
    \end{split}
    \label{Adiab_Speed}
\end{equation}
from which we can read the adiabatic speed of sound

\begin{equation}
    c_s = \bigg( 3 \pm 24\frac{\delta \phi}{\tilde{s}^\mu \partial_\mu(\delta \beta)} \frac{a^\frac{3}{2}m}{\sqrt{|M|}}\bigg)^{-\frac{1}{2}}. 
    \label{Adiab_Speed_Final}
\end{equation}
We see, then, that $c_s$ depends on the scale parameter and on the perturbations of two of the three quantities that define the spinorial field: $\delta \phi$ and $\delta \beta$.
\newline
\newline
Unfortunately, if we consider all types of perturbations, we cannot obtain the same result, since, in that case, $\Omega = \frac{1}{2}\epsilon^{\rho \sigma} \partial_\rho U_\sigma$, where $U_\sigma$ is the vector part of $\delta u_\sigma$. This makes it impossible for us to express $\delta E$ in terms of the pressure perturbation and, hence, to derive an explicit expression of the adiabatic speed of sound in terms of the quantities that characterize the spinorial field $\psi$.

\chapter{The spherically symmetric case}\label{Spher_Symm}

In the previous chapters, we studied the behavior of the spinorial SET in the cosmological context; both from the background perspective and from the perturbations perspective. 
\newline
In this chapter, we will analyze the second simplest possible case: the spherically symmetric one.
\newline
Indeed, from observations, we infer that galaxies should be surrounded by a spherically symmetric halo of DM. Therefore, if the model presented proposes to be a valid alternative to DM, we have to study its behavior in the aforementioned context and the results it leads to.

\section{The spinorial SET in spherical symmetry}
It is well-known that, in a spherically symmetric space-time, there exists a system of coordinates in which the metric is given by the line element

\begin{equation}
    ds^2 = -e^\nu dt^2 + e^\lambda dr^2 + r^2(d\theta^2 + \sin^2{\theta} \hspace{0.5mm} d\phi^2),
    \label{SpherSymm_LineElem}
\end{equation}
as demonstrated in every book about GR, like \cite{Landau} or \cite{dInverno}
(there, the authors adopt the mostly-minus signature).
\newline
From this metric, analogously to what we did in the other chapters, we can extract the tetrad\footnote{We remember that the expressions of the SET that we are going to compute are completely independent of the tetrad chosen. Hence, it is better to choose a tetrad which simplifies calculations as much as possible.} which will be used to compute the explicit expressions of the components of the SET. 
\newline
The simplest choice we can make is

\begin{equation}
    \begin{split}
    & e^0 = e^{\frac{\nu}{2}}dt, \\
    & e^1 = e^{\frac{\lambda}{2}}dr,\\
    & e^2 = r d\theta, \\
    & e^3 = r \sin{\theta} \hspace{0.5mm} d\phi.
    \end{split}
    \label{Tetr_Spher}
\end{equation}
Taking into account this choice, the expressions of the curved gamma matrices are:

\begin{equation}
    \begin{split}
    &\Gamma^0 = g^{00} \Gamma_0 = g^{00} e^L_{\hphantom{L} 0} \gamma_L = -e^{-\nu} e^{\frac{\nu}{2}}\gamma_0 = -e^{-\frac{\nu}{2}}\gamma_0, \\
    &\Gamma^r = g^{rr} \Gamma_r = g^{rr} e^L_{\hphantom{L} r} \gamma_L = e^{-\lambda} e^{\frac{\lambda}{2}}\gamma_1 = e^{-\frac{\lambda}{2}}\gamma_1, \\
    &\Gamma^\theta = g^{\theta \theta} \Gamma_\theta = g^{\theta \theta} e^L_{\hphantom{L} \theta} \gamma_L = \frac{1}{r^2} r \gamma_2 = \frac{1}{r} \gamma_2, \\
    &\Gamma^\phi = g^{\phi \phi} \Gamma_\phi = g^{\phi \phi} e^L_{\hphantom{L} \phi} \gamma_L = \frac{1}{r^2 \sin^2{\theta}} r \sin{\theta} \hspace{0.5mm} \gamma_3 = \frac{1}{r \sin{\theta}} \gamma_3.
   \end{split}
   \label{Gam_Spher_Symm}
\end{equation}
Successively, after computing the explicit form of the spinorial connection\footnote{The reader can find its complete derivation in Section \ref{Connection_Spher_Symm} of Appendix \ref{Append_Spher_Symm}.}, we can compute the explicit expressions of the components of the SET and of the Dirac equation.
\newline
However, before presenting them to the reader, we have to make a remark. 
\newline
As it is shown in Appendix \ref{Append_Spher_Symm}, the components of the connection explicitly depend on the azimuthal angle $\theta$. Therefore, since these terms do not cancel each other out, we have to consider a $\theta$-dependence also in the spinorial field. This dependence, however, must be such that the scalar $E$, the pseudo-scalar $B$, and other spinor bilinears are independent of this angle, since these terms appear in the Einstein equations, which do not depend on $\theta$.
\newline
In light of this, we can write the Dirac equation

\begin{equation}
    \begin{split}
    &-ie^{-\frac{\nu}{2}}\gamma_0 \dot{\psi} + ie^{-\frac{\lambda}{2}}\gamma_1 \psi' + \frac{i}{r}\gamma_2 \partial_\theta \psi + \frac{i\nu'}{4}e^{-\frac{\lambda}{2}}\gamma_1 \psi -\frac{i\dot{\lambda}}{4} e^{-\frac{\nu}{2}} \gamma_0 \psi + \\
    & + \frac{i}{r} \gamma_1\psi + \frac{i}{2} \cot{\theta} \hspace{0.5mm} \gamma_2 \psi = (U'(E) + 2\xi E) \psi - 2i\xi B \gamma^5 \psi,
   \end{split}
   \label{Dirac_Eq_Spher_Sym_Text}
\end{equation}
(where we have denoted with $\dot{}$ the derivative with respect to $t$ and with $'$ that with respect to $r$) and the components of the SET

\begin{equation}
    T^0_{\hphantom{0} 0} =  \frac{i}{2} e^{-\frac{\nu}{2}}\Big[\bar{\psi} \gamma_0 \dot{\psi} - \dot{\bar{\psi}} \gamma_0 \psi\Big] + \mathcal{E},
    \label{T^0_0_Spher_text}
\end{equation}

\begin{equation}
    T^r_{\hphantom{r} r} = -\frac{i}{2} e^{-\frac{\lambda}{2}}\Big[\bar{\psi} \gamma_1 \psi' - \bar{\psi}' \gamma_1 \psi\Big] + \mathcal{E},
    \label{T^r_r_Spher_text}
\end{equation}

\begin{equation}
    T^\theta_{\hphantom{\theta} \theta} = -\frac{i}{2r} \Big[\bar{\psi} \gamma_2 \partial_\theta \psi - \partial_\theta \bar{\psi} \gamma_2 \psi \Big] + \mathcal{E},
    \label{T^theta_theta_Spher_text}
\end{equation}

\begin{equation}
    T^\phi_{\hphantom{\phi} \phi} =  \mathcal{E},
    \label{T^phi_phi_Spher_text}
\end{equation}

\begin{equation}
    T^0_{\hphantom{0} r} =  \frac{i}{4} e^{-\frac{\nu}{2}}\Big[\bar{\psi} \gamma_0 \psi' - \bar{\psi}' \gamma_0 \psi + e^{\frac{\lambda - \nu}{2}} \big( \bar{\psi} \gamma_1 \dot{\psi} - \dot{\bar{\psi}}\gamma_1 \psi \big)\Big],
    \label{T^0_r_Spher_text}
\end{equation}

\begin{equation}
    T^0_{\hphantom{0} \theta} = \frac{i}{4} e^{-\frac{\nu}{2}} \Big[\bar{\psi} \gamma_0 \partial_\theta \psi -\partial_\theta \bar{\psi} \gamma_0 \psi + e^{-\frac{\nu}{2}}r \big( \bar{\psi}\gamma_2 \dot{\psi} - \dot{\bar{\psi}}\gamma_2 \psi \big) - \frac{1}{2} \big(e^{-\frac{\lambda}{2}} r \nu' - 2 \big) \bar{\psi}\gamma_0 \gamma_1 \gamma_2 \psi \Big],
    \label{T^0_theta_Spher_text}
\end{equation}

\begin{equation}
    T^0_{\hphantom{0} \phi} = \frac{i}{4} e^{-\frac{\nu}{2}} \sin{\theta}\Big[r \cot{\theta} \hspace{0.5mm} \bar{\psi}\gamma_0 \gamma_2 \gamma_3 \psi + e^{-\frac{\nu}{2}}r \big( \bar{\psi}\gamma_3 \dot{\psi} - \dot{\bar{\psi}}\gamma_3 \psi \big) - \frac{1}{2} \big(e^{-\frac{\lambda}{2}}r\nu' - 2 \big) \bar{\psi}\gamma_0 \gamma_1 \gamma_3 \psi \Big],
    \label{T^0_phi_Spher_text}
\end{equation}

\begin{equation}
    T^r_{\hphantom{r} \theta} = -\frac{i}{4} e^{-\frac{\lambda}{2}} \Big[ \bar{\psi} \gamma_1 \partial_\theta \psi -\partial_\theta \bar{\psi} \gamma_1 \psi + e^{-\frac{\lambda}{2}}r \big( \bar{\psi}\gamma_2 \psi' - \bar{\psi}'\gamma_2 \psi \big) - \frac{e^{-\frac{\nu}{2}} r \dot{\lambda}}{2} \bar{\psi}\gamma_0 \gamma_1 \gamma_2 \psi \Big],
    \label{T^r_theta_Spher_text}
\end{equation}

\begin{equation}
    T^r_{\hphantom{r} \phi} = -\frac{i}{4} e^{-\frac{\lambda}{2}}r \sin{\theta} \Big[ e^{-\frac{\lambda}{2}}\big( \bar{\psi}\gamma_3 \psi' - \bar{\psi}'\gamma_3 \psi \big) + \cot{\theta} \hspace{0.5mm} \bar{\psi}\gamma_1\gamma_2 \gamma_3 \psi - \frac{e^{-\frac{\nu}{2}}\dot{\lambda}}{2}\bar{\psi}\gamma_0 \gamma_1 \gamma_3 \psi \Big],
    \label{T^r_phi_Spher_text}
\end{equation}

\begin{equation}
    T^\theta_{\hphantom{\theta} \phi} =  -\frac{i\sin{\theta}}{4r}\big( \bar{\psi}\gamma_3 \partial_\theta \psi - \partial_\theta \bar{\psi} \gamma_3 \psi \big).
    \label{T^theta_phi_Spher_text}
\end{equation}
All these equations are derived in detail in Section \ref{Der_Spher_Symm_SET} of Appendix \ref{Append_Spher_Symm}.
\newline
We see that the tensor is not diagonal and that, actually, none of its components vanishes in the presence of torsion. Therefore, as stated also in the appendix, when we solve the Einstein equations, the last five expressions will play the role of constraints on the spinorial field. 
\newline
In order for this model to be compatible with spherical symmetry, these constraints will have to be simultaneously satisfied. 
\newline
This fact makes us think that the spinorial field under investigation could not be compatible with this kind of symmetry. This could be due to the presence of $A^\mu$, which specifies a preferred spatial direction of motion of the spinorial fluid; thus breaking the symmetry, in the same way as it does in the cosmological background case with $K \neq 0$. 
\newline
This should not happen in the case of axial symmetry. Therefore, in future works, this model could be applied to axially symmetric space-times, such as the Kerr or the Lense-Thirring one.
\newline
\newline
Nevertheless, in the next section, we will try to find a possible solution to the Einstein equations in this context. We will see that this is not a simple task and we will have to make some assumptions.

\section{A possible solution to the Einstein equations}
The first step required to solve the Einstein field equations is writing down the explicit expressions of the components of the Einstein tensor. These can be found in \cite{dInverno} and are\footnote{The different signature has been taken into account.}

\begin{equation}
    G^0_{\hphantom{0} 0} = \frac{e^{-\lambda}}{r^2} - \frac{e^{-\lambda}}{r}\lambda' - \frac{1}{r^2},
    \label{G^0_0_Spher}
\end{equation}

\begin{equation}
    G^0_{\hphantom{0} r} = - \frac{1}{r}e^{-\nu}\dot{\lambda},
    \label{G^0_r_Spher}
\end{equation}

\begin{equation}
    G^r_{\hphantom{r} r} = \frac{e^{-\lambda}}{r}\nu' + \frac{e^{-\lambda}}{r^2} - \frac{1}{r^2},
    \label{G^r_r_Spher}
\end{equation}

\begin{equation}
    G^\theta_{\hphantom{\theta} \theta} = -\frac{e^{-\lambda}}{2r}(\lambda' - \nu') + \frac{e^{-\lambda}}{2}\nu'' - \frac{e^{-\nu}}{4} \dot{\lambda}^2 + \frac{e^{-\nu}}{4}\dot{\nu} \dot{\lambda} - \frac{e^{-\nu}}{2}\ddot{\lambda} + \frac{e^{-\lambda}}{4}(\nu')^2 - \frac{e^{-\lambda}}{4}\lambda' \nu',
    \label{G^theta_theta_Spher}
\end{equation}

\begin{equation}
    G^\phi_{\hphantom{\phi} \phi} = G^\theta_{\hphantom{\theta} \theta}.
    \label{G^phi_phi_Spher}
\end{equation}
Considering the expressions of the components of the SET found, we see that it is very difficult to solve the field equations, unless we make some additional assumption. 
\newline
First of all, we consider vanishing torsion ($\xi = 0$). In this way, we can immediately notice that $T^\phi_{\hphantom{\phi} \phi}$ becomes equal to zero. Then, although this assumption does not have any mathematical or physical basis, we assume that $\psi(t, r, \theta) = f(r)\tilde{\psi}(t, \theta)$, where $f(r)$ is a scalar real function and the $\theta$-dependence is such as to satisfy the conditions mentioned before. 
\newline
This last supposition simplifies considerably our equations, since all terms of the form $\bar{\psi}\gamma_I \psi' - \bar{\psi}'\gamma_I \psi$ vanish. In fact,

\begin{equation}
    \bar{\psi}\gamma_I \psi' - \bar{\psi}'\gamma_I \psi = f'(r)f(r)\big(\bar{\tilde{\psi}}(t, \theta)\gamma_I \tilde{\psi}(t, \theta) - \bar{\tilde{\psi}}(t, \theta) \gamma_I \tilde{\psi}(t, \theta) \big) = 0.
    \label{Condition_f(r)}
\end{equation}
This causes the vanishing of another component of the SET (the $rr$-one) and 
makes other components' form simpler. In fact, we get

\begin{equation}
    T^0_{\hphantom{0} 0} = \frac{i}{2} e^{-\frac{\nu}{2}}\Big[\bar{\psi} \gamma_0 \dot{\psi} - \dot{\bar{\psi}} \gamma_0 \psi\Big],
    \label{T^0_0_Spher_Condition_f(r)}
\end{equation}

\begin{equation}
    T^\theta_{\hphantom{\theta} \theta} =-\frac{i}{2r} \Big[\bar{\psi} \gamma_2 \partial_\theta \psi - \partial_\theta \bar{\psi} \gamma_2 \psi \Big],
    \label{T^theta_theta_Spher_Condition_f(r)}
\end{equation}

\begin{equation}
    T^0_{\hphantom{0} r} = \frac{i}{4} e^{-\nu + \frac{\lambda}{2}}\big( \bar{\psi} \gamma_1 \dot{\psi} - \dot{\bar{\psi}}\gamma_1 \psi \big),
    \label{T^0_r_Spher_Condition_f(r)}
\end{equation}

\begin{equation}
    T^r_{\hphantom{r} \theta} = -\frac{i}{4} e^{-\frac{\lambda}{2}} \Big[ \bar{\psi} \gamma_1 \partial_\theta \psi -\partial_\theta \bar{\psi} \gamma_1 \psi - \frac{e^{-\frac{\nu}{2}} r \dot{\lambda}}{2} \bar{\psi}\gamma_0 \gamma_1 \gamma_2 \psi \Big],
    \label{T^r_theta_Spher_Condition_f(r)}
\end{equation}

\begin{equation}
    T^r_{\hphantom{r} \phi} = -\frac{i}{4} e^{-\frac{\lambda}{2}}r \sin{\theta} \Big[ \cot{\theta} \hspace{0.5mm} \bar{\psi}\gamma_1\gamma_2 \gamma_3 \psi - \frac{e^{-\frac{\nu}{2}}\dot{\lambda}}{2}\bar{\psi}\gamma_0 \gamma_1 \gamma_3 \psi \Big].
    \label{T^r_phi_Spher_Condition_f(r)}
\end{equation}
Now, we can solve the field equations.
\newline
Firstly, from eq. \ref{G^phi_phi_Spher} we obtain that $T^\theta_{\hphantom{\theta} \theta} = T^\phi_{\hphantom{\phi} \phi} = 0$, from which we derive $\bar{\psi} \gamma_2 \partial_\theta \psi - \partial_\theta \bar{\psi} \gamma_2 \psi = 0$.
\newline
This condition, together with our assumptions, has an important consequence on the equations of motion of the metric tensor. In fact, if we contract the Dirac equation with $\bar{\psi}$ and we add to it its complex conjugate, considering what we stated, we arrive at 

\begin{equation}
    -ie^{-\frac{\nu}{2}}\big(\bar{\psi}\gamma_0 \dot{\psi} - \dot{\bar{\psi}}\gamma_0 \psi \big) = 2mE,
    \label{Dirac_Eq_Spher_Sym_Condition_f(r)}
\end{equation}
which directly implies $T^0_{\hphantom{0} 0} = -mE$.
\newline
\newline
This fact, together with eq. \ref{G^0_0_Spher}, translates into the equation

\begin{equation}
    \frac{e^{-\lambda}}{r^2} - \frac{e^{-\lambda}}{r}\lambda' - \frac{1}{r^2} =\frac{(re^{-\lambda})'}{r^2} - \frac{1}{r^2} = -8\pi G m E,
    \label{Einst_Eq_00_Spher}
\end{equation}
which, taking into account that $E(r, t) = f^2(r) g(t)$ (where $g(t)$ is a generic scalar function of $t$), as our previous results require, can be integrated obtaining\footnote{We remark that we should have added a function $A(t)$ divided by $r$. However, since outside the halo the metric must be the Schwarzschild one, we must have $A(t) = 0$.}

\begin{equation}
    e^{-\lambda} = 1 - \frac{8 \pi G m}{r}\int_0^r E(\tilde{r}, t)\tilde{r}^2d\tilde{r} = 1 - \frac{8 \pi G m g(t)}{r}\int_0^r f^2(\tilde{r})\tilde{r}^2d\tilde{r}.
    \label{Inv_Grav_Pot_Spher}
\end{equation}
Then, we can use eq. \ref{G^r_r_Spher} to obtain the behavior of $e^\nu$. 
\newline
Denoting the integral with $I(r)$, this equation gives 

\begin{equation}
    \nu' = \Big(\frac{1}{r^2} - \frac{e^{-\lambda}}{r^2} \Big)e^\lambda r = -\frac{1}{r} + \frac{1}{r -  8\pi G m g(t)I(r)} = \frac{1}{r}\frac{8\pi G m g(t)I(r)}{r -  8\pi G m g(t)I(r)},
    \label{Nu'}
\end{equation}
which can be integrated once the explicit expression of $f(r)$ is known. 
\newline
This last can be obtained by solving the Dirac equation. 
\newline
However, since this equation cannot be analytically solved, we can suppose a functional form of $f(r)$. The simplest we can choose is, obviously, a constant. Therefore, let us define $f^2(r) =: \rho_0$.
\newline
This leads to $I(r) = \frac{\rho_0 r^3}{3}$, which in turn leads to

\begin{equation}
    e^{\lambda} = \Bigg(1 - \frac{8 \pi G m g(t) \rho_0}{3} r^2\Bigg)^{-1}.
    \label{Grav_Pot_Spher_f(r)_Const}
\end{equation}
Regarding $e^\nu$, instead, if we take into account the assumption made and we integrate eq. \ref{Nu'} between 0 and $r$, we obtain

\begin{equation}
    \nu = -\frac{1}{2} \ln \Bigg({1- \frac{8 \pi G m g(t) \rho_0}{3}r^2} \Bigg) \Longrightarrow e^\nu = \Bigg(1- \frac{8 \pi G m g(t) \rho_0}{3}r^2 \Bigg)^{-\frac{1}{2}}.
    \label{Nu_f(r)_Const}
\end{equation}
Therefore, $e^\nu = e^{\frac{\lambda}{2}}$. 
\newline
As a consequence, the line element becomes

\begin{equation}
    ds^2 = -\frac{1}{\sqrt{1 - D(t)r^2}} dt^2 + \frac{1}{1 - D(t)r^2} dr^2 + r^2(d\theta^2 + \sin^2{\theta} \hspace{0.5mm} d\phi^2),
    \label{SpherSymm_LineElem_Solved}
\end{equation}
where we have defined $D(t) := \frac{8 \pi G m g(t) \rho_0}{3}$ and $g(t)$ is to be determined by means of the Einstein and the Dirac equations.
\newline
We see that this line element is well-defined up to $r^2 < \frac{1}{D(t)} = \frac{3}{8 \pi G m g(t) \rho_0}$.
\newline
\newline
However, as we remarked in the previous section, we believe that, without making any assumption, a solution in this context could not be found, since the spinorial field naturally breaks spherical symmetry.

\chapter*{Conclusions}
\addcontentsline{toc}{chapter}{Conclusions} \markboth{Conclusions}{}

In this master's thesis, we focused our attention on the spinor model proposed in the article \cite{Magueijo}, and we analyzed it from the point of view of the cosmological background. We showed that this model, under suitable conditions, can well-describe the background behavior of Dark Matter (DM) and, under other suitable conditions, the behavior of Dark Energy (DE). Furthermore, we showed that the Stress-Energy Tensor (SET) of the spinorial fluid, in the context of the cosmological background, can be recast in the form of that of a Perfect Fluid whose four-velocity $u^\mu$ is given by the normalized vector current density $V^\mu := \bar{\psi} \Gamma^\mu \psi$.
\newline
\newline
Successively, we concentrated on the analysis of the scalar cosmological perturbations of this model, following the usual Scalar-Vector-Tensor (SVT) Decomposition approach, although in \cite{Farnsworth} it was proven that the SVT Decomposition cannot be applied when dealing with spinorial matter. We showed that the treatment of cosmological perturbations is very difficult and cannot be done directly; i.e., by using the explicit expressions of the perturbed SET. Therefore, a further simplification is necessary. 
\newline
Despite all these problems, we managed to obtain a relation between the density perturbation and the pressure perturbation using the expressions found. Moreover, we proved that this relation is also valid for general perturbations.
\newline
\newline
Due to the impossibility of directly making use of the results found, we decided to tackle the problem with another method: the (1+3)-decomposition. We showed that this kind of decomposition can be applied to any 2-covariant symmetric tensor and that, in the case of the Stress-Energy tensor, it assumes physical significance. We also proved that we can further decompose the SET of a spinorial fluid thanks to the presence of the axial current density $A^\mu$. This new decomposition takes the name of (1+1+2)-decomposition.
\newline
Employing the results found in \cite{Fabbri} and the polar form of spinors, we showed how the thermodynamical quantities that characterize the spinorial fluid can be written in terms of the four-velocity, the normalized axial current density $s^\mu$, the chiral angle, and the projector on the hyperplane orthogonal to $u^\mu$ and $s^\mu$.
\newline
Moreover, thanks to these results, in the context of scalar perturbations, we managed to obtain an expression for the adiabatic speed of sound of the pressure perturbation in terms of the parameters that characterize the spinorial field written in polar form. We also showed that, at present, an analogous expression cannot be obtained if all types of perturbations are considered.
\newline
\newline
In the end, we tackled the problem of spherically symmetric halos surrounding galaxies by analyzing the behavior of the Stress-Energy Tensor of the spinorial fluid in a spherically symmetric space-time. We demonstrated that, in general, all its components do not vanish, giving rise to the problem of the suitability of the model proposed for the description of spherically symmetric objects.
\newline
Despite this problem, we tried to find a solution to the Einstein field equations by supposing a possible form of the spinorial field, which helped to significantly simplify our calculations.
\newline
\newline
Since this field of research is almost unexplored, there are several possible ideas that can be pursued in future works; beginning from the background and moving to cosmological perturbations or other kinds of space-times.
\newline
First of all, at the background level, other possible solutions related to different forms of the scalar function $U(E)$ could be explored. One possibility, for example, could be $U(E) = mE -\xi E^2 + C$. In fact, considering the results obtained in the case $A^0 = 0$, this choice could behave in the same way as DM and DE do when they are considered at the same time. Therefore, the model, at least at the background level, could act as a unified description of these two matter components of our Universe.
\newline
Then, regarding cosmological perturbations, the other thermodynamical properties of the spinorial fluid could be computed, and calculations could be carried out using these quantities. Furthermore, it could be verified that the explicit expressions of the perturbed SET obtained by considering all types of perturbations coincide with those obtained by means of the (1+1+2)-decomposition. These results could help to find an expression for the adiabatic speed of sound in case all possible perturbations are taken into account. This would provide the conditions that the spinorial field must satisfy in order to be a good candidate for Dark Matter.
\newline
In the end, the model could be applied to other kinds of space-times. The most likely could be the axially symmetric one. In this case, it could be verified if the model is compatible with this kind of symmetry and, in the case of confirmation, the expressions of the thermodynamical quantities in terms of the spinor bilinears could be found. Then, in case the Einstein equations could be solved, their solutions could be used to obtain the rotation curves of those galaxies surrounded by an axially symmetric DM halo, in order to verify that they match with the ones extracted by observations.
\newline
\newline
In conclusion, this master's thesis has shown that dealing with a spinor model with the aim of describing DM or DE is not a simple task. This leads directly to the problem of defining a suitable decomposition for the Stress-Energy Tensor of this kind of fluid and of  expressing its thermodynamical quantities in terms of the spinor bilinears. Since this requires a lot of effort and time, in this work we focused only on scalar perturbations and a spherically symmetric space-time. Therefore, the possibilities for future works are diverse, paving the way for the chance of finding a valid candidate for DM or DE in at least one of the contexts previously analyzed or proposed.

\appendix

\renewcommand{\chaptername}{Appendix}

\chapter{The spinorial SET and Scalar Cosmological Perturbations}\label{Computation_Scalar_Pert}

In this appendix, we perform all calculations required to derive the results reported in Chapter \ref{Cosm_Perturb} in the context of scalar cosmological perturbations.

\section{Derivation of the perturbed connection}\label{Conn_Perturbed}

As anticipated in the main text, once we have specified a tetrad, we can compute the components of the connection (hence, the entire connection) in that RF. In particular, following what we have done in the case of the cosmological background, we can obtain the components of the connection by using the structure equations \ref{Structure_Eqs} with $T^I = 0$.
\newline
\newline
If we consider the tetrad given in eq. \ref{PertTetr}, it is easy to see that 

\begin{equation}
    \begin{split}
    de^0  & = - \partial_l \Psi dt \wedge dx^l,\\
    de^i & = - \dot{a} \big( 1 + \Phi \big) dx^i \wedge dt - a \partial_\mu \Phi dx^i \wedge dx^{\mu} = \\
    & = - \Big[ \dot{a} \big( 1 + \Phi \big) + a \dot{\Phi} \Big] dx^i \wedge dt - a \partial_l \Phi dx^i \wedge dx^{l},
    \end{split}
    \label{Diff_Pert_Tetrad}
\end{equation}
where the sum over $l$ excludes the term $l = i$.
\newline
As a consequence, putting these results into the first of the structure equations, we get

\begin{equation}
    \begin{split}
    &\omega^i_{\hphantom{i} 0} =  \Big[ \dot{a} \big( 1 + \Phi - \Psi \big) + a \dot{\Phi} \Big] dx^i + ( \propto dt), \\
    &\omega^0_{\hphantom{0} i} = \frac{1}{a} \partial_i \Psi dt + ( \propto dx^i),\\
    &\omega^i_{\hphantom{i} j} = \partial_j \Phi dx^i + (\propto dx^j). 
    \end{split}
    \label{NewConn}
\end{equation}
We point out that these results have been obtained at first order in perturbations theory, as the whole treatment require. Furthermore, the notation $\propto$, used for terms like ($\propto dt$), comes out from the fact that we cannot derive the explicit expression of these terms by using only one of the equations, but we need also the second structure equation.
\newline
If we take also this one into account, in the end, we arrive at

\begin{equation}
    \begin{split}
    & \omega^{i0} = -\omega^{0i} = - \Bigg[ \Big[ \dot{a} \big( 1 + \Phi - \Psi \big) + a \dot{\Phi}\Big]\delta^i_{\hphantom{i}\mu} + \frac{1}{a} \delta^{li} \partial_l \Psi \delta^0_{\hphantom{0}\mu}  \Bigg] dx^\mu, \\
    &\omega^{ij} = - \omega^{ji} =  \Big[ \delta^{jl} \partial_l \Phi \delta^i_{\hphantom{i}\mu} - \delta^{il} \partial_l \Phi \delta^j_{\hphantom{j}\mu} \Big] dx^\mu. 
    \end{split}
    \label{NewConnUp}
\end{equation}
We immediately notice that the connection components obtained can be written as the sum of the background connection components and a perturbation.
\newline
Therefore, remembering the form of the background connection, given in Section \ref{Cosmo_Eq_Sec}, we can obtain the explicit form of the perturbed connection $\delta \omega^{IJ}_{\hphantom{I}\hphantom{J}\nu} J_{IJ}$.
\newline
More in detail, we get

\begin{equation}
    \begin{split}
    \delta \omega^{IJ}_{\hphantom{I}\hphantom{J}0} J_{IJ} & = 2 \delta \omega^{0i}_{\hphantom{0}\hphantom{i}0} J_{0i} = -i\delta \omega^{0i}_{\hphantom{0}\hphantom{i}0} \gamma_0 \gamma_i = -\frac{i}{a} \delta^{il} \partial_l \Psi \gamma_0 \gamma_i,\\
    \delta \omega^{IJ}_{\hphantom{I}\hphantom{J}k} J_{IJ} & = 2 \delta \omega^{0i}_{\hphantom{0}\hphantom{i} k} J_{0i} + \delta\omega^{ij}_{\hphantom{i}\hphantom{j}k} J_{ij} = \\
    & = -i \Big[ \dot{a} \big( \Phi - \Psi \big) + a \dot{\Phi}\Big]\delta^i_{\hphantom{i}k} \gamma_0 \gamma_i - \frac{i}{2} \Big[ \delta^{jl} \partial_l \Phi \delta^i_{\hphantom{i}k} - \delta^{il} \partial_l \Phi \delta^j_{\hphantom{j}k} \Big]\gamma_i \gamma_j = \\
    & = -i \Big[ \dot{a} \big( \Phi - \Psi \big) + a \dot{\Phi}\Big] \gamma_0 \gamma_k - \frac{i}{2} \Big[ \delta^{jl} \partial_l \Phi \gamma_k \gamma_j  - \delta^{il} \partial_l \Phi \gamma_i \gamma_k \Big] = \\
    & = -i \Big[ \dot{a} \big( \Phi - \Psi \big) + a \dot{\Phi}\Big] \gamma_0 \gamma_k - \frac{i}{2} \delta^{jl} \partial_l \Phi \Big[ \gamma_k \gamma_j - \gamma_j \gamma_k \Big] = \\
    & = -i \Big[ \dot{a} \big( \Phi - \Psi \big) + a \dot{\Phi}\Big] \gamma_0 \gamma_k - i \delta^{jl} \partial_l \Phi \Big[ \gamma_k \gamma_j + \delta_{kj} \Big] = \\
    & = -i \Big[ \dot{a} \big( \Phi - \Psi \big) + a \dot{\Phi}\Big] \gamma_0 \gamma_k - i \delta^{jl} \partial_l \Phi \gamma_k \gamma_j -i \partial_k \Phi \hspace{1mm} .
    \end{split}
    \label{PertConn}
\end{equation}

\section{Derivation of the perturbed spinorial SET}\label{SET_Derivation}

The expression of the connection derived in the previous section, the perturbed gamma matrices and the perturbed tetrad can be all used for computing the explicit expressions of the mixed components of the spinorial SET reported in Chapter \ref{Cosm_Perturb}. 

\subsection{Covariant 00-component}
We begin from the covariant 00-component:

\begin{equation}
    \begin{split}
   \delta T_{00} =  &-\frac{i}{2}\Big[ \delta \bar{\psi} \tilde{\Gamma}_{0} \tilde{D}_{0}\psi + \bar{\psi} \tilde{\Gamma}_{0} \tilde{D}_{0}\delta \psi -\frac{i}{2}\bar{\psi} \tilde{\Gamma}_{0} \delta \omega^{IJ}_{\hphantom{I}\hphantom{J}0} J_{IJ} \psi + \bar{\psi} \delta \Gamma_{0} \tilde{D}_{0}\psi + \\
   & - \tilde{D}_{0} \delta \bar{\psi} \tilde{\Gamma}_{0}\psi - \tilde{D}_{0} \bar{\psi} \tilde{\Gamma}_{0}\delta \psi -\frac{i}{2}\bar{\psi} \delta \omega^{IJ}_{\hphantom{I}\hphantom{J}0} J_{IJ} \tilde{\Gamma}_{0}\psi - \tilde{D}_{0} \bar{\psi} \delta \Gamma_{0}\psi \Big] + \\
   & + \delta g_{00} \mathcal{E} + g_{00} \delta \mathcal{E} = \\
   = &-\frac{i}{2}\Big[ \delta \bar{\psi} \tilde{\Gamma}_{0} \dot{\psi} + \bar{\psi} \tilde{\Gamma}_{0} (\dot{\delta \psi}) -\frac{i}{2}\bar{\psi} \tilde{\Gamma}_{0} \delta \omega^{IJ}_{\hphantom{I}\hphantom{J}0} J_{IJ} \psi + \bar{\psi} \delta \Gamma_{0} \dot{\psi} + \\
   & - ( \dot{\delta \bar{\psi}}) \tilde{\Gamma}_{0}\psi - \dot{\bar{\psi}} \tilde{\Gamma}_{0}\delta \psi -\frac{i}{2}\bar{\psi} \delta \omega^{IJ}_{\hphantom{I}\hphantom{J}0} J_{IJ} \tilde{\Gamma}_{0}\psi - \dot{\bar{\psi}} \delta \Gamma_{0}\psi \Big] - 2 \Psi \mathcal{E} - \delta \mathcal{E}.
    \end{split}
    \label{SETens00_1}
\end{equation}
Since $ \delta \omega^{IJ}_{\hphantom{I}\hphantom{J}0} J_{IJ}$ is proportional to $\gamma_0 \gamma_i$ and $\tilde{\Gamma}_0 = \tilde{e}^0_{\hphantom{0}0} \gamma_0 = \gamma_0$, it is straightforward to see that 

\begin{equation}
    \bar{\psi}\tilde{\Gamma}_0 \delta \omega^{IJ}_{\hphantom{I}\hphantom{J}0} J_{IJ}\psi + \bar{\psi}\delta \omega^{IJ}_{\hphantom{I}\hphantom{J}0} J_{IJ}\tilde{\Gamma}_0 \psi = 0.
    \label{Gam0_Rel_1}
\end{equation}
Furthermore, since $\delta \Gamma_0 = \delta e^0_{\hphantom{0}0}\gamma_0 = \Psi \gamma_0 $, we obtain

\begin{equation}
    \begin{split}
    &\bar{\psi} \delta\Gamma_{0} \dot{\psi} - \dot{\bar{\psi}} \delta \Gamma_{0}\psi = \Psi \Big[\bar{\psi} \gamma_{0} \dot{\psi} - \dot{\bar{\psi}} \gamma_{0}\psi \Big] = - \Psi \Big[\bar{\psi} \gamma^{0} \dot{\psi} - \dot{\bar{\psi}} \gamma^{0}\psi \Big] = \\
    =&- \Psi \bigg[\bar{\psi} \bigg(-\frac{3}{2}H\gamma^0\psi -i\frac{\delta W}{\delta \bar{\psi}}\bigg) - \bigg(-\frac{3}{2}H\bar{\psi} +i\bigg(\frac{\delta W}{\delta \bar{\psi}}\bigg)^{\dagger} \bigg)\gamma^0 \psi \bigg] = \\
    =& i \Psi \bigg[\bar{\psi} \bigg(\frac{\delta W}{\delta \bar{\psi}}\bigg) + \bigg( \bigg(\frac{\delta W}{\delta \bar{\psi}}\bigg)^{\dagger} \bigg)\gamma^0 \psi \bigg] = \\
    =& i \Psi \Big[ U' E + 2\xi E^2 + 2 \xi B^2 + U'E + 2 \xi E^2 + 2 \xi B^2 \Big] = 2i \Psi \Big[ \mathcal{E} + W \Big].
    \end{split}
    \label{Gam0_rel}
\end{equation}
From these results it directly follows:

\begin{equation}
    \begin{split}
   \delta T_{00} = &-\frac{i}{2}\Big[ \delta \bar{\psi} \tilde{\Gamma}_{0} \dot{\psi} + \bar{\psi} \tilde{\Gamma}_{0} (\dot{\delta \psi}) - ( \dot{\delta \bar{\psi}}) \tilde{\Gamma}_{0}\psi - \dot{\bar{\psi}} \tilde{\Gamma}_{0}\delta \psi \Big] + \Psi \Big[ W - \mathcal{E}  \Big] - \delta \mathcal{E}.
    \end{split}
    \label{SETens00_2}
\end{equation}

\subsection{Covariant $ij$-component}

Now, we can move on to the covariant $ij$-component:

\begin{equation}
    \begin{split}
    \delta T_{i j} =  &-\frac{i}{4}\Big[ \delta \bar{\psi} \tilde{\Gamma}_{(i} \tilde{D}_{j)}\psi + \bar{\psi} \tilde{\Gamma}_{(i} \tilde{D}_{j)}\delta \psi -\frac{i}{2}\bar{\psi} \tilde{\Gamma}_{(i} \delta \omega^{IJ}_{\hphantom{I}\hphantom{J}j)} J_{IJ} \psi + \bar{\psi} \delta \Gamma_{(i} \tilde{D}_{j)}\psi + \\
    & - \tilde{D}_{(j} \delta \bar{\psi} \tilde{\Gamma}_{i)}\psi - \tilde{D}_{(j} \bar{\psi} \tilde{\Gamma}_{i)}\delta \psi -\frac{i}{2}\bar{\psi} \delta \omega^{IJ}_{\hphantom{I}\hphantom{J}(j} J_{IJ} \tilde{\Gamma}_{i)}\psi - \tilde{D}_{(j} \bar{\psi} \delta \Gamma_{i)}\psi \Big] + \\
    & + \delta g_{ij} \mathcal{E} + g_{ij} \delta \mathcal{E} = \\
    = & -\frac{i}{4}\Big[ \delta \bar{\psi} \tilde{\Gamma}_{(i} \tilde{D}_{j)}\psi + \bar{\psi} \tilde{\Gamma}_{(i} \tilde{D}_{j)}\delta \psi -\frac{i}{2}\bar{\psi} \tilde{\Gamma}_{(i} \delta \omega^{IJ}_{\hphantom{I}\hphantom{J}j)} J_{IJ} \psi + \bar{\psi} \delta \Gamma_{(i} \tilde{D}_{j)}\psi + \\
    & - \tilde{D}_{(j} \delta \bar{\psi} \tilde{\Gamma}_{i)}\psi - \tilde{D}_{(j} \bar{\psi} \tilde{\Gamma}_{i)}\delta \psi -\frac{i}{2}\bar{\psi} \delta \omega^{IJ}_{\hphantom{I}\hphantom{J}(j} J_{IJ} \tilde{\Gamma}_{i)}\psi - \tilde{D}_{(j} \bar{\psi} \delta \Gamma_{i)}\psi \Big] + \\
    & + 2 a^2 \Phi \delta_{ij} \mathcal{E} + a^2 \delta_{ij} \delta \mathcal{E}. \\
    \end{split}
    \label{SETensij_1}
\end{equation}
Given that $\tilde{\Gamma}_i = a \gamma_i$, $\delta \Gamma_i = \delta e^L_{\hphantom{L} i} \gamma_L = a \Phi \gamma_i$ and $\tilde{\omega}^{IJ}_{\hphantom{I}\hphantom{J}j} J_{IJ} = -i \dot{a} \gamma_0 \gamma_j$ and considering that the background spinorial field depends only on the cosmic time, we obtain

\begin{equation}
    \begin{split}
    & \delta \bar{\psi} \tilde{\Gamma}_{(i} \tilde{D}_{j)}\psi + \bar{\psi} \delta \Gamma_{(i} \tilde{D}_{j)}\psi - \tilde{D}_{(j} \bar{\psi} \tilde{\Gamma}_{i)}\delta \psi  - \tilde{D}_{(j} \bar{\psi} \delta \Gamma_{i)}\psi = \\
    = & - \frac{i}{2}\delta \bar{\psi} \tilde{\Gamma}_{i} \tilde{\omega}^{IJ}_{\hphantom{I}\hphantom{J}j} J_{IJ}\psi - \frac{i}{2}\delta \bar{\psi} \tilde{\Gamma}_{j} \tilde{\omega}^{IJ}_{\hphantom{I}\hphantom{J}i} J_{IJ}\psi  - \frac{i}{2} \bar{\psi} \delta \Gamma_{i} \tilde{\omega}^{IJ}_{\hphantom{I}\hphantom{J}j} J_{IJ}\psi - \frac{i}{2} \bar{\psi} \delta \Gamma_{j} \tilde{\omega}^{IJ}_{\hphantom{I}\hphantom{J}i} J_{IJ}\psi + \\
    & - \frac{i}{2}\bar{\psi} \tilde{\omega}^{IJ}_{\hphantom{I}\hphantom{J}j} J_{IJ} \tilde{\Gamma}_{i}\delta \psi - \frac{i}{2}\bar{\psi} \tilde{\omega}^{IJ}_{\hphantom{I}\hphantom{J}i} J_{IJ} \tilde{\Gamma}_{j}\delta \psi  - \frac{i}{2} \bar{\psi} \tilde{\omega}^{IJ}_{\hphantom{I}\hphantom{J}j} J_{IJ} \delta \Gamma_{i}\psi - \frac{i}{2} \bar{\psi} \tilde{\omega}^{IJ}_{\hphantom{I}\hphantom{J}i} J_{IJ} \delta \Gamma_{j}\psi = \\
    = & -\frac{1}{2} \dot{a} \bigg[\delta \bar{\psi} \tilde{\Gamma}_{i} \gamma_0 \gamma_j \psi + \delta \bar{\psi} \tilde{\Gamma}_{j} \gamma_0 \gamma_i \psi + \bar{\psi} \delta \Gamma_{i} \gamma_0 \gamma_j \psi + \bar{\psi} \delta \Gamma_{j} \gamma_0 \gamma_i \psi + \bar{\psi} \gamma_0 \gamma_j  \tilde{\Gamma}_{i}\delta \psi + \\
    & + \bar{\psi} \gamma_0 \gamma_i \tilde{\Gamma}_{j}\delta \psi + \bar{\psi} \gamma_0 \gamma_j \delta \Gamma_{i}\psi + \bar{\psi} \gamma_0 \gamma_i \delta \Gamma_{j}\psi \bigg] = - H a^2 \bigg[ \delta_{ij} \delta \bar{\psi} \gamma_0 \psi - \delta_{ij} \bar{\psi} \gamma_0 \delta \psi \bigg],  
    \end{split}
    \label{Pert_ID}
\end{equation}
which, substituted in eq. \ref{SETensij_1}, leads to

\begin{equation}
    \begin{split}
    \delta T_{i j} =  & -\frac{i}{4}\Big[ \bar{\psi} \tilde{\Gamma}_{(i} \tilde{D}_{j)}\delta \psi -\frac{i}{2}\bar{\psi} \tilde{\Gamma}_{(i} \delta \omega^{IJ}_{\hphantom{I}\hphantom{J}j)} J_{IJ} \psi - \tilde{D}_{(j} \delta \bar{\psi} \tilde{\Gamma}_{i)}\psi -\frac{i}{2}\bar{\psi} \delta \omega^{IJ}_{\hphantom{I}\hphantom{J}(j} J_{IJ} \tilde{\Gamma}_{i)}\psi \\
    & - H a^2 \delta _{ij} \delta \bar{\psi} \gamma_0 \psi + H a^2 \delta _{ij} \bar{\psi} \gamma_0 \delta \psi \Big] + 2 a^2 \Phi \delta_{ij} \mathcal{E} + a^2 \delta_{ij} \delta \mathcal{E}. \\
    \end{split}
    \label{SETensij_int}
\end{equation}
Furthermore, if we take into account the result obtained in eq. \ref{PertConn}, the following identity holds:

\begin{equation}
    \begin{split}
    & -\frac{i}{2}\bar{\psi} \tilde{\Gamma}_{(i} \delta \omega^{IJ}_{\hphantom{I}\hphantom{J}j)} J_{IJ} \psi -\frac{i}{2}\bar{\psi} \delta \omega^{IJ}_{\hphantom{I}\hphantom{J}(j} J_{IJ} \tilde{\Gamma}_{i)}\psi = \\
    & -\frac{1}{2}\bar{\psi} \tilde{\Gamma}_{(i} \Bigg[ \Big[ \dot{a} \big( \Phi - \Psi \big) + a \dot{\Phi}\Big] \gamma_0 \gamma_{j)} + \delta^{kl} \partial_l \Phi \gamma_{j)} \gamma_k + \partial_{j)} \Phi \Bigg] \psi +\\
    &-\frac{1}{2}\bar{\psi} \Bigg[ \Big[ \dot{a} \big( \Phi - \Psi \big) + a \dot{\Phi}\Big] \gamma_0 \gamma_{(j} + \delta^{kl} \partial_l \Phi \gamma_{(j} \gamma_k + \partial_{(j} \Phi \Bigg] \tilde{\Gamma}_{i)} \psi = \\
    = & -\frac{1}{2} \bar{\psi} \Bigg[ \tilde{\Gamma}_{(i} \delta^{kl} \partial_l \Phi \gamma_{j)} \gamma_k + \tilde{\Gamma}_{(i}\partial_{j)} \Phi + \delta^{kl} \partial_l \Phi \gamma_{(j} \gamma_k \tilde{\Gamma}_{i)} + \partial_{(j} \Phi \tilde{\Gamma}_{i)}\Bigg] \psi = 0.
    \end{split}
    \label{Pert_conn_ID}
\end{equation}
In the end, if we consider that

\begin{equation}
    \begin{split}
     & \bar{\psi} \tilde{\Gamma}_{(i} \tilde{D}_{j)}\delta \psi - \tilde{D}_{(j} \delta \bar{\psi} \tilde{\Gamma}_{i)}\psi = \\
     & \bar{\psi} \tilde{\Gamma}_{(i} \partial_{j)}\delta \psi - \partial_{(j} \delta \bar{\psi} \tilde{\Gamma}_{i)}\psi -\frac{i}{2}\bar{\psi} \tilde{\Gamma}_{(i} \tilde{\omega}^{IJ}_{\hphantom{I}\hphantom{J}j)} J_{IJ} \delta \psi -\frac{i}{2} \delta \bar{\psi} \tilde{\omega}^{IJ}_{\hphantom{I}\hphantom{J}(j} J_{IJ} \tilde{\Gamma}_{i)} \psi = \\
    = & \bar{\psi} \tilde{\Gamma}_{(i} \partial_{j)}\delta \psi - \partial_{(j} \delta \bar{\psi} \tilde{\Gamma}_{i)}\psi + H a^2 \delta_{ij} \delta \bar{\psi} \gamma_0 \psi - H a^2 \delta_{ij} \bar{\psi} \gamma_0 \delta \psi,
    \end{split}
    \label{Pert_ID_2}
\end{equation}
we arrive at

\begin{equation}
    \begin{split}
    \delta T_{i j} =  & -\frac{i}{4}\Big[ \bar{\psi} \Gamma_{(i} \partial_{j)}\delta \psi  - \partial_{(j} \delta \bar{\psi} \Gamma_{i)}\psi + H a^2 \delta_{ij} \delta \bar{\psi} \gamma_0 \psi - H a^2 \delta_{ij} \bar{\psi} \gamma_0 \delta \psi\\
    & - H a^2 \delta _{ij} \delta \bar{\psi} \gamma_0 \psi + H a^2 \delta _{ij} \bar{\psi} \gamma_0 \delta \psi \Big] + 2 a^2 \Phi \delta_{ij} \mathcal{E} + a^2 \delta_{ij} \delta \mathcal{E} = \\
    = & -\frac{i}{4}\Big[ \bar{\psi} \Gamma_{(i} \partial_{j)}\delta \psi  - \partial_{(j} \delta \bar{\psi} \Gamma_{i)}\psi \Big] + 2 a^2 \Phi \delta_{ij} \mathcal{E} + a^2 \delta_{ij} \delta \mathcal{E}.
    \end{split}
    \label{SETensij_2}
\end{equation}

\subsection{Covariant $0i$-component}
The last component of the SET that we have to compute is the covariant 0$i$-component:

\begin{equation}
    \begin{split}
    \delta T_{0 i} =  &-\frac{i}{4}\Big[ \delta \bar{\psi} \tilde{\Gamma}_{(0} \tilde{D}_{i)}\psi + \bar{\psi} \tilde{\Gamma}_{(0} \tilde{D}_{i)}\delta \psi -\frac{i}{2}\bar{\psi} \tilde{\Gamma}_{(0} \delta \omega^{IJ}_{\hphantom{I}\hphantom{J}i)} J_{IJ} \psi + \bar{\psi} \delta \Gamma_{(0} \tilde{D}_{i)}\psi + \\
   & - \tilde{D}_{(i} \delta \bar{\psi} \tilde{\Gamma}_{0)}\psi - \tilde{D}_{(i} \bar{\psi} \tilde{\Gamma}_{0)}\delta \psi -\frac{i}{2}\bar{\psi} \delta \omega^{IJ}_{\hphantom{I}\hphantom{J}(i} J_{IJ} \tilde{\Gamma}_{0)}\psi - \tilde{D}_{(i} \bar{\psi} \delta \Gamma_{0)}\psi \Big] = \\
   = & -\frac{i}{4}\Big[ \delta \bar{\psi} \tilde{\Gamma}_{(0} \partial_{i)}\psi -\frac{i}{2} \delta \bar{\psi} \tilde{\Gamma}_{(0} \tilde{\omega}^{IJ}_{\hphantom{IJ}i)}J_{IJ}\psi + \bar{\psi} \tilde{\Gamma}_{(0} \partial_{i)}\delta \psi -\frac{i}{2}\bar{\psi} \tilde{\Gamma}_{(0} \tilde{\omega}^{IJ}_{\hphantom{IJ}i)}J_{IJ}\delta \psi + \\
   & -\frac{i}{2}\bar{\psi} \tilde{\Gamma}_{(0} \delta \omega^{IJ}_{\hphantom{I}\hphantom{J}i)} J_{IJ} \psi + \bar{\psi} \delta \Gamma_{(0} \partial_{i)}\psi -\frac{i}{2}\bar{\psi} \delta \Gamma_{(0} \tilde{\omega}^{IJ}_{\hphantom{IJ}i)}J_{IJ}\psi \Big] + h.c. \hspace{2mm}.
    \end{split}
    \label{SETens0i_1}
\end{equation}
This expression seems quite difficult to handle with respect to those seen in previous subsections. Nevertheless, we can notice that a lot of terms cancel each other  out. In fact, 

\begin{equation}
    \begin{split}
    & -\frac{i}{2} \delta \bar{\psi} \tilde{\Gamma}_{(0} \tilde{\omega}^{IJ}_{\hphantom{IJ}i)}J_{IJ}\psi -\frac{i}{2} \delta \bar{\psi} \tilde{\omega}^{IJ}_{\hphantom{IJ}(i}J_{IJ}\tilde{\Gamma}_{0)}\psi = \\
    & -\frac{\dot{a}}{2} \delta \bar{\psi} \Big[ \gamma_0 \gamma_0 \gamma_i + \gamma_0 \gamma_i \gamma_0 \Big] \psi = -\frac{\dot{a}}{2} \delta \bar{\psi} \Big[ \gamma_i - \gamma_i \Big] \psi= 0,
    \end{split}
    \label{Conn_ID}
\end{equation}
as well as the terms $-\frac{i}{2}\bar{\psi} \delta \Gamma_{(0} \tilde{\omega}^{IJ}_{\hphantom{IJ}i)}J_{IJ}\psi + h.c. = -\frac{i}{2} \Psi \bar{\psi} \tilde{\Gamma}_{(0} \tilde{\omega}^{IJ}_{\hphantom{IJ}i)}J_{IJ}\psi +h.c.$ or the terms like those above, but with $\delta \bar{\psi}$ and $\psi$ inverted, since the gamma matrices involved are the same.
\newline
All these results lead to

\begin{equation}
    \begin{split}
    \delta T_{0 i} = & -\frac{i}{4}\Big[ \delta \bar{\psi} \tilde{\Gamma}_{(0} \partial_{i)}\psi + \bar{\psi} \tilde{\Gamma}_{(0} \partial_{i)}\delta \psi -\frac{i}{2}\bar{\psi} \tilde{\Gamma}_{(0} \delta \omega^{IJ}_{\hphantom{I}\hphantom{J}i)} J_{IJ} \psi + \bar{\psi} \delta \Gamma_{(0} \partial_{i)}\psi \Big] + h.c. \hspace{2mm}.
    \end{split}
    \label{SETens0i_int}
\end{equation}
Furthermore, if we consider, as remembered before, that the background spinorial field does not depend on spatial coordinates, it is straightforward to see that

\begin{equation}
    \bar{\psi} \delta \Gamma_{(0} \partial_{i)}\psi - \partial_{(i} \bar{\psi} \delta \Gamma_{0)}\psi = \bar{\psi} \delta \Gamma_{i} \dot{\psi} - \dot{\bar{\psi}} \delta \Gamma_{i}\psi = a\Phi \Big[ \bar{\psi} \gamma_{i} \dot{\psi} - \dot{\bar{\psi}} \gamma_{i}\psi \Big],
    \label{Pert_Gam_ID_2a}
\end{equation}
which is equal to zero, since

\begin{equation}
    \begin{split}
    & \bar{\psi} \gamma_{i} \dot{\psi} - \dot{\bar{\psi}} \gamma_{i}\psi =  \bar{\psi} \gamma_{i} \bigg( -\frac{3}{2}H\psi +i \gamma_0 \frac{\delta W}{\delta \bar{\psi}} \bigg) - \bigg( -\frac{3}{2}H\bar{\psi} +i \bigg(\frac{\delta W}{\delta \bar{\psi}} \bigg)^{\dagger} \bigg) \gamma_i \psi = \\
    = & i\bar{\psi} \gamma_{i} \bigg( \gamma_0 \frac{\delta W}{\delta \bar{\psi}} \bigg) - i\bigg( \bigg(\frac{\delta W}{\delta \bar{\psi}} \bigg)^{\dagger} \bigg) \gamma_i \psi = \\
    = & i \Big[\bar{\psi} \gamma_{i} \gamma_0 \psi \big( U' +2 \xi E \big) -2i \xi B \bar{\psi} \gamma_{i} \gamma_0 \gamma^5 \psi -  \psi^\dagger \gamma_i \psi \big( U' + 2 \xi E \big) - 2i \xi B \psi^\dagger \gamma^5 \gamma_i \psi  \Big] = \\
    = & i \Big[\bar{\psi} \gamma_{i} \gamma_0 \psi \big( U' +2 \xi E \big) -2i \xi B \bar{\psi} \gamma_{i} \gamma_0 \gamma^5 \psi -  \bar{\psi} \gamma_i \gamma_0 \psi \big( U' + 2 \xi E \big) + 2i \xi B \bar{\psi} \gamma_i \gamma_0 \gamma^5 \psi  \Big] = 0.
    \end{split}
    \label{Pert_Gam_ID_2b}
\end{equation}
If we add to this the fact that

\begin{equation}
    \begin{split}
    & -\frac{i}{2}\bar{\psi} \tilde{\Gamma}_{(0} \delta \omega^{IJ}_{\hphantom{I}\hphantom{J}i)} J_{IJ} \psi -\frac{i}{2}\bar{\psi} \delta \omega^{IJ}_{\hphantom{I}\hphantom{J}(i} J_{IJ} \tilde{\Gamma}_{0)} \psi = \\ = & -\frac{i}{2} \bar{\psi} \Big[-i \delta^{jl} \partial_l \Psi \gamma_i \gamma_0 \gamma_j -i \delta^{jl} \partial_l \Psi \gamma_0 \gamma_j \gamma_i -i \Big[ \dot{a} \big( \Phi - \Psi \big) + a \dot{\Phi}\Big] \gamma_i \\
    & - i \delta^{jl} \partial_l \Phi \gamma_0 \gamma_i \gamma_j -2i \partial_i \Phi \gamma_0 -i \Big[ \dot{a} \big( \Phi - \Psi \big) + a \dot{\Phi}\Big] \gamma_0 \gamma_i \gamma_0 - i \delta^{jl} \partial_l \Phi \gamma_i \gamma_j \gamma_0 \Big] \psi = \\
    = & -\frac{i}{2} \bar{\psi} \Big[i \delta^{jl} \partial_l \Psi \gamma_0 \gamma_i \gamma_j + i \delta^{jl} \partial_l \Psi \gamma_0 \gamma_i \gamma_j +2i \delta_{ij}\delta^{jl} \partial_l \Psi \gamma_0 -2i \partial_i \Phi \gamma_0 - 2i \delta^{jl} \partial_l \Phi \gamma_0 \gamma_i \gamma_j \Big] \psi = \\
    = & -\frac{i}{2} \bar{\psi} \Big[ 2i \delta^{jl} \partial_l \Big( \Psi - \Phi \Big) \gamma_0 \gamma_i \gamma_j +2i \gamma_0 \partial_i \Big( \Psi - \Phi \Big) \Big] \psi = \\
    = &  \delta^{jl} \partial_l \Big( \Psi - \Phi \Big) \bar{\psi} \gamma_0 \gamma_i \gamma_j \psi + \partial_i \Big( \Psi - \Phi \Big) \bar{\psi} \gamma_0 \psi,
    \end{split}
    \label{Pert_Conn_ID_3}
\end{equation} 
the expression of the covariant $0i$-component of the perturbed SET, in the end, becomes

\begin{equation}
    \begin{split}
    \delta T_{0 i} = & -\frac{i}{4}\Big[ \delta \bar{\psi} \tilde{\Gamma}_{i} \dot{\psi} + \bar{\psi} \tilde{\Gamma}_{(0} \partial_{i)}\delta \psi - \dot{\bar{\psi}} \tilde{\Gamma}_i \delta \psi - \partial_{(i} \delta \bar{\psi} \tilde{\Gamma}_{0)} \psi + \\
    & + \delta^{jl} \partial_l \Big( \Psi - \Phi \Big) \bar{\psi} \gamma_0 \gamma_i \gamma_j \psi + \partial_i \Big( \Psi - \Phi \Big) \bar{\psi} \tilde{\Gamma}_0 \psi \Big].
    \end{split}
    \label{SETens0i_2}
\end{equation}

\subsection{Mixed components}
Now, we can compute the mixed components of the perturbed SET.
\newline
From the general theory of cosmological perturbations, we know that, at first order,

\begin{equation}
    \delta T^\mu_{\hphantom{\mu}\nu} = g^{\mu \sigma} \delta T_{\sigma \nu} + \delta g^{\mu \sigma}T_{\sigma \nu} = g^{\mu \sigma} \delta T_{\sigma \nu} - g^{\mu \rho} \delta g_{\rho \lambda} g ^{\lambda \sigma} T_{\sigma \nu}.
    \label{Ind_Change}
\end{equation}
This relation directly leads to

\begin{equation}
    \begin{split}
    & \delta T^0_{\hphantom{0}0} = g^{0 \sigma} \delta T_{\sigma 0} - g^{0 \rho} \delta g_{\rho \lambda} g ^{\lambda \sigma} T_{\sigma 0} = g^{0 0} \delta T_{0 0} - g^{00} \delta g_{00} g^{00} T_{00} = -\delta T_{0 0} + 2 \Psi T_{00}, \\
    & \delta T^0_{\hphantom{0}i} = g^{0 \sigma} \delta T_{\sigma i} - g^{0 \rho} \delta g_{\rho \lambda} g ^{\lambda \sigma} T_{\sigma i} = g^{00} \delta T_{0i} - g^{00} \delta g_{0m} g ^{ml} T_{li} = -\delta T_{0i} = - a^2 \delta T^i_{\hphantom{i}0}, \\
    & \delta T^i_{\hphantom{i}j} = g^{i m} \delta T_{m j} - g^{i n} \delta g_{n m} g ^{m l} T_{l j} =  g^{i m} \delta T_{m j} - \frac{2 \Phi}{a^2}\delta^{in} \delta_{nm}  \delta^{ml}T_{l j} = g^{i m} \delta T_{m j} - 2 \Phi g^{il}T_{l j}.
    \end{split}
    \label{SET_Mixed_Comp}
\end{equation}
Therefore, the mixed components we were looking for are

\begin{equation}
    \begin{split}
   \delta T^0_{\hphantom{0}0} = &-\frac{i}{2}\Big[ \delta \bar{\psi} \tilde{\Gamma}^{0} \dot{\psi} + \bar{\psi} \tilde{\Gamma}^{0} (\dot{\delta \psi}) - ( \dot{\delta \bar{\psi}}) \tilde{\Gamma}^{0}\psi - \dot{\bar{\psi}} \tilde{\Gamma}^{0}\delta \psi \Big] - \Psi \Big[ W - \mathcal{E}  \Big] + \delta \mathcal{E} + 2 \Psi W = \\
   = & -\frac{i}{2}\Big[ \delta \bar{\psi} \tilde{\Gamma}^{0} \dot{\psi} + \bar{\psi} \tilde{\Gamma}^{0} (\dot{\delta \psi}) - ( \dot{\delta \bar{\psi}}) \tilde{\Gamma}^{0}\psi - \dot{\bar{\psi}} \tilde{\Gamma}^{0}\delta \psi \Big] + \Psi \Big[ W + \mathcal{E}  \Big] + \delta \mathcal{E},  
    \end{split}
    \label{SETens00Mix_Derived}
\end{equation}

\begin{equation}
    \begin{split}
    \delta T^0_{\hphantom{0} i} = &\frac{i}{4}\Big[ \delta \bar{\psi} \tilde{\Gamma}_{i} \dot{\psi} + \bar{\psi} \tilde{\Gamma}_{(0} \partial_{i)}\delta \psi - \dot{\bar{\psi}} \tilde{\Gamma}_i \delta \psi - \partial_{(i} \delta \bar{\psi} \tilde{\Gamma}_{0)} \psi + \\
    & + \delta^{jl} \partial_l \Big( \Psi - \Phi \Big) \bar{\psi} \gamma_0 \gamma_i \gamma_j \psi + \partial_i \Big( \Psi - \Phi \Big) \bar{\psi} \tilde{\Gamma}_0 \psi \Big],
    \end{split}
    \label{SETens0iMix_Derived}
\end{equation}

\begin{equation}
    \begin{split}
    \delta T^i_{\hphantom{i} j} = & -\frac{i}{4}g^{im}\Big[ \bar{\psi} \tilde{\Gamma}_{(m} \partial_{j)}\delta \psi  - \partial_{(j} \delta \bar{\psi} \tilde{\Gamma}_{m)}\psi \Big] + 2 \Phi \mathcal{E} \delta^i_{\hphantom{i}j} + \delta^i_{\hphantom{i}j} \delta \mathcal{E} - 2 \Phi \mathcal{E}\delta^i_{\hphantom{i}j} = \\
    = & -\frac{i}{4}g^{im}\Big[ \bar{\psi} \tilde{\Gamma}_{(m} \partial_{j)}\delta \psi  - \partial_{(j} \delta \bar{\psi} \tilde{\Gamma}_{m)}\psi \Big] + \delta^i_{\hphantom{i}j} \delta \mathcal{E}.
    \end{split}
    \label{SETensijMix_Derived}
\end{equation}

\section{Pressure and density perturbations}\label{Rel_Dens_Press}

In order to obtain the relation between the pressure perturbation and the density one, we need, as stated in the main text, the perturbed Dirac equation \ref{PertDirEq} and the perturbations of $W$ and $\mathcal{E}$. These lasts are

\begin{equation}
    \begin{split}
    & \delta W = U' \delta E + 2 \xi\big( E \delta E + B \delta B\big), \\
    & \delta \mathcal{E} =  \delta W - 2 U' \delta E + U'\delta E + U'' E \delta E = 2 \xi\big( E \delta E + B \delta B\big) + U'' E \delta E. 
    \end{split}
    \label{PertOfWandE}
\end{equation}
Considering the result obtained in eq. \ref{Pert_SET_Trace} and the expressions of the mixed components of the perturbed SET just got, we obtain

\begin{equation}
    \begin{split}
    & 2\delta \rho - 6 \delta P = -2 \delta T^\mu_{\hphantom{\mu}\mu} = \\
    = & i\Big[ \delta \bar{\psi} \tilde{\Gamma}^{0} \dot{\psi} + \bar{\psi} \tilde{\Gamma}^{0} (\dot{\delta \psi}) - ( \dot{\delta \bar{\psi}}) \tilde{\Gamma}^{0}\psi - \dot{\bar{\psi}} \tilde{\Gamma}^{0}\delta \psi \Big] -2 \Psi \Big[ W + \mathcal{E}  \Big] -2 \delta \mathcal{E} + \\
    & i \Big[ \bar{\psi} \tilde{\Gamma}^{i} \partial_{i}\delta \psi - \partial_{i} \delta \bar{\psi} \tilde{\Gamma}^{i}\psi \Big] - 6 \delta \mathcal{E} = \\
    = & i\Big[ \delta \bar{\psi} \tilde{\Gamma}^{0} \dot{\psi} - \dot{\bar{\psi}} \tilde{\Gamma}^{0}\delta \psi \Big] -2 \Psi \Big[ W + \mathcal{E}  \Big] - 8 \delta \mathcal{E} + i\Big[ \bar{\psi} \tilde{\Gamma}^{\mu} \partial_{\mu}\delta \psi - \partial_{\mu} \delta \bar{\psi} \tilde{\Gamma}^{\mu}\psi \Big].
    \end{split}
    \label{Press_Dens_rel}
\end{equation}
Furthermore, taking into account the fact that $\delta E = \delta \bar{\psi}\psi + \bar{\psi}\delta \psi$ and $\delta B = -i\delta \bar{\psi} \gamma^5 \psi -i \bar{\psi} \gamma^5 \delta \psi$, given the relations for $\delta W$ and $\delta \mathcal{E}$ presented at the beginning of this section, we have that

\begin{equation}
    \begin{split}
    & \delta \bar{\psi} \tilde{\Gamma}^{0} \dot{\psi} - \dot{\bar{\psi}} \tilde{\Gamma}^{0}\delta \psi = \delta \bar{\psi} \bigg(-\frac{3}{2}H\tilde{\Gamma}^0\psi -i\frac{\delta W}{\delta \bar{\psi}}\bigg) - \bigg(-\frac{3}{2}H\bar{\psi} +i\bigg(\frac{\delta W}{\delta \bar{\psi}}\bigg)^{\dagger} \bigg)\tilde{\Gamma}^0 \delta \psi = \\
    = & -\frac{3}{2}H\Big[\delta \bar{\psi} \tilde{\Gamma}^0 \psi - \bar{\psi} \tilde{\Gamma}^0 \delta \psi \Big] - i \Big(U' \delta E + 2 \xi E \delta E + 2 \xi B \delta B\Big) = \\
    = & -\frac{3}{2}H\Big[\delta \bar{\psi} \tilde{\Gamma}^0 \psi - \bar{\psi} \tilde{\Gamma}^0 \delta \psi \Big] - i \delta W.
    \end{split}
    \label{Pert_Gamm0_rel}
\end{equation}
Then, by exploiting the perturbed Dirac equation, we can substitute the last term in eq. \ref{Press_Dens_rel}.
In fact, the equation leads to 

\begin{equation}
    \begin{split}
    & i\Big[ \bar{\psi} \tilde{\Gamma}^{\mu} \partial_{\mu}\delta \psi - \partial_{\mu} \delta \bar{\psi} \tilde{\Gamma}^{\mu}\psi \Big] = i \bar{\psi} \tilde{\Gamma}^{\mu} \partial_{\mu}\delta \psi + h.c. = \\
    = & -i\bar{\psi}\delta \Gamma^\mu \tilde{D}_\mu \psi -\frac{1}{2}\bar{\psi}\tilde{\Gamma}^\mu \delta \omega^{IJ}_{\hphantom{IJ}\mu} J_{IJ}\psi -\frac{1}{2}\bar{\psi} \tilde{\Gamma}^\mu \tilde{\omega}^{IJ}_{\hphantom{IJ}\mu} J_{IJ} \delta \psi + \bar{\psi} \delta \Bigg( \frac{\delta W}{\delta \bar{\psi}}\Bigg) + h.c. \hspace{2mm} .
    \end{split}
    \label{Pert_Dirac_eq_2}
\end{equation}
Now, if we consider that

\begin{equation}
    \begin{split}
    -\frac{1}{2}\bar{\psi}\tilde{\Gamma}^\mu \tilde{\omega}^{IJ}_{\hphantom{IJ}\mu} J_{IJ} \delta \psi + h.c. = -i\frac{3}{2} H \bar{\psi} \tilde{\Gamma}^0 \delta \psi + i\frac{3}{2} H \delta \bar{\psi} \tilde{\Gamma}^0 \psi = -i\frac{3}{2} H \Big[ \bar{\psi} \tilde{\Gamma}^0 \delta \psi - \delta \bar{\psi} \tilde{\Gamma}^0 \psi \Big],
    \end{split}
    \label{Pert_Gamm_Rel_2b}
\end{equation}
and that, given $\delta \Gamma^\mu :=\delta g^{\mu \sigma} \tilde{\Gamma}_\sigma + g^{\mu \sigma} \delta \Gamma_\sigma$,

\begin{equation}
    \begin{split}
    & -i\bar{\psi}\delta \Gamma^\mu \tilde{D}_\mu \psi + h.c. = \\
    = & -i\bar{\psi}\delta \Gamma^0 \dot{\psi} -\frac{1}{2} \bar{\psi}\delta \Gamma^k \tilde{\omega}^{IJ}_{\hphantom{IJ}k} J_{IJ}\psi + h.c. =  i \Psi \bar{\psi} \tilde{\Gamma}^0 \dot{\psi} + \frac{1}{2} \Phi \bar{\psi} \tilde{\Gamma}^k \tilde{\omega}^{IJ}_{\hphantom{IJ}k} J_{IJ}\psi + h.c. = \\
    = & i \Psi \Big[ \bar{\psi} \tilde{\Gamma}^0 \dot{\psi} - \dot{\bar{\psi}} \tilde{\Gamma}^0 \psi \Big] - \Phi \bigg[ -\frac{1}{2}\bar{\psi}\tilde{\Gamma}^\mu \tilde{\omega}^{IJ}_{\hphantom{IJ}\mu} J_{IJ} \psi + h.c. \bigg] = \\
    = & 2 \Psi \Big[\mathcal{E} + W \Big] + i \Phi \frac{3}{2} H \Big[ \bar{\psi} \tilde{\Gamma}^0 \psi - \bar{\psi} \tilde{\Gamma}^0 \psi \Big] = 2 \Psi \Big[\mathcal{E} + W \Big],
    \end{split}
    \label{Pert_Gamm_Rel_3}
\end{equation}
we arrive at the relation:

\begin{equation}
    \begin{split}
    & 2\delta \rho - 6 \delta P = \\
    = & i\bigg[ -\frac{3}{2}H\Big[\delta \bar{\psi} \tilde{\Gamma}^0 \psi - \bar{\psi} \tilde{\Gamma}^0 \delta \psi \Big] - i \delta W \bigg] -2 \Psi \Big[ W + \mathcal{E}  \Big] - 8 \delta \mathcal{E} + 2 \Psi \Big[\mathcal{E} + W \Big] + \\
    &-i\frac{3}{2} H \Big[ \bar{\psi} \tilde{\Gamma}^0 \delta \psi - \delta \bar{\psi} \tilde{\Gamma}^0 \psi \Big] + \Bigg[ -\frac{1}{2}\bar{\psi}\tilde{\Gamma}^\mu \delta \omega^{IJ}_{\hphantom{IJ}\mu} J_{IJ}\psi + \bar{\psi} \delta \Bigg( \frac{\delta W}{\delta \bar{\psi}}\Bigg) + h.c. \Bigg] = \\
    = & \delta W - 8 \delta \mathcal{E} + \Bigg[ -\frac{1}{2}\bar{\psi}\tilde{\Gamma}^\mu \delta \omega^{IJ}_{\hphantom{IJ}\mu} J_{IJ}\psi + \bar{\psi} \delta \Bigg( \frac{\delta W}{\delta \bar{\psi}}\Bigg) + h.c. \Bigg].
    \end{split}
    \label{Press_Dens_rel_int}
\end{equation}
The latter can be further simplified, considering that

\begin{equation}
    \begin{split}
    & \bar{\psi} \delta \Bigg( \frac{\delta W}{\delta \bar{\psi}}\Bigg) + h.c. = \bar{\psi} \delta \Bigg( \frac{\delta W}{\delta \bar{\psi}}\Bigg) + \delta \Bigg( \frac{\delta W}{\delta \bar{\psi}}\Bigg)^{\dagger} \gamma^0 \psi = \\
    = & \Big( U' + 2 \xi E \Big) \delta E + 2 \xi B \delta B + \Big( 2U'' E + 4 \xi E \Big) \delta E + 4 \xi B \delta B = \\
    = & \delta W + 2 \delta \mathcal{E}
    \end{split}
    \label{Pert_Delta_deltaW}
\end{equation}
and that

\begin{equation}
    \begin{split}
     & -\frac{1}{2}\bar{\psi}\tilde{\Gamma}^\mu \delta \omega^{IJ}_{\hphantom{IJ}\mu} J_{IJ}\psi + h.c. = -\frac{1}{2}\bar{\psi}\tilde{\Gamma}^\mu \delta \omega^{IJ}_{\hphantom{IJ}\mu} J_{IJ}\psi -\frac{1}{2}\bar{\psi} \delta \omega^{IJ}_{\hphantom{IJ}\mu} J_{IJ}\tilde{\Gamma}^\mu\psi = \\
     = & -\frac{1}{2} \bar{\psi} \bigg[  -\frac{i}{a} \tilde{\Gamma}^0 \delta^{il} \partial_l \Psi \gamma_0 \gamma_i -i \tilde{\Gamma}^k \Big[ \dot{a} \big( \Phi - \Psi \big) + a \dot{\Phi}\Big] \gamma_0 \gamma_k - i \tilde{\Gamma}^k \delta^{jl} \partial_l \Phi \gamma_k \gamma_j -i \tilde{\Gamma}^k \partial_k \Phi + \\
     & -\frac{i}{a} \delta^{il} \partial_l \Psi \gamma_0 \gamma_i \tilde{\Gamma}^0 -i \Big[ \dot{a} \big( \Phi - \Psi \big) + a \dot{\Phi}\Big] \gamma_0 \gamma_k \tilde{\Gamma}^k - i \delta^{jl} \partial_l \Phi \gamma_k \gamma_j \tilde{\Gamma}^k -i \tilde{\Gamma}^k \partial_k \Phi\bigg] \psi = \\
     = & -\frac{1}{2} \bar{\psi} \Big[- i \tilde{\Gamma}^k \delta^{jl} \partial_l \Phi \gamma_k \gamma_j - i \delta^{jl} \partial_l \Phi \gamma_k \gamma_j \tilde{\Gamma}^k -2i \tilde{\Gamma}^k \partial_k \Phi\Big] \psi = \\
     = & -\frac{1}{2} \bar{\psi} \Big[- i \tilde{\Gamma}^k \delta^{jl} \partial_l \Phi \gamma_k \gamma_j + i \tilde{\Gamma}^k \delta^{jl} \partial_l \Phi \gamma_k \gamma_j + 2i \delta^{jl} \partial_l \Phi \delta_{kj} \tilde{\Gamma}^k -2i \tilde{\Gamma}^k \partial_k \Phi\Big] \psi = \\
     = &  -\frac{1}{2} \bar{\psi} \Big[2i \partial_k \Phi \tilde{\Gamma}^k -2i \tilde{\Gamma}^k \partial_k \Phi\Big] \psi = 0.
    \end{split}
    \label{Pert_Conn_ID_4}
\end{equation}
These results give

\begin{equation}
    \begin{split}
    & 2\delta \rho - 6 \delta P = \\
    = & \delta W - 8 \delta \mathcal{E} + \delta W + 2 \delta \mathcal{E} = 2 \delta W - 6 \delta \mathcal{E} = 2 U' \delta E - 2 U'' E \delta E - 4 \delta \mathcal{E},
    \end{split}
    \label{PressDens_1}
\end{equation}
which, in the end, leads to the relation we were looking for:

\begin{equation}
     \delta \rho = 3 \delta P + U' \delta E - U'' E \delta E - 2  \delta \mathcal{E}.
    \label{PressDens_2}
\end{equation}

\chapter{The (1+1+2)-Decomposition in the Cosmological Background}\label{112_Decomp_Cosmo_Back}
In this appendix, we carry out all calculations required to prove that the SET of the spinorial fluid, in the context of the cosmological background, can be recast as that of a Relativistic Dust.
\newline
As anticipated in Section \ref{Cosmo_SET_Decomp} of Chapter \ref{SET_Decomp}, we need to demonstrate that the quantities defined in eqs. \ref{Scalars_Spin_Field_Def}, \ref{Vectors_Spin_Field_Def}, \ref{Tensor_Spin_Field_Def} and \ref{Dir_Der_Spin_Field_Def} vanish. We will begin from scalar ones. Then, we will pass to vector ones, successively to the tensor one and, in the end, to directional derivatives.
\newline
\newline
We remember that the unique non-vanishing Christoffel Symbols, in this case, are

\begin{equation}
    \Gamma^0_{\hphantom{0}ij} = a \dot{a} \delta_{ij}, \hspace{3cm} \Gamma^i_{\hphantom{i}0j} = H \delta^i_{\hphantom{i}j}.
    \label{Christ_Symb_Cosmo_Back}
\end{equation}
Since we cannot obtain any result without making an assumption on the four-velocity, we will assume that, in the Coordinate RF, $u^i = 0$. In this way, we will be able to prove that the background SET can be written as that of a Relativistic Dust if and only if the fluid is at rest in the Coordinate RF.

\section{Scalar quantities}
Since $\epsilon^{\mu \nu}$ is antisymmetric, we can substitute the covariant derivatives in eq. \ref{Scalars_Spin_Field_Def} with partial derivatives. This leads to the following expressions of the scalar quantities:

\begin{equation}
    \Omega = \frac{1}{2} \partial_\mu u_\nu \epsilon^{\mu \nu}, \hspace{3cm}  \xi = \frac{1}{2}\partial_\mu s_\nu \epsilon^{\mu \nu}.
    \label{Scalars_Spin_Field_Def_2}
\end{equation}
It is immediate to notice that $\Omega$ vanishes, since $u^\mu$ must be independent of spatial coordinates, as the Cosmological Principle requires.
\newline
Regarding $\xi$, instead, all partial derivatives of $s_\mu$, but $\partial_0 s_j$, vanish. However, it is easy to see that the contribution due to this factor is null, since the only non-vanishing component of $u_\mu$ is $u_0$ and $\epsilon^{0j0\rho} = 0$.

\section{Vector quantities}
If we consider that, as a consequence of our choice about $u^\mu$ in this RF, $N_0^{\hphantom{0}0} = N_0^{\hphantom{0}i} = N_i^{\hphantom{i}0} = 0$, we immediately obtain that $\Sigma_0 = \Omega^0 = \mathcal{A}_0 = a_0 = 0$.
\newline
Then, taking into account the explicit expressions of the Christoffel Symbols and the fact that $1 = s_\mu s^\mu = s^i s^j g_{ij} = s^i s^j \delta_{ij}a^2$, we get 

\begin{equation}
    \Sigma_i = \frac{1}{2}N_i^{\hphantom{i}j}s^l \nabla_{(j}u_{l)} = N_i^{\hphantom{i}j}s^l \Gamma^0_{\hphantom{0}jl}u_0 = N_i^{\hphantom{i}j}s^l a^2\delta_{jl} H u_0 = N_i^{\hphantom{i}j}s_j H u_0 = 0,
    \label{Sigma_Identity}
\end{equation}

\begin{equation}
    \Omega^i = \frac{1}{2}N^i_{\hphantom{i}j}\epsilon^{j0\sigma \tau} u_0\nabla_\sigma u_\tau = \frac{1}{2}N^i_{\hphantom{i}j}\epsilon^{j0\sigma \tau} u_0 \partial_\sigma u_\tau = 0,
    \label{Omega_Identity}
\end{equation}

\begin{equation}
    \mathcal{A}_i = N_i^{\hphantom{i}j} u^0\nabla_0 u_j = N_i^{\hphantom{i}j} u^0\Gamma^0_{\hphantom{0}0j}u_0 = 0,
    \label{A_Cors_Identity}
\end{equation}

\begin{equation}
    a_i = N_i^{\hphantom{i}j} s^k\nabla_k s_j = N_i^{\hphantom{i}j} s^k\Gamma^l_{\hphantom{l}kj}s_l = 0.
    \label{a_Identity}
\end{equation}
Therefore, under the assumption made, all vectors quantities vanish.

\section{The tensor quantity and the directional derivatives}
It is straightforward to see, taking into account our hypotheses, that $\nabla_\rho u_\sigma$ is different from zero only in one case: $\nabla_i u_j = \Gamma^0_{\hphantom{0}ij}u_0 = a^2 H \delta_{ij}u_0$.
\newline
Hence, given the expression of $\Sigma_{\mu \nu}$ and remembering what we stated in the previous section about $N^\nu_{\hphantom{\nu}\mu}$, we obtain

\begin{equation}
    \begin{split}
        \Sigma_{\mu \nu} =&  \frac{1}{2}(N^i_{\hphantom{i}\mu}N^j_{\hphantom{j}\nu} + N^i_{\hphantom{i}\nu}N^j_{\hphantom{j}\mu} - N_{\mu \nu}N^{ij})\Gamma^0_{\hphantom{0}ij}u_0 = \\
        = &\frac{1}{2}(N^i_{\hphantom{i}\mu}N^j_{\hphantom{j}\nu} + N^i_{\hphantom{i}\nu}N^j_{\hphantom{j}\mu} - N_{\mu \nu}N^{ij})H g_{ij}u_0 = \\
        =& \frac{1}{2}(2N_{\mu \nu} - N_{\mu \nu}N^{ij}g_{ij})H u_0 = (N_{\mu \nu} - N_{\mu \nu}) H u_0 = 0.  
    \end{split}
    \label{Tensor_Sigma_Identity}
\end{equation}
Now, it is left to be proven that the directional derivatives vanish. 
\newline
\newline
Since $\beta$ is a pseudo-scalar, its covariant derivative coincides with the ordinary one. Moreover, we know that $\beta = 0, \pi$. Therefore, we conclude that all its directional derivatives vanish.
\newline
Regarding $\ln{\phi^2}$, instead, since it is proportional to the logarithm of $a$ and the scale parameter depends only on time, we have that the unique directly non-vanishing component is

\begin{equation}
    \delta_0 = N_0^{\hphantom{0}0}\partial_0 \ln{\phi^2},
    \label{Dir_Der_Phi}
\end{equation}
which is equal to zero, since $N_0^{\hphantom{0}0} = 0$.
\newline
\newline
In conclusion, we have proven that the only one thermodynamical quantity of the spinorial fluid which, in the context of the cosmological background, does not vanish is $\rho = 2 \phi^2 m \cos{\beta}$.

\chapter{Computation of the SET in the Spherically Symmetric case}\label{Append_Spher_Symm}

In this appendix, we perform all calculations required to derive the explicit expressions of the Dirac equation and of the components of the SET reported in Chapter \ref{Spher_Symm}.

\section{Derivation of the spherically symmetric connection}\label{Connection_Spher_Symm}
In the same way as in Appendix \ref{Computation_Scalar_Pert}, in order to derive the components of the connection in the spherically symmetric case, we have to solve the structure equations \ref{Structure_Eqs} with $T^I = 0$. 
\newline
\newline
Considering the tetrad specified in eq. \ref{Tetr_Spher} and the fact that $\nu$ and $\lambda$ are functions only of $r$ and $t$, we get  

\begin{equation}
    \begin{split}
    de^0  & = - \frac{\nu'}{2} e^{\frac{\nu}{2}} dt \wedge dr,\\
    de^1 & = \frac{\dot{\lambda}}{2} e^{\frac{\lambda}{2}} dt \wedge dr,\\
    de^2 & = dr \wedge d\theta,\\
    de^3 & = \sin{\theta} \hspace{0.5mm} dr \wedge d\phi + r \cos{\theta} \hspace{0.5mm} d\theta \wedge d\phi,\\
    \end{split}
    \label{Struct_eqs_Spher}
\end{equation}
where we have denoted with $\dot{}$ the derivative with respect to $t$ and with $'$ that with respect to $r$, in order to simplify our notations. 
\newline
Thanks to these expressions we obtain

\begin{equation}
    \begin{split}
    &\omega^0_{\hphantom{0} 1} = \frac{\nu'}{2} e^{\frac{\nu - \lambda}{2}} dt + (\propto dr),\\
    &\omega^1_{\hphantom{1} 0} = \frac{\dot{\lambda}}{2} e^{\frac{\lambda - \nu}{2}} dr + (\propto dt),\\
    &\omega^2_{\hphantom{2} 1} = d\theta + (\propto dr),\\
    &\omega^3_{\hphantom{3} 1} = \sin{\theta} \hspace{0.5mm} d\phi + (\propto dr),\\
    &\omega^3_{\hphantom{3} 2} = r \cos{\theta} \hspace{0.5mm} d\phi + (\propto d\theta),\\
    \end{split}
    \label{SpherConn}
\end{equation}
the other terms being zero.
\newline
The meaning of the notation $\propto$ is the same as that in eq. \ref{NewConn}. Therefore, in order to obtain the complete expressions of these quantities, we have to use the second structure equation.
\newline
This leads to

\begin{equation}
    \begin{split}
    &\omega^{01} = \frac{\nu'}{2} e^{\frac{\nu - \lambda}{2}} dt + \frac{\dot{\lambda}}{2} e^{\frac{\lambda - \nu}{2}} dr  = - \omega^{10},\\
    &\omega^{21} = d\theta = - \omega^{12},\\
    &\omega^{31} = \sin{\theta} \hspace{0.5mm} d\phi = - \omega^{13},\\
    &\omega^{32} = r \cos{\theta} \hspace{0.5mm} d\phi = - \omega^{23},\\
    & \omega^{00} = \omega^{11} = \omega^{22} = \omega^{33} = \omega^{02} = \omega^{03} = \omega^{30} = \omega^{20} = 0.
    \end{split}
    \label{SpherConn_1}
\end{equation}
From these results we can obtain the expression of the connection defined on the Spin Bundle. 
\newline
In particular, we arrive at

\begin{equation}
    \begin{split}
    \omega^{IJ}_{\hphantom{IJ} 0} J_{IJ} = & 2\omega^{0i}_{\hphantom{IJ} 0} J_{0i} + \omega^{ij}_{\hphantom{IJ} 0} J_{ij} = 2\omega^{01}_{\hphantom{IJ} 0} J_{01} = -\frac{i\nu'}{2} e^{\frac{\nu - \lambda}{2}} \gamma_0 \gamma_1,\\
    \omega^{IJ}_{\hphantom{IJ} r} J_{IJ} = & 2\omega^{0i}_{\hphantom{IJ} r} J_{0i} + \omega^{ij}_{\hphantom{IJ} r} J_{ij} = 2\omega^{01}_{\hphantom{IJ} r} J_{01} = -\frac{i\dot{\lambda}}{2} e^{\frac{\lambda - \nu}{2}} \gamma_0 \gamma_1,\\
    \omega^{IJ}_{\hphantom{IJ} \theta} J_{IJ} = & 2\omega^{0i}_{\hphantom{IJ} \theta} J_{0i} + \omega^{ij}_{\hphantom{IJ} \theta} J_{ij} = \omega^{12}_{\hphantom{IJ} \theta} J_{12} + \omega^{21}_{\hphantom{IJ} \theta} J_{21} = 2\omega^{12}_{\hphantom{IJ} \theta} J_{12} = i \gamma_1 \gamma_2,\\
    \omega^{IJ}_{\hphantom{IJ} \phi} J_{IJ} = & 2\omega^{0i}_{\hphantom{IJ} \phi} J_{0i} + \omega^{ij}_{\hphantom{IJ} \phi} J_{ij} = \omega^{13}_{\hphantom{IJ} \phi} J_{13} + \omega^{23}_{\hphantom{IJ} \phi} J_{23} + \omega^{31}_{\hphantom{IJ} \phi} J_{31} + \omega^{32}_{\hphantom{IJ} \phi} J_{32} = \\
    = & 2\omega^{13}_{\hphantom{IJ} \phi} J_{13} + 2\omega^{23}_{\hphantom{IJ} \phi} J_{23} = i \big[\sin{\theta} \gamma_1 \gamma_3 + r \cos{\theta} \gamma_2 \gamma_3 \big].
    \end{split}
    \label{SpherSpinConn}
\end{equation}

\section{Derivation of the Dirac equation and of the SET}\label{Der_Spher_Symm_SET}

We notice that some of the components of the connection depend on $\theta$. Hence, as anticipated in the main text, since, in all the equations of motion, these components do not cancel each other out (as we will see), also the spinorial field must depend on the azimuthal angle $\theta$.
\newline
This fact leads to the expression of the Dirac equation

\begin{equation}
    \begin{split}
    i\Gamma^\mu D_\mu \psi = & i \Gamma^0 D_0 \psi + i\Gamma^r D_r \psi + i\Gamma^\theta D_\theta \psi + i\Gamma^\phi D_\phi \psi = \\
    = &-ie^{-\frac{\nu}{2}}\gamma_0 \dot{\psi} + ie^{-\frac{\lambda}{2}}\gamma_1 \psi' + \frac{i}{r}\gamma_2 \partial_\theta \psi + \frac{i\nu'}{4}e^{-\frac{\lambda}{2}}\gamma_1 \psi -\frac{i\dot{\lambda}}{4} e^{-\frac{\nu}{2}} \gamma_0 \psi + \\
    & + \frac{i}{2r} \gamma_1\psi + \frac{i}{2 r \sin{\theta}} \big[\sin{\theta} \hspace{0.5mm} \gamma_1 + r \cos{\theta} \hspace{0.5mm} \gamma_2 \big] \psi = \\
    = & -ie^{-\frac{\nu}{2}}\gamma_0 \dot{\psi} + ie^{-\frac{\lambda}{2}}\gamma_1 \psi' + \frac{i}{r}\gamma_2 \partial_\theta \psi + \frac{i\nu'}{4}e^{-\frac{\lambda}{2}}\gamma_1 \psi -\frac{i\dot{\lambda}}{4} e^{-\frac{\nu}{2}} \gamma_0 \psi + \\
    & + \frac{i}{r} \gamma_1\psi + \frac{i}{2} \cot{\theta} \hspace{0.5mm} \gamma_2 \psi = (U'(E) + 2\xi E) \psi - 2i\xi B \gamma^5 \psi.
   \end{split}
   \label{Dirac_Eq_Spher_Sym}
\end{equation}
This result can be used to solve the Einstein equations; not before having computed the 10 independent components of the SET. We begin from diagonal ones.

\subsection{Diagonal components}
These components, due to spherical symmetry, are the simplest both in form and to compute.
\newline
In particular, they are
\begin{equation}
    \begin{split}
    T^0_{\hphantom{0} 0} = & -\frac{i}{2}\Big[\bar{\psi} \Gamma^0 D_{0} \psi - D_{0} \bar{\psi} \Gamma^0 \psi \Big] + \mathcal{E} = \\
    = & \frac{i}{2} \Big[e^{-\frac{\nu}{2}}\big(\bar{\psi} \gamma_0 \dot{\psi} - \dot{\bar{\psi}} \gamma_0 \psi \big) - \frac{\nu'}{4}e^{-\frac{\lambda}{2}} \big( \bar{\psi} \gamma_1 \psi - \bar{\psi}\gamma_1 \psi \big)\Big] + \mathcal{E} = \\
    = & \frac{i}{2} e^{-\frac{\nu}{2}}\Big[\bar{\psi} \gamma_0 \dot{\psi} - \dot{\bar{\psi}} \gamma_0 \psi\Big] + \mathcal{E},
    \end{split}
    \label{T^0_0_Spher}
\end{equation}

\begin{equation}
    \begin{split}
    T^r_{\hphantom{r} r} = & -\frac{i}{2}\Big[\bar{\psi} \Gamma^r D_{r} \psi - D_{r} \bar{\psi} \Gamma^r \psi \Big] + \mathcal{E} = \\
    = & -\frac{i}{2} \Big[e^{-\frac{\lambda}{2}}\big(\bar{\psi} \gamma_1 \psi' - \bar{\psi}' \gamma_1 \psi \big) - \frac{\dot{\lambda}}{4}e^{-\frac{\nu}{2}} \big( \bar{\psi} \gamma_0 \psi - \bar{\psi}\gamma_0 \psi \big)\Big] + \mathcal{E} = \\
    = & -\frac{i}{2} e^{-\frac{\lambda}{2}}\Big[\bar{\psi} \gamma_1 \psi' - \bar{\psi}' \gamma_1 \psi\Big] + \mathcal{E},
    \end{split}
    \label{T^r_r_Spher}
\end{equation}

\begin{equation}
    \begin{split}
    T^\theta_{\hphantom{\theta} \theta} = & -\frac{i}{2}\Big[\bar{\psi} \Gamma^\theta D_{\theta} \psi - D_{\theta} \bar{\psi} \Gamma^\theta \psi \Big] + \mathcal{E} = \\
    = & -\frac{i}{2} \Big[\frac{1}{r}\big(\bar{\psi} \gamma_2 \partial_\theta \psi - \partial_\theta \bar{\psi} \gamma_2 \psi \big) + \frac{1}{2r}\big(\bar{\psi} \gamma_1 \psi - \bar{\psi}\gamma_1 \psi \big)\Big] + \mathcal{E} = \\
    = & -\frac{i}{2r} \Big[\bar{\psi} \gamma_2 \partial_\theta \psi - \partial_\theta \bar{\psi} \gamma_2 \psi \Big] + \mathcal{E},
    \end{split}
    \label{T^theta_theta_Spher}
\end{equation}

\begin{equation}
    \begin{split}
    T^\phi_{\hphantom{\phi} \phi} = & -\frac{i}{2}\Big[\bar{\psi} \Gamma^\phi D_{\phi} \psi - D_{\phi} \bar{\psi} \Gamma^\phi \psi \Big] + \mathcal{E} = \\
    = & -\frac{i}{2} \Big[ \frac{1}{2r}\big(\bar{\psi} \gamma_1 \psi - \bar{\psi}\gamma_1 \psi \big) + \frac{1}{2} \cot{\theta} \big(\bar{\psi} \gamma_2 \psi - \bar{\psi}\gamma_2 \psi \big) \Big] + \mathcal{E} = \mathcal{E}.
    \end{split}
    \label{T^phi_phi_Spher}
\end{equation}

\subsection{Time-space components}
Now, we can pass to time-space components, which are all different from zero, in contrast to what is expected from spherical symmetry. These are
\begin{equation}
    \begin{split}
    T^0_{\hphantom{0} r} = & -\frac{i}{4}\Big[\bar{\psi} \Gamma^0 D_{r} \psi + \frac{g^{00}}{g^{rr}}\bar{\psi} \Gamma^r D_{0} \psi - c.c. \Big] = \\
    = & -\frac{i}{4} \Big[-e^{-\frac{\nu}{2}}\big(\bar{\psi} \gamma_0 \psi' - \bar{\psi}' \gamma_0 \psi \big) - e^{-\nu + \frac{\lambda}{2}} \big( \bar{\psi} \gamma_1 \dot{\psi} - \dot{\bar{\psi}}\gamma_1 \psi \big) + \\
    & +\frac{\dot{\lambda}}{4} e^{- \nu + \frac{\lambda}{2}}\big(\bar{\psi} \gamma_1 \psi + \bar{\psi} \gamma_0 \gamma_1 \gamma_0 \psi \big) + \frac{\nu'}{4} e^{-\frac{\nu}{2}} \big(\bar{\psi} \gamma_1 \gamma_0 \gamma_1 \psi + \bar{\psi} \gamma_0 \gamma_1 \gamma_1 \psi \big)  \Big] = \\
    = & \frac{i}{4} e^{-\frac{\nu}{2}}\Big[\bar{\psi} \gamma_0 \psi' - \bar{\psi}' \gamma_0 \psi + e^{\frac{\lambda - \nu}{2}} \big( \bar{\psi} \gamma_1 \dot{\psi} - \dot{\bar{\psi}}\gamma_1 \psi \big)\Big],
    \end{split}
    \label{T^0_r_Spher}
\end{equation}

\begin{equation}
    \begin{split}
    T^0_{\hphantom{0} \theta} = & -\frac{i}{4}\Big[\bar{\psi} \Gamma^0 D_{\theta} \psi + \frac{g^{00}}{g^{\theta \theta}}\bar{\psi} \Gamma^\theta D_{0} \psi - c.c. \Big] = \\
    = & -\frac{i}{4} \Big[-e^{-\frac{\nu}{2}}\big( \bar{\psi} \gamma_0 \partial_\theta \psi -\partial_\theta \bar{\psi} \gamma_0 \psi \big) - e^{-\nu}r \big( \bar{\psi}\gamma_2 \dot{\psi} - \dot{\bar{\psi}}\gamma_2 \psi \big) +\\
    & - \frac{e^{-\frac{\nu}{2}}}{2} \big( \bar{\psi}\gamma_0 \gamma_1 \gamma_2 \psi + \bar{\psi} \gamma_1 \gamma_2 \gamma_0 \psi \big) + \frac{e^{-\frac{\nu + \lambda}{2}} r \nu'}{4} \big( \bar{\psi}\gamma_2 \gamma_0 \gamma_1 \psi + \bar{\psi}\gamma_0 \gamma_1 \gamma_2 \psi \big) \Big] = \\
    = & \frac{i}{4} e^{-\frac{\nu}{2}} \Big[\bar{\psi} \gamma_0 \partial_\theta \psi -\partial_\theta \bar{\psi} \gamma_0 \psi + e^{-\frac{\nu}{2}}r \big( \bar{\psi}\gamma_2 \dot{\psi} - \dot{\bar{\psi}}\gamma_2 \psi \big) - \frac{1}{2} \big(e^{-\frac{\lambda}{2}} r \nu' - 2 \big) \bar{\psi}\gamma_0 \gamma_1 \gamma_2 \psi \Big],
    \end{split}
    \label{T^0_theta_Spher}
\end{equation}

\begin{equation}
    \begin{split}
    T^0_{\hphantom{0} \phi} = & -\frac{i}{4}\Big[\bar{\psi} \Gamma^0 D_{\phi} \psi + \frac{g^{00}}{g^{\phi \phi}}\bar{\psi} \Gamma^\phi D_{0} \psi - c.c. \Big] = \\
    = & -\frac{i}{4} \Big[-\frac{e^{-\frac{\nu}{2}}}{2}\big(\bar{\psi}\gamma_0\big[\sin{\theta} \gamma_1 \gamma_3 + r \cos{\theta} \gamma_2 \gamma_3 \big]\psi + \bar{\psi}\big[\sin{\theta} \gamma_1 \gamma_3 + r \cos{\theta} \gamma_2 \gamma_3 \big]\gamma_0\psi \big) +\\
    & - e^{-\nu}r \sin{\theta} \big( \bar{\psi}\gamma_3 \dot{\psi} - \dot{\bar{\psi}}\gamma_3 \psi \big) + \frac{e^{-\frac{\nu + \lambda}{2}}r\nu'}{4}  \sin{\theta} \big( \bar{\psi}\gamma_3 \gamma_0 \gamma_1 \psi + \bar{\psi}\gamma_0 \gamma_1 \gamma_3 \psi \big) \Big] = \\
    = & \frac{i}{4} e^{-\frac{\nu}{2}} \sin{\theta}\Big[r \cot{\theta} \hspace{0.5mm} \bar{\psi}\gamma_0 \gamma_2 \gamma_3 \psi + e^{-\frac{\nu}{2}}r \big( \bar{\psi}\gamma_3 \dot{\psi} - \dot{\bar{\psi}}\gamma_3 \psi \big) - \frac{1}{2} \big(e^{-\frac{\lambda}{2}}r\nu' - 2 \big) \bar{\psi}\gamma_0 \gamma_1 \gamma_3 \psi \Big].
    \end{split}
    \label{T^0_phi_Spher}
\end{equation}

\subsection{Space-space components}
In the end, we can compute the space-space components, which, like the time-space ones, are non-vanishing, in contrast to expectations.
\newline
More in detail, we find

\begin{equation}
    \begin{split}
    T^r_{\hphantom{r} \theta} = & -\frac{i}{4}\Big[\bar{\psi} \Gamma^r D_{\theta} \psi + \frac{g^{rr}}{g^{\theta \theta}}\bar{\psi} \Gamma^\theta D_{r} \psi - c.c. \Big] = \\
    = & -\frac{i}{4} \Big[e^{-\frac{\lambda}{2}}\big( \bar{\psi} \gamma_1 \partial_\theta \psi -\partial_\theta \bar{\psi} \gamma_1 \psi \big) + e^{-\lambda}r \big( \bar{\psi}\gamma_2 \psi' - \bar{\psi}'\gamma_2 \psi \big) +\\
    & + \frac{e^{-\frac{\lambda}{2}}}{2} \big( \bar{\psi}\gamma_1 \gamma_1 \gamma_2 \psi + \bar{\psi} \gamma_1 \gamma_2 \gamma_1 \psi \big) - \frac{e^{-\frac{\nu + \lambda}{2}} r \dot{\lambda}}{4} \big( \bar{\psi}\gamma_2 \gamma_0 \gamma_1 \psi + \bar{\psi}\gamma_0 \gamma_1 \gamma_2 \psi \big) \Big] = \\
    = & -\frac{i}{4} e^{-\frac{\lambda}{2}} \Big[ \bar{\psi} \gamma_1 \partial_\theta \psi -\partial_\theta \bar{\psi} \gamma_1 \psi + e^{-\frac{\lambda}{2}}r \big( \bar{\psi}\gamma_2 \psi' - \bar{\psi}'\gamma_2 \psi \big) - \frac{e^{-\frac{\nu}{2}} r \dot{\lambda}}{2} \bar{\psi}\gamma_0 \gamma_1 \gamma_2 \psi \Big],
    \end{split}
    \label{T^r_theta_Spher}
\end{equation}

\begin{equation}
    \begin{split}
    T^r_{\hphantom{r} \phi} = & -\frac{i}{4}\Big[\bar{\psi} \Gamma^r D_{\phi} \psi + \frac{g^{rr}}{g^{\phi \phi}}\bar{\psi} \Gamma^\phi D_{r} \psi - c.c. \Big] = \\
    = & -\frac{i}{4} \Big[\frac{e^{-\frac{\lambda}{2}}}{2}\big(\bar{\psi}\gamma_1\big[\sin{\theta} \gamma_1 \gamma_3 + r \cos{\theta} \gamma_2 \gamma_3 \big]\psi + \bar{\psi}\big[\sin{\theta} \gamma_1 \gamma_3 + r \cos{\theta} \gamma_2 \gamma_3 \big]\gamma_1\psi \big) +\\
    & + e^{-\lambda}r \sin{\theta} \big( \bar{\psi}\gamma_3 \psi' - \bar{\psi}'\gamma_3 \psi \big) - r \sin{\theta} \frac{e^{-\frac{\nu + \lambda}{2}}\dot{\lambda}}{4} \big( \bar{\psi}\gamma_3 \gamma_0 \gamma_1 \psi + \bar{\psi}\gamma_0 \gamma_1 \gamma_3 \psi \big) \Big] = \\
    = & -\frac{i}{4} e^{-\frac{\lambda}{2}}r \sin{\theta} \Big[ e^{-\frac{\lambda}{2}}\big( \bar{\psi}\gamma_3 \psi' - \bar{\psi}'\gamma_3 \psi \big) + \cot{\theta} \hspace{0.5mm} \bar{\psi}\gamma_1\gamma_2 \gamma_3 \psi - \frac{e^{-\frac{\nu}{2}}\dot{\lambda}}{2}\bar{\psi}\gamma_0 \gamma_1 \gamma_3 \psi \Big],
    \end{split}
    \label{T^r_phi_Spher}
\end{equation}

\begin{equation}
    \begin{split}
    T^\theta_{\hphantom{\theta} \phi} = & -\frac{i}{4}\Big[\bar{\psi} \Gamma^\theta D_{\phi} \psi + \frac{g^{\theta \theta}}{g^{\phi \phi}}\bar{\psi} \Gamma^\phi D_{\theta} \psi - c.c. \Big] = \\
    = & -\frac{i}{4} \Big[\frac{1}{2r} \Big( \bar{\psi} \gamma_2 \big[\sin{\theta} \gamma_1 \gamma_3 + r \cos{\theta} \gamma_2 \gamma_3 \big]\psi + \bar{\psi} \big[\sin{\theta} \gamma_1 \gamma_3 + r \cos{\theta} \gamma_2 \gamma_3 \big]\gamma_2 \psi \Big) +\\
    & + \frac{\sin{\theta}}{r} \big( \bar{\psi}\gamma_3 \partial_\theta \psi - \partial_\theta \bar{\psi} \gamma_3 \psi \big) + \frac{\sin{\theta}}{2r} \big( \bar{\psi}\gamma_3 \gamma_1 \gamma_2 \psi + \bar{\psi} \gamma_1 \gamma_2 \gamma_3 \psi \big) \Big] = \\
    = & -\frac{i}{4} \Big[-\frac{\sin{\theta}}{r} \bar{\psi} \gamma_1 \gamma_2 \gamma_3 + \frac{\sin{\theta}}{r} \big( \bar{\psi}\gamma_3 \partial_\theta \psi - \partial_\theta \bar{\psi} \gamma_3 \psi \big) + \frac{\sin{\theta}}{r}\bar{\psi} \gamma_1 \gamma_2 \gamma_3 \psi \Big] = \\
    = & -\frac{i\sin{\theta}}{4r}\big( \bar{\psi}\gamma_3 \partial_\theta \psi - \partial_\theta \bar{\psi} \gamma_3 \psi \big).
    \end{split}
    \label{T^theta_phi_Spher}
\end{equation}
When we solve the Einstein equations, these expressions and the last two expressions of the time-space components will play the role of constraints on the spinorial field under analysis, since the non-vanishing components of the Einstein tensor, in the spherically symmetric case, are only the diagonal ones and the $0r$-one. 

\renewcommand{\chaptername}{Acknowledgements}

\chapter*{Acknowledgements}
Prima di tutto, vorrei ringraziare il Prof. Piattella per avermi guidato in questo lavoro di tesi. Lo ringrazio per avermi nuovamente supportato in questi ultimi mesi. Poi, lo ringrazio per la sua continua disponibilità, per la sua competenza e per i consigli fornitimi per quanto riguarda le scelte future.
\newline
\newline
In seguito, vorrei ringraziare i miei compagni di corso, che mi hanno accompagnato anche in questo percorso magistrale e con i quali ho passato dei bellissimi anni. In particolare, li vorrei ringraziare per le risate e la spensieratezza datemi, ma anche per il continuo confronto in ambito fisico, che mi ha permesso di proseguire al meglio in questo percorso.
\newline
\newline
Successivamente, vorrei ringraziare i miei amici, che in questi anni sono stati sempre accanto a me; sia nei momenti belli che in quelli brutti. Mi hanno supportato nei momenti più brutti e mi hanno fatto ridere e divertire sempre e comunque. Li ringrazio anche per le belle esperienze vissute insieme a loro, sperando di poterne rivivere altre.
\newline
\newline
Poi, ringrazio le mie nonne, i miei nonni e i miei zii. In questi venticinque anni, mi avete dato tanto affetto e serenità ed io spero di aver dato tanto a voi. Con alcuni avrei voluto passare un po' più di tempo, ma non ce n'è stata possibilità. Spero, però, ovunque voi siate e sarete, di avervi reso felici e orgogliosi di me.
\newline
\newline
Infine, per penultimi e più importanti, ringrazio mia mamma e mio papà, perché senza di loro non sarei riuscito a sostenere tutto ciò. Mi sono sempre stati vicini e so che lo saranno sempre, in un modo o nell'altro, ed io non smetterò mai di ringraziarli per tutto questo. Li ringrazio di cuore per avermi sostenuto in ogni mia scelta senza dubitare mai di me.
\newline
\newline
Per ultimo, ringrazio me stesso. Grazie di non aver mai mollato e di aver sempre affrontato tutte le sfide con serietà, ma anche con un pizzico di irriverenza. Spero che tu riesca a mettere la stessa grinta e passione anche in tutto quello che ti capiterà in futuro e in tutte le sfide che dovrai affrontare.

\bibliographystyle{apalike}
\bibliography{Cit}

@article{Einstein_Electr,
    author = "Einstein, A.",
    title = {{Zur Elektrodynamik bewegter K{\"o}rper}},
    doi = "10.1002/andp.19053221004",
    journal = "Annalen der Physik",
    volume = {322},
    number = {10},
    pages = {891-921},
    year = "1905"
}

@article{Einstein_GenRel,
        author = {{Einstein}, A.},
        title = "{Die Grundlage der allgemeinen Relativit{\"a}tstheorie}",
        journal = {Annalen der Physik},
        year = "1916",
        month = jan,
        volume = {354},
        number = {7},
        pages = {769-822},
        doi = {10.1002/andp.19163540702},
        adsurl = {https://ui.adsabs.harvard.edu/abs/1916AnP...354..769E}
}

@article{Schwarzschild,
        author = {{Schwarzschild}, Karl},
        title = "{{\"U}ber das Gravitationsfeld eines Massenpunktes nach der Einsteinschen Theorie}",
        journal = {Sitzungsberichte der K{\"o}niglich Preussischen Akademie der Wissenschaften},
        year = 1916,
        month = jan,
        pages = {189-196},
        adsurl = {https://ui.adsabs.harvard.edu/abs/1916SPAW.......189S}
}

@article{Einstein_Cosmo,
        author = {{Einstein}, Albert},
        title = "{Kosmologische Betrachtungen zur allgemeinen Relativit{\"a}tstheorie}",
        journal = {Sitzungsberichte der K{\"o}niglich Preussischen Akademie der Wissenschaften},
        year = 1917,
        month = jan,
        pages = {142-152},
        adsurl = {https://ui.adsabs.harvard.edu/abs/1917SPAW.......142E}
}

@inproceedings{HistLambda,
    author = "Straumann, Norbert",
    title = "{The History of the cosmological constant problem}",
    booktitle = "{18th IAP Colloquium on the Nature of Dark Energy: Observational and Theoretical Results on the Accelerating Universe}",
    eprint = "gr-qc/0208027",
    archivePrefix = "arXiv",
    month = "8",
    year = "2002"
}

@article{Magueijo,
  title = {Cosmology with a spin},
  author = {Magueijo, J. and Zlosnik, T. G. and Kibble, T. W. B.},
  journal = {Phys. Rev. D},
  volume = {87},
  issue = {6},
  pages = {063504},
  numpages = {13},
  year = {2013},
  month = {Mar},
  publisher = {American Physical Society},
  doi = {10.1103/PhysRevD.87.063504},
  url = {https://link.aps.org/doi/10.1103/PhysRevD.87.063504}
}

@article{Fabbri,
    author = "Fabbri, Luca and Vignolo, Stefano and De Maria, Giuseppe and Carloni, Sante",
    title = "{Dirac Fields in Hydrodynamic Form and their Thermodynamic Formulation}",
    eprint = "2506.02608",
    archivePrefix = "arXiv",
    primaryClass = "math-ph",
    month = "6",
    year = "2025"
}

@book{Landau,
    author = "Landau, L. D. and Lifschits, E. M.",
    title = "{The Classical Theory of Fields}",
    doi = "10.1016/c2009-0-14608-1",
    isbn = "978-0-08-018176-9",
    publisher = "Pergamon Press",
    address = "Oxford",
    series = "Course of Theoretical Physics",
    volume = "Volume 2",
    year = "1975"
}

@book{Wheeler,
    author = "Wheeler, J. A. and Ford, K.",
    title = "{Geons, black holes, and quantum foam: A life in physics}",
    publisher = "W. W. Norton and Company",
    year = "1998"
}

@book{Friedrich,
    author = "Friedrich, Thomas and Nestke, Andreas",
    title = "{Dirac Operators in Riemannian Geometry}",
    isbn = "978-0-8218-2055-1",
    publisher = "American Mathematical Society",
    address = "Providence, R.I.",
    series = "Graduate Studies in Mathematics",
    volume = "25",
    year = "2000"
}

@book{Lawson,
    author = "Lawson, H. B. and Michelsohn, M. L.",
    title = "{Spin geometry}",
    publisher = "Princeton University Press",
    year = "1998"
}

@article{Dirac,
    author = {Dirac, Paul Adrien Maurice},
    title = {The quantum theory of the electron},
    journal = {Proceedings of the Royal Society of London. Series A, Containing Papers of a Mathematical and Physical Character},
    volume = {117},
    number = {778},
    pages = {610-624},
    year = {1928},
    month = {02},
    issn = {0950-1207},
    doi = {10.1098/rspa.1928.0023},
    url = {https://doi.org/10.1098/rspa.1928.0023},
    eprint = {https://royalsocietypublishing.org/rspa/article-pdf/117/778/610/25048/rspa.1928.0023.pdf},
}

@article{Holst,
    author = "Holst, Soren",
    title = "{Barbero's Hamiltonian derived from a generalized Hilbert-Palatini action}",
    eprint = "gr-qc/9511026",
    archivePrefix = "arXiv",
    reportNumber = "USITP-95-10",
    doi = "10.1103/PhysRevD.53.5966",
    journal = "Phys. Rev. D",
    volume = "53",
    pages = "5966--5969",
    year = "1996"
}

@article{Dolan,
    author = "Dolan, Brian P.",
    title = "{Chiral fermions and torsion in the early Universe}",
    eprint = "0911.1636",
    archivePrefix = "arXiv",
    primaryClass = "gr-qc",
    reportNumber = "DIAS-STP-09-11",
    doi = "10.1088/0264-9381/27/9/095010",
    journal = "Class. Quant. Grav.",
    volume = "27",
    pages = "095010",
    year = "2010",
    note = "[Erratum: Class.Quant.Grav. 27, 249801 (2010)]"
}

@article{Picon,
        author = {{Armend{\'a}riz-Pic{\'o}n}, C. and {Greene}, Patrick B.},
        title = "{Spinors, Inflation, and Non-Singular Cyclic Cosmologies}",
        journal = {General Relativity and Gravitation},
        year = 2003,
        month = sep,
        volume = {35},
        number = {9},
        pages = {1637-1658},
        doi = {10.1023/A:1025783118888},
        archivePrefix = {arXiv},
        eprint = {hep-th/0301129},
        primaryClass = {hep-th},
        adsurl = {https://ui.adsabs.harvard.edu/abs/2003GReGr..35.1637A}
}

@article{Isham,
    author = "Isham, C. J. and Nelson, J. E.",
    title = "{Quantization of a Coupled Fermi Field and Robertson-Walker Metric}",
    reportNumber = "Print-74-1382 (KING'S COLL.)",
    doi = "10.1103/PhysRevD.10.3226",
    journal = "Phys. Rev. D",
    volume = "10",
    pages = "3226",
    year = "1974"
}

@article{Planck,
    author = "Aghanim, N. and others",
    collaboration = "Planck",
    title = "{Planck 2018 results. VI. Cosmological parameters}",
    eprint = "1807.06209",
    archivePrefix = "arXiv",
    primaryClass = "astro-ph.CO",
    doi = "10.1051/0004-6361/201833910",
    journal = "Astron. Astrophys.",
    volume = "641",
    pages = "A6",
    year = "2020",
    note = "[Erratum: Astron.Astrophys. 652, C4 (2021)]"
}

@book{Piattella,
    author = "Piattella, Oliver F.",
    title = "{Lecture Notes in Cosmology}",
    eprint = "1803.00070",
    archivePrefix = "arXiv",
    primaryClass = "astro-ph.CO",
    doi = "10.1007/978-3-319-95570-4",
    isbn = "978-3-319-95569-8, 978-3-030-07060-1, 978-3-319-95570-4",
    publisher = "Springer",
    address = "Cham",
    series = "UNITEXT for Physics",
    year = "2018"
}

@article{Farnsworth,
    author = "Farnsworth, Shane and Lehners, Jean-Luc and Qiu, Taotao",
    title = "{Spinor driven cosmic bounces and their cosmological perturbations}",
    doi = "10.1103/PhysRevD.96.083530",
    journal = "Phys. Rev. D",
    volume = "96",
    number = "8",
    pages = "083530",
    year = "2017"
}

@inproceedings{Maartens,
    author = "Maartens, Roy",
    title = "{Causal thermodynamics in relativity}",
    eprint = "astro-ph/9609119",
    archivePrefix = "arXiv",
    reportNumber = "PU-RCG-96-14",
    month = "9",
    year = "1996"
}

@book{dInverno,
    author = "d'Inverno, R.",
    title = "{Introducing Einstein's relativity}",
    isbn = "978-0-19-859686-8",
    publisher = "Oxford University Press",
    year = "1992"
}

@article{Shapiro,
    author = "Shapiro, Ilya L.",
    title = "{Covariant Derivative of Fermions and All That}",
    eprint = "1611.02263",
    archivePrefix = "arXiv",
    primaryClass = "gr-qc",
    doi = "10.3390/universe8110586",
    journal = "Universe",
    volume = "8",
    number = "11",
    pages = "586",
    year = "2022"
}

@article{Saha,
    author = "Saha, Bijan",
    title = "{Spinor Field in FLRW Cosmology}",
    doi = "10.3390/universe9050243",
    journal = "Universe",
    volume = "9",
    number = "5",
    pages = "243",
    year = "2023"
}

@article{Chimento,
    author = "Chimento, L. P. and Devecchi, F. P. and Forte, M. and Kremer, G. M. and Ribas, M. O. and Samojeden, L. L.",
    editor = "Diosi, Lajos and Elze, Hans-Thomas and Fronzoni, Leone and Halliwell, Jonathan and Prati, Enrico and Vitiello, Giuseppe and Yearsley, James",
    title = "{Fermionic cosmologies}",
    doi = "10.1088/1742-6596/306/1/012052",
    journal = "J. Phys. Conf. Ser.",
    volume = "306",
    pages = "012052",
    year = "2011"
}

@article{Jantzen,
    author = "Jantzen, R. T.",
    title = "{Spinor sources in cosmology}",
    doi = "10.1063/1.525481",
    journal = "J. Math. Phys.",
    volume = "23",
    pages = "1137--1146",
    year = "1982"
}

@article{Ribas,
    author = "Ribas, M. O. and Devecchi, F. P. and Kremer, G. M.",
    title = "{Fermions as sources of accelerated regimes in cosmology}",
    eprint = "gr-qc/0511099",
    archivePrefix = "arXiv",
    doi = "10.1103/PhysRevD.72.123502",
    journal = "Phys. Rev. D",
    volume = "72",
    pages = "123502",
    year = "2005"
}

@article{Villalba,
    author = "Villalba, V. M. and Percoco, U.",
    title = "{Separation of Variables and Exact Solution to Dirac and Weyl Equations in {Robertson-Walker} Space-times}",
    reportNumber = "IVIC-CFLE-89-05",
    doi = "10.1063/1.528799",
    journal = "J. Math. Phys.",
    volume = "31",
    pages = "715",
    year = "1990"
}

@article{Szekeres,
  title={Cosmology with spinor connection},
  author={Szekeres, G},
  journal={Computers \& Mathematics with Applications},
  volume={1},
  number={3-4},
  pages={345--350},
  year={1975},
  publisher={Elsevier}
}

\end{document}